\documentclass[11pt]{article}
\usepackage{amsmath,amssymb,amsthm}
\usepackage{graphicx,import,pdfpages}
\usepackage{tikz}
\usepackage{multicol}
\usepackage{tikz-cd}
\usetikzlibrary{arrows}
\usepackage[normalem]{ulem}
\usepackage{mathtools}
\usepackage{comment}
\usepackage{color} 
\usepackage[linktoc=page]{hyperref}
\usepackage[shortlabels]{enumitem}

\usepackage{fontawesome5}
\usepackage{academicons}

\usepackage{soul}

\definecolor{grey}{rgb}{0.5, 0.5, 0.5}
\definecolor{textgreen}{rgb}{0, 0.5, 0.2}
\definecolor{highlightgreen}{rgb}{0.6, 0.95, 0.6}

\usepackage[nameinlink,capitalize]{cleveref}

\crefname{equation}{}{}

\makeatletter
\DeclareFontEncoding{LS1}{}{}
\DeclareFontSubstitution{LS1}{stix}{m}{n}
\DeclareMathAlphabet{\mathcal}{LS1}{stixscr}{m}{n}
\makeatother

\usepackage[calc,showdow,en-GB,showseconds=false]{datetime2}
\DTMsettimestyle{iso}

\usepackage{tabularx}
\newcolumntype{M}[1]{>{\raggedright}m{#1}}
\newcolumntype{L}[1]{>{\raggedright\let\newline\\\arraybackslash\hspace{0pt}}m{#1}}
\newcolumntype{C}[1]{>{\centering\let\newline\\\arraybackslash\hspace{0pt}}m{#1}}
\newcolumntype{R}[1]{>{\raggedleft\let\newline\\\arraybackslash\hspace{0pt}}m{#1}}

\newcolumntype{a}{>{$}c<{$}}

\let\epsilon\varepsilon
\let\phi\varphi
\let\emptyset\varnothing

\let\hat\widehat
\let\tilde\widetilde

\newcommand{\ZZ}{\mathbb{Z}}
\newcommand{\CC}{\mathbb{C}}

\renewcommand{\AA}{\mathbb{A}}

\newcommand{\PP}{\mathbb{P}}

\def\co{\colon\thinspace}
\usepackage{colonequals}
\def\coeq{\colonequals}
\def\eqco{\equalscolon}

\let\div\Div

\DeclareMathOperator{\Id}{Id}

\newcommand{\pp}[1]{\frac{\partial}{\partial{#1}}}

\def\pdz{\frac{\partial}{\partial z}}

\DeclareMathOperator{\Spec}{Spec}

\DeclareMathOperator{\End}{End}
\DeclareMathOperator{\Aut}{Aut}
\DeclareMathOperator{\GL}{GL}
\DeclareMathOperator{\Gr}{Gr}

\DeclareMathOperator{\id}{Id}

\DeclareMathOperator{\Hom}{Hom}

\DeclareMathOperator{\intHom}{\underline{Hom}}
\DeclareMathOperator{\intEnd}{\underline{End}}
\DeclareMathOperator{\intAut}{\underline{Aut}}

\DeclareMathOperator{\coker}{coker}

\DeclareMathOperator{\tr}{tr}
\DeclareMathOperator{\str}{str}

\DeclareMathOperator{\sym}{sym_1}

\DeclareMathOperator*{\Ber}{Ber}

\DeclareMathOperator{\red}{red}

\newtheorem{thm}{Theorem}[section]
\newtheorem{lem}[thm]{Lemma}
\newtheorem{prop}[thm]{Proposition}
\newtheorem{cor}[thm]{Corollary}
\newtheorem{clm}[thm]{Claim}

\theoremstyle{definition}
\newtheorem{defn}[thm]{Definition}
\theoremstyle{definition}

\theoremstyle{definition}

\theoremstyle{remark}

\theoremstyle{remark}

\newtheorem*{thm*}{Theorem}
\newtheorem*{lem*}{Lemma}
\newtheorem*{prop*}{Proposition}
\newtheorem*{cor*}{Corollary}
\newtheorem*{clm*}{Claim}
\newtheorem*{cnj*}{Conjecture}
\newtheorem*{ass*}{Assertion}

\theoremstyle{definition}
\newtheorem*{defn*}{Definition}
\theoremstyle{definition}
\newtheorem*{exmp*}{Example}
\theoremstyle{definition}
\newtheorem*{exr*}{Exercise}

\theoremstyle{remark}
\newtheorem*{rem*}{Remark}
\theoremstyle{remark}
\newtheorem*{goal*}{Goal}

\usepackage[margin=1in]{geometry}

\usepackage[noblocks]{authblk}

\usepackage[T2A]{fontenc}
\usepackage{CJKutf8}

\usepackage{mathtools}

\hypersetup{colorlinks=true, linkcolor=blue!70!black, urlcolor=blue!70!black, citecolor=blue!70!black,linktoc=page}

\usepackage[style=trad-alpha,sorting=nyt,backend=biber,giveninits=true,maxbibnames=99,isbn=false,maxalphanames=6]{biblatex}
\AtEveryBibitem{\clearlist{language}}

\renewbibmacro*{doi+eprint+url}{%
	\iftoggle{bbx:doi}
	{\printfield{doi}}
	{}%
	\newunit\newblock
	\iftoggle{bbx:eprint}
	{\usebibmacro{eprint}}
	{}%
	\newunit\newblock
	\iftoggle{bbx:url}
	{\usebibmacro{url+urldate}}
	{}.%
}

\DeclareFieldFormat{url}{url\addcolon\url{#1}}
\DeclareFieldFormat{doi}{%
	doi\addcolon
	\ifhyperref
	{\href{https://doi.org/#1}{\nolinkurl{#1}}}
	{\nolinkurl{#1}}}
\DeclareFieldFormat{eprint:arxiv}{%
	arXiv\addcolon
	\ifhyperref
	{\href{https://arxiv.org/abs/#1}{%
			\nolinkurl{#1}%
		}}
	{\nolinkurl{#1}%
	}}

\title{The super Mumford form and Sato Grassmannian}
\author{Katherine A. Maxwell \footnote{\noindent{\href{mailto:maxwe204@umn.edu}{maxwe204@umn.edu}\\\indent\indent Max Plank Institute for Mathematics, 53111 Bonn, Germany.}}}
\date{January 16, 2022}

\begin{document}
	\begin{CJK}{UTF8}{min}

\maketitle
\begin{abstract}
We describe a supersymmetric generalization of the construction of
Kontsevich and Arbarello, De Concini, Kac, and Procesi, which utilizes a relation between the
moduli space of curves with the infinite-dimensional Sato Grassmannian. Our main result is the existence of a flat holomorphic connection on the line bundle $\lambda_{3/2}\otimes\lambda_{1/2}^{-5}$ on the moduli space of triples: a super Riemann surface, a Neveu-Schwarz puncture, and a formal coordinate system. We also prove a superconformal Noether normalization lemma for families of super Riemann surfaces.
\end{abstract}  

\tableofcontents

\section{Introduction}

The constant $c_j=6j^2-6j+1$, proportional to the 2\textsuperscript{nd} Bernoulli polynomial, arises independently in studying representations of the Virasoro algebra and by studying the geometry of the moduli space of curves. Manin \cite{Manin--1986} conjectured based on this numerical coincidence that there exists a direct connection between these seemingly independent mathematical areas. This conjecture was shown to be true in the simultaneous publications \cite{Arbarello-DeConcini-Kac-Procesi--1988,Beilinson-Shechtman--1988,Kawamoto-Namikawa-Tsuchiya-Yamada--1988,Kontsevich--1987e}, which was important to the quantum theory of (bosonic) strings, as it unified path integral quantization and operator quantization.

The numerical coincidence goes as follows. The Witt algebra has natural representations $\varrho_j$ defined by Lie derivative action on $\mathbb{C}(\!(z)\!)dz^{\otimes j}$. The Japanese cocycle $\eta$ on $\mathfrak{gl}(\mathbb{C}(\!(z)\!))$ by restriction to the Witt algebra induces the unique central extension of the Witt algebra, known as the Virasoro algebra. Then we find
\begin{align*}
	\varrho_j^*(\eta)=c_j\varrho_1^*(\eta), && c_j=6j^2-6j+1.
\end{align*}
When the Virasoro algebra is realized as operators on the state space of a string, the string must propagate in in $2c_j$ (real) spacetime dimensions for both unitarity and Lorentz invariance of the string theory to hold. 
On the other hand, for any proper family of curves $\pi\co X\to S$, the Mumford isomorphism is the isomorphism of line bundles 
\begin{align*}
	\lambda_j\cong \lambda_1^{c_j},
\end{align*}
where $\lambda_j\coeq \det R \pi_*(\omega_{X/S}^{\otimes j})$ and $\omega_{X/S}\coeq \Omega_{X/S}^1$. If $S$ is the moduli space of curves and $\pi\co X \to S$ is the universal curve, 
then we find the dualizing sheaf $\omega_S\coeq\det\Omega^1_S\cong\lambda_2$, and therefore $\omega_S\cong \lambda_1^{c_j}\cong(\pi_*\omega_{X/S})^{c_j}$. This is important to Polyakov path integration on the moduli space.

An alternative description of the Polyakov measure comes through the connection between representations of the Virasoro algebra and the moduli space; this connection is roughly the following.
The Virasoro algebra acts on the moduli space $\mathcal{M}_{g,1^{\infty}}$ of triples $(C,p,z)$, a Riemann surface, an point, and a parameter. This action can be used to show that there exists a flat holomorphic connection on the line bundle $\lambda_j\otimes \lambda_1^{-c_j}$, which can be used to write differential equations for the Polyakov measure. For another summary, see the introduction of \cite{Manin--1988}.

In fact, \textcite{Manin--1986} hypothesized this relationship existed in the super setting as well. The numerical coincidence is analogous: The super Witt algebra has representations $\varrho_j$ defined by Lie derivative action on $\mathbb{C}(\!(z)\!)[\zeta]\,[dz|d\zeta]^{\otimes j}$ such that pulling back the super Japanese cocycle $\eta$ (which defines the unique central extension: the Neveu-Schwarz superalgebra) gives
\begin{align*}
	\varrho_j^*(\eta)=c_j\varrho_1^*(\eta), && c_j=-(-1)^j(2j-1).
\end{align*}
And for any proper family of supercurves $\pi\co X\to S$, the super Mumford isomorphism is the isomorphism of line bundles 
\begin{align*}
	\lambda_{j/2}\cong \lambda_{1/2}^{c_j},
\end{align*}
where $\lambda_{j/2}\coeq \Ber R\pi_*(\omega^{\otimes j}_{X/S})$ and $\omega_{X/S}=\Ber\Omega_{X/S}^1$.

The goal of the paper is to describe a supersymmetric generalization of the construction of \textcite{Kontsevich--1987e} and \textcite{Arbarello-DeConcini-Kac-Procesi--1988}, which utilizes a relation between the moduli space of curves with the infinite-dimensional Grassmannian. The consequences of their construction shows that the Chern classes of the line bundles of the Mumford isomorphism are equal. Manin proved a supersymmetric generalization of this result in \cite[Theorem 3.3]{Manin--1988} by generalizing
the methods of \textcite{Beilinson-Shechtman--1988}, which prove a version of the Riemann-Roch theorem for the Atiyah algebras of vector bundles. Instead, our paper extends the results of \textcite{Ueno-Yamada--1988} of the representations of the Neveu-Schwarz superalgebra and the results of \textcite{Mulase-Rabin--1991} and \textcite{AlvarezVazquez-MunozPorras-PlazaMartin--1998} of the super Sato Grassmannian and the super Krichever map. Our method depends on showing that the Lie superalgebroid of superconformal vector fields on a family of open super Riemann surfaces is perfect, which follows from a superconformal Noether normalization lemma for families of super Riemann surfaces.

Our use of the super Grassmannian $\Gr\big(\mathbb{C}(\!(z)\!)[\zeta]\,[dz|d\zeta]^{\otimes j}\big)$ appears to provide an alternative approach for integrating over the moduli space of SRSs $\mathfrak{M}_g$. The Torelli map sending a super Riemann surface to its Jacobian $J\co\mathfrak{M}_g \to \mathfrak{A}_g$, where $\mathfrak{A}_g$ is the moduli space of principally polarized abelian supervarieties, plays a prominent role in the papers \cite{DHoker-Phong--2002i,DHoker-Phong--2002ii,DHoker-Phong--2002iii,DHoker-Phong--2002iv} of D'Hoker and Phong and \cite{Grushevsky--2009} of Grushevsky, who compute explicitly or propose an ansatz for chiral superstring measure in low genus. 
In higher genus, these methods face the problem that the locus of moduli space inside $\mathfrak{A}_g$, known as the super Jacobian locus, is very hard to describe. The known characterizations of the super Jacobian locus, i.e. solutions to the super Schottky problem, are very implicit. Recently, the super Schottky problem has been studied in \cite{Felder-Kazhdan-Polishchuk--2019X} and \cite{Codogni-Viviani--2019}. However, since the known solution of the super Schottky problem, as in \textcite{Mulase--1991}, goes through the super Krichever map $\mathfrak{M}_g \to \Gr(\mathbb{C}(\!(z)\!)[\zeta]\,[dz|d\zeta]^{\otimes j})$, and since a more explicit description of the moduli space locus is an orbit of the super Witt algebra under the super Krichever map, it appears that the computation of physically meaningful integrals, such as scattering amplitudes, can be carried out in the super Grassmannian, rather than on the moduli space $\mathfrak{A}_g$. This idea does not seem to have been utilized in the classical (i.e. non-super) case. 

We note that the super Mumford isomorphism for a compactification of $\mathfrak{M}_g$ has recently been studied in the important paper \cite{Felder-Kazhdan-Polishchuk--2020X}.

In \cref{Lie algebroids}, we recall details of Lie superalgebroids, which are used primarily in \cref{flat section}. In \cref{moduli section}, we define the moduli space of triples $\mathfrak{M}_{g,1^\infty_\text{NS}}$: a super Riemann surface, a Neveu-Schwarz puncture, and a formal coordinate system. The main new results are a superconformal Noether normalization over a family in \cref{superconformalNoether}, \cref{commutant}, and the action of the Witt superalgebra by vector fields in \cref{switt action}. In \cref{Gr section}, we define the super Sato Grassmannian $\Gr\big(\mathbb{C}(\!(z)\!)[\zeta]\,[dz|d\zeta]^{\otimes j}\big)$; the main result is the action of the central extension of the general linear superalgebra on the Berezinian line bundle in \cref{diff operators on det}. Finally, the main results of this paper are found in \cref{flat section}, where we state our definition of the super Krichever map in \cref{super Krichever}, find the action of the Neveu-Schwarz superalgebra in \cref{ns action}, and find the existence of a flat holomorphic connection in \cref{flat}. Finally, \cref{classical} contains no new results, but is meant to aid the reader by restating the previous results of the classical setting.

\section{Preliminaries on Lie superalgebroids}\label{Lie algebroids}

We provide the needed background information about Lie superalgebroids and specifically Atiyah superalgebras. Much of the information in this section is based on \cite{KosmannSchwarzbach-Mackenzie--2002}. Lie superalgebroids are also covered in  \cite[Section 3.6]{Manin--1988}.

\subsection{Action Lie superalgebroids}

We recall the definition of a Lie superalgebra. The difference from a classical Lie algebra is the supercommutator, in other words the rule of signs. 
\begin{defn}
	A \emph{Lie superalgebra} is a super vector space $\mathfrak{g}$ with a morphism $[\;,\;]\co\mathfrak{g}\otimes\mathfrak{g}\to\mathfrak{g}$ satisfying
	\begin{itemize}
		\item  $[x,y]+(-1)^{|x||y|}[y,x]=0$
		\item the super Jacobi identity $(-1)^{|x||z|}[x,[y,z]]+(-1)^{|y||x|}[y,[z,x]]+(-1)^{|z||y|}[z,[x,y]]=0$
	\end{itemize}
\end{defn}
In what follows, one may easily determine the classical statements by omitting the sign rule. We also cite sources for the classical versions of these statements.

 A Lie superalgebroid is simply the many-object generalization of a Lie superalgebra. That is, a Lie superalgebroid over a point is a Lie superalgebra. 
 
 We prefer to work with locally free sheaves, and so denote Lie superalgebroids in script as $\mathcal{E}$. If we work with instead the associated super vector bundle, we denote the Lie superalgebroid in roman as $E$. Note that $TS$ denotes the tangent bundle and $\mathcal{T}_S$ the tangent sheaf of $S$.
 
\begin{defn}[cf. {\cite[Definition 1.3]{KosmannSchwarzbach-Mackenzie--2002}}]
	A Lie superalgebroid on a complex supermanifold $S$ is a triple $(\mathcal{E},a,[\;,\;])$ such that
	\begin{itemize}
		\item $\mathcal{E}$ is the sheaf of sections of $E\to S$ a super vector bundle 
		\item $a\co E\to TS$ is a (parity preserving) super vector bundle morphism, called the anchor map, where $TS$ is the tangent bundle of $S$
		\item $[\; ,\; ] : \Gamma(S,\mathcal{E}) \times \Gamma(S,\mathcal{E}) \to \Gamma(S,\mathcal{E})$ is a $\mathbb{C}$-bilinear alternating bracket which satisfies for all $X,Y\in\Gamma(S,\mathcal{E})$ and $f\in\Gamma(S,\mathcal{O}_S)$
		\subitem $\bullet$\; the super Jacobi identity
		\subitem $\bullet$\; $a([X, Y ]) = [a(X), a(Y )]$
		\subitem $\bullet$\; $[X, f Y ] = (-1)^{|f||X|}f [X, Y ] + a(X)(f )Y$
	\end{itemize}

	Given Lie superalgebroids $(\mathcal{E},a)$ and  $(\mathcal{F},b)$ on the same base $S$, a Lie superalgebroid morphism $\phi$ is a (parity preserving) super vector bundle morphism $\phi \co E \to F$ over $S$ such that for all $X, Y \in \Gamma(S,\mathcal{E})$
	\begin{itemize}
		\item $b\circ \phi = a$
		\item $\phi([X, Y ]) = [\phi(X), \phi(Y )]$
	\end{itemize}
\end{defn}

In this paper we will need use of two specific types of Lie superalgebroids: action Lie superalgebroids and Atiyah superalgebras.

\begin{defn}[cf. {\cite[last paragraph page 8]{KosmannSchwarzbach-Mackenzie--2002}}]\label{action superalgebroid def}
	The action Lie superalgebroid 
	$(\mathcal{G},a,[\;,\;])$ 
	associated to a Lie superalgebra morphism $\phi\co\mathfrak{g}\to\Gamma(S,TS)$ 
	is given by 
	\begin{itemize}
		\item $\mathcal{G}$ is the sheaf of sections of $S\times\mathfrak{g}\to S$, the trivial super vector bundle with fiber $\mathfrak{g}$
		\item anchor map $a\co S\times \mathfrak{g}\to TS$ defined by $a(s,X)\coeq\phi(X)(s) $
		\item bracket for all $V,W\in\Gamma(S,S\times\mathfrak{g})$ defined by 
		\begin{align}
			[V,W]\coeq\mathcal{L}_{\hat{\phi}(V)}(W)-(-1)^{|V||W|}\mathcal{L}_{\hat{\phi}(W)}(V)+[V,W]_{\mathfrak{g}}
		\end{align} 
	where $\hat{\phi}\co \Gamma(S,S\times\mathfrak{g})\to\Gamma(S,TS)$ is defined by $\hat{\phi}(V)(s)\coeq\phi(V(s))(s)$ and $\mathcal{L}$ is the Lie derivative.
	\end{itemize} 
\end{defn}
To compare the map $\phi$ and $\hat{\phi}$, the map $\phi$ assigns a global vector field on $S$ to each $X\in\mathfrak{g}$, however, the map $\hat{\phi}$ allows for the $X\in\mathfrak{g}$ to vary over $S$ before mapping to the corresponding vectors on $S$.

If $V,W\in\Gamma(S,S\times\mathfrak{g})$ are constant sections corresponding to $X,Y\in\mathfrak{g}$ respectively, then $[V,W]=[X,Y]_{\mathfrak{g}}$, the constant section of $S\times\mathfrak{g}$ associated to $[X,Y]\in\mathfrak{g}$. This property along with definition of the anchor map are enough to characterize the action Lie superalgebroid.

\subsection{Atiyah Lie superalgebroids}\label{Atiyah section}

The Atiyah Lie superalgebroid $\mathcal{A}_L$ of a line bundle $L$ on a supermanifold
$S$ may be regarded as the Lie superalgebroid of infinitesimal symmetries of the pair $(S,L)$. For historical reasons, what we call the Atiyah Lie superalgebroid may be called the Atiyah superalgebra in other literature.

First, recall that a first order differential operator on a super vector bundle $E$ is a map $D\co \Gamma(U,E)\to\Gamma(U,E)$ such that 
\begin{align*}
	[D,f](e)=D(fe)-(-1)^{|D||f|}f D(e) && [[D,f],g]=0
\end{align*}
for all $e\in\Gamma(U,E)$ and functions $f,g$.

\begin{defn}[cf. {\cite[Theorem 1.4]{KosmannSchwarzbach-Mackenzie--2002}}] \label{Atiyah line def}
	Let $L$ be a line bundle on a complex supermanifold $S$, that is a rank $1|0$ or $0|1$ vector bundle. 
	
	The Atiyah Lie superalgebroid $\mathcal{A}_L$ is the Lie superalgebroid on $S$
	of order 1 operators on the line bundle $L$.\footnote{For higher rank super vector bundles, we must also require that the symbol map of sections of the Atiyah Lie superalgebroid is a scalar. In general, a first order first order differential operator takes values in $\mathcal{T}_S\otimes \End(E)$.}
	It is a Lie superalgebra extension of $\mathcal{T}_S$ with a compatible left $\mathcal{O}_S$-module structure:
	\begin{align}\label{Atiyah}
		0\to \mathcal{O}_S\to \mathcal{A}_L\xrightarrow{\sym} \mathcal{T}_S\to 0,
	\end{align}
	where the anchor map is the symbol map 
	defined as $\sym(D)(df)=[D,f]$.
\end{defn}

We will make use of the following correspondence.

\begin{prop}[cf. {\cite[Theorem 2.4]{KosmannSchwarzbach-Mackenzie--2002}}]\label{Lie alg algoid rep corrsp}
	Consider the action Lie superalgebroid $\mathcal{G}$ associated to a Lie superalgebra morphism $\phi\co\mathfrak{g}\to\Gamma(S,TS)$. Separately consider a line bundle $L\to S$.
	
	There is a bijection between 
	\begin{itemize}
		\item Lie superalgebra morphisms $\mathfrak{g}\to \Gamma(S,\mathcal{A}_L)$
		
		\item Lie superalgebroid morphisms $\mathcal{G}\to \mathcal{A}_L$
	\end{itemize}
\end{prop}
\begin{proof}
	The proof in the super setting is no different than the proof given in \cite{KosmannSchwarzbach-Mackenzie--2002}.
\end{proof}

In the standard sense, a holomorphic connection on a line bundle $L$ is a splitting of the Atiyah sequence  \cref{Atiyah} as a sequence of $\mathcal{O}_S$-modules, that is an $\mathcal{O}_S$-linear map $\nabla\co \mathcal{T}_S\to \mathcal{A}_L$ such that $\sym\circ\nabla=\id_{\mathcal{T}_S}$. The curvature $c_\nabla\in \Omega^2$ is defined by 
\begin{align*}
	c_{\nabla}(v, w)=\Big[\nabla(v),\nabla(w)\Big]-\nabla\Big([v,w]\Big).
\end{align*}
When $c_\nabla=0$, we say the connection $\nabla$ is flat or integrable.  
A connection is flat exactly when it is a morphism of Lie algebroids. If $L$ admits a flat connection, then we say $\mathcal{A}_L$ has the trivial Lie superalgebroid structure, which is that of $\mathcal{O}_S\oplus\mathcal{T}_S$.

For the remainder of this section we detail the properties of Atiyah superalgebras with respect to line bundles. Our main source is \cite[Section 1.1.5]{Beilinson-Shechtman--1988}. Another reference is \cite{Higgins-Mackenzie--1990}.

We first define three constructions of Lie superalgebroids using Atiyah Lie superalgebroids.

\begin{defn}\label{Atiyah defs} Let $\lambda\in\mathbb{C}$ and let $L_i\to S$ be line bundles.
	\begin{itemize} 
		\item Define $\lambda\mathcal{A}_L$ as the semi-direct product $\mathcal{O}_S\rtimes \mathcal{A}_L$ subject to the relations $(\lambda f, 0) = (0, f )$ for $f \in \mathcal{O}_S$. The anchor maps is $a(g,D)=\sym(D)$ and the bracket is $[(g,D),(g',D')]=[D,D']$.
		
		\item Define  $\mathcal{A}_{L_1}\times \mathcal{A}_{L_2}$ as the sheaf $\mathcal{A}_{L_1}\times_{\mathcal{T}_S} \mathcal{A}_{L_2}$ with anchor map $a(D_1, D_2)=a_1(D_1)=a_2(D_2)$ and bracket $\Big[(D_1,D_2),(D'_1,D'_2)\Big]=\Big([D_1,D'_1]_1,[D_2,D'_2]_2\Big)$.
		
		\item Define $\mathcal{A}_{L_1}\otimes \mathcal{A}_{L_2}$ as the semi-direct product $\mathcal{O}_S \rtimes(\mathcal{A}_{L_1} \times \mathcal{A}_{L_2} )$
		subject to the relations $(f + g, 0) = (0, (f, g))$ for $f,g\in\mathcal{O}_S$. The anchor and bracket are the appropriate combinations of those defined above.
	\end{itemize}
\end{defn}

There is a canonical map of Lie superalgebroids $\mathcal{A}_L \to \lambda\mathcal{A}_L$ which commutes as below.
\begin{equation}\label{Aityah scalar}
	\begin{tikzcd}
		0 \arrow{r} & \mathcal{O}_S \arrow{r}\arrow{d}{\lambda\id} & \mathcal{A}_L\arrow{r}\arrow{d}& \mathcal{T}_S\arrow{r}\arrow{d}{\id}& 0\\
		0 \arrow{r} & \mathcal{O}_S \arrow{r} & \lambda\mathcal{A}_L \arrow{r} & \mathcal{T}_S\arrow{r} & 0
	\end{tikzcd}
\end{equation}

While it is not clear that any of the above constructions yield another Atiyah Lie superalgebroid, we have the following result.

\begin{lem}[cf. {\cite[Lemma in Section 1.1.5]{Beilinson-Shechtman--1988}}]\label{Atiyah defs iso lemma}
	Let $L\to S$ be a line bundle. Canonically, $\mathcal{A}_{L_1}\otimes \mathcal{A}_{L_2}\xrightarrow{\sim} \mathcal{A}_{L_1\otimes L_2}$. Further, for $n\in\mathbb{Z}$, canonically $n\mathcal{A}_L \cong \mathcal{A}_{L^{\otimes n}}$. 
\end{lem}
\begin{proof}

Analogous to the diagram \cref{Aityah scalar}, we have the commuting diagram below.
\begin{equation*}
	\begin{tikzcd}
		0 \arrow{r} & \mathcal{O}_S\times \mathcal{O}_S \arrow{r}\arrow{d}{}{{(f,g)\mapsto f+ g}} & \mathcal{A}_{L_1}\times \mathcal{A}_{L_2}\arrow{r}\arrow{d}& \mathcal{T}_S\arrow{r}\arrow{d}{\id}& 0\\
		0 \arrow{r} & \mathcal{O}_S \arrow{r} & \mathcal{A}_{L_1}\otimes \mathcal{A}_{L_2} \arrow{r} & \mathcal{T}_S\arrow{r} & 0
	\end{tikzcd}
\end{equation*}
Now consider $\phi \co \mathcal{A}_{L_1}\times \mathcal{A}_{L_2}\to \mathcal{A}_{L_1\otimes L_2}$ given by the super Leibnitz rule
\begin{align*}
	\phi(D_1 , D_2 )(s_1 \otimes s_2 ) &=
	D_1 (s_1 ) \otimes s_2 + (-1)^{|D_2||s_1|}s_1 \otimes D_2 (s_2 ),
\end{align*}
where $D_i \in \Gamma(U,\mathcal{A}_{L_i})$ and $s_i \in \Gamma(U,L_i)$. This map restricts to $\phi(f,g)=f+g$ on $\mathcal{O}_S\times\mathcal{O}_S$ and projects to $\Id_{\mathcal{T}_S}$ on $\mathcal{T}_S$. Further, since $\phi$ is $\mathcal{O}_S$-bilinear, it factors through $\mathcal{A}_{L_1}\otimes \mathcal{A}_{L_2}$. The isomorphism $\mathcal{A}_{L_1}\otimes\mathcal{A}_{L_2}\xrightarrow{\sim} \mathcal{A}_{L_1\otimes L_2}$ follows.

Comparing the definitions of the Atiyah algebras $\mathcal{A}_L\otimes \mathcal{A}_L$ and $2\mathcal{A}_L$, you find a canonical isomorphism 
\begin{align*}
	\mathcal{A}_L\otimes \mathcal{A}_L\xrightarrow{\sim}2\mathcal{A}_L &&\Big(f,(D_1,D_2)\Big)\mapsto \left(f,\frac{D_1+D_2}{2}\right)
\end{align*}
and the inverse given by the natural double inclusion. 

Let $L^\vee$ be the inverse line bundle of $L$. 
The super Leibnitz rule gives a map 
\begin{align*}
	\psi \co \mathcal{A}_{L}\to \mathcal{A}_{L^\vee} &&\psi(D)(s^\vee )(t) \coeq [D,s^\vee(t)]-(-1)^{|D||s|}s^\vee (Dt)
\end{align*}
 where $D \in \Gamma(U,\mathcal{A}_{L})$, $s^\vee\in \Gamma(U,L^\vee )$, and $t\in\Gamma(U,L)$. This map restricts to $\psi(f)=-f$ on $\mathcal{O}_S$ and projects to $\Id_{\mathcal{T}_S}$ on $\mathcal{T}_S$. By composing $\psi^{-1}$ and the canonical map $\mathcal{A}_L\to -\mathcal{A}_L$ in \cref{Aityah scalar}, we find a canonical isomorphism of Atiyah algebras $\mathcal{A}_{L^\vee}\xrightarrow{\sim}-\mathcal{A}_L$. 
 \color{black}
 The isomorphism $n\mathcal{A}_L \cong \mathcal{A}_{L^{\otimes n}}$ follows.
\end{proof}

\subsection{Pullbacks over a base change}

\begin{defn}
	A Lie superalgebroid is called transitive if its anchor map is surjective.
\end{defn}

\begin{defn}\label{pullback of algebroid def}
	Consider a map $f\co X\to Y$ and a transitive Lie superalgebroid $(\mathcal{E},a_\mathcal{E},[\;,\;]_\mathcal{E})$ on $Y$. The pullback Lie superalgebroid of $\mathcal{E}$ along $f$ is defined as $(f^!\mathcal{E},a,[\;,\;])$ where
	\begin{itemize}
		\item $f^!\mathcal{E}$ is the sheaf of sections of $E\times_{TY} TX$ over $X$, the vector bundle pullback defined by $a_\mathcal{E}$ and $df$ as below
\begin{equation*}
	\begin{tikzcd}
		E\times_{TY} TX\arrow{d}\arrow[very near start, phantom]{rd}{\lrcorner} \arrow{r}{a}&TX\arrow{d}{df}\\
		E\arrow{r}{a_\mathcal{E}}&TY
	\end{tikzcd}
\end{equation*}
		\item anchor map $a$ is the projection to $TX$ resulting from the pullback construction above
		\item bracket inherited from the bracket of $E\times TX\to X\times Y$ when viewing $E\times_{TY} TX$ as the restriction of $E\times TX$ to $\textrm{graph}(f)\subset X\times Y$
	\end{itemize} 
\end{defn}

The pullback construction is described in more detail in \cite{Higgins-Mackenzie--1990}. The pullback satisfies a universal property in the category of Lie superalgebroids; we state a consequence of this property below.

\begin{prop}[{\textcite[Proposition 1.8]{Higgins-Mackenzie--1990}}]\label{pullback morphism prop}
	Consider $f\co X\to Y$ and transitive Lie superalgebroids $\mathcal{E}$ and $\mathcal{F}$ with Lie superalgebroid morphism $\phi\co\mathcal{E}\to\mathcal{F}$. Then there exists a Lie superalgebroid morphism $f^!\phi\co f^!\mathcal{E}\to f^!\mathcal{F}$.
\end{prop}

\begin{lem}\label{pullback diagram lemma}
	In an abelian category, if the following commutative diagram has short exact rows and the map $A\xrightarrow{\sim}A'$ is an isomorphism, then the right square is a pullback square, that is $B=B'\times_{C'}C$.
	\begin{equation*}
		\begin{tikzcd}
			0 \arrow{r} & A \arrow{r}\arrow{d}[swap]{\wr} & B \arrow{r}\arrow{d} & C \arrow{r}\arrow{d} & 0\\
			0 \arrow{r}& A' \arrow{r} & B'\arrow{r} & C' \arrow{r} & 0
		\end{tikzcd}
	\end{equation*}
\end{lem}
\begin{proof}See \cref{pullback proof}.
\end{proof}

Since Atiyah Lie superalgebroids are transitive, we consider their pullbacks.

\begin{prop}\label{pullback of atiyah prop}
	 Consider a map $\phi\co X\to Y$ and a line bundle $L$ on $Y$.
	The pullback of the Atiyah Lie superalgebroid $\phi^!\mathcal{A}_L$ is isomorphic to the Atiyah Lie superalgebroid of the pullback $\mathcal{A}_{\phi^*L}$.
	\begin{equation*}
		\begin{tikzcd}
			0 \arrow{r} & \mathcal{O}_X\arrow{r} \arrow{d}{}{\id}& \phi^!\mathcal{A}_L\cong\mathcal{A}_{\phi^*L}  \arrow{r}{\sym}\arrow{d}\arrow[very near start, phantom]{rd}{\lrcorner} & \mathcal{T}_X \arrow{d}{d\phi}\arrow{r} & 0\\
			0\arrow{r} & \phi^*\mathcal{O}_{Y}\arrow{r} & \phi^*\mathcal{A}_{L}  \arrow{r}{\phi^*\sym} &  \phi^*\mathcal{T}_{Y}\arrow{r} & 0
		\end{tikzcd}
	\end{equation*}
\end{prop}
\begin{proof}
	 The natural map $\phi_*\co\mathcal{A}_{\phi^*L}\to \phi^*\mathcal{A}_L$ is given by $\phi_*\delta(s)=\delta(s\circ \phi)$ where $\delta\in\mathcal{A}_{\phi^*L}$ and $s$ is a section of $L$.
	
	By \cref{pullback diagram lemma} applied to the category of locally free sheaves, it suffices to notice that $\phi$ restricts to an isomorphism $\mathcal{O}_X\to \phi^*\mathcal{O}_Y\cong\mathcal{O}_X$, and that the right square (with the similarly defined map $d\phi$ on the tangent sheaves) commutes.
	
	The check that the bracket is the same is straightforward.
\end{proof}

\section{The supermoduli space of triples $\mathfrak{M}_{g,1^{\infty}_{\mathrm{NS}}}$}\label{moduli section}

The supermoduli space of triples $\mathfrak{M}_{g,1^{\infty}_{\mathrm{NS}}}$ has not been described explicitly in the literature. The closest object in the literature is a supermoduli space of quintuples as described in \cite{Mulase--1991,Mulase-Rabin--1991}. See \cref{classical MG1infty} for a discussion of the analogous classical moduli space of triples.

\subsection{Background material on super Riemann surfaces and their superconformal structure}

References covering the background material of superalgebras, superspaces, and SRSs in more detail include  \cite{Deligne-Morgan--1999,Bartocci-Bruzzo-HernandezRuiperez--1991,Manin--1988e} and more recently \cite{Bruzzo-HernandezRuiperez-Polishchuk--2020X,Fioresi-Kwok--2014,Carmeli-Caston-Fioresi--2011,Witten--2019}.

\begin{defn}
A super Riemann surface (SUSY curve, or SRS) is a complex supermanifold of dimension $1|1$ with a maximally nonintegrable distribution $\mathcal{D}_\Sigma$ of rank $0|1$. Precisely, a super Riemann surface is the data $(\Sigma,\mathcal{O}_\Sigma,\mathcal{D}_\Sigma)$ such that 
\begin{itemize}\itemsep0em
\item $\Sigma$ is a complex manifold of dimension $1$
\item $\mathcal{O}_\Sigma$ is a sheaf on $\Sigma$ of supercommutative $\mathbb{C}$-algebras
\item $\mathcal{O}_\Sigma$ is locally isomorphic to $\mathcal{O}_{\mathbb{C}}[\xi]$
\item $\mathcal{D}_\Sigma$ is an odd subbundle of the tangent bundle $\mathcal{T}_\Sigma$
\item the induced map by Lie bracket $[\:\:,\:\:]:\mathcal{D}_\Sigma^{\otimes2}\to \mathcal{T}_\Sigma/\mathcal{D}_\Sigma$ is an isomorphism.
\end{itemize}

Define a family of SUSY curves in the holomorphic category: a proper submersion $\pi:X\to S$ of super Riemann surfaces (i.e. relative dimension $1|1$) parameterized by the complex supermanifold $S$ with a relative distribution $\mathcal{D}_{X/S}\subset \mathcal{T}_{X/S}$ of rank $0|1$ such that the map $[\:\:,\:\:]:\mathcal{D}_{X/S}^{\otimes2}\to \mathcal{T}_{X/S}/\mathcal{D}_{X/S}$ is an isomorphism. 

\end{defn}

Locally on a single SRS, we can always find coordinates (called \emph{superconformal coordinates}) $z|\zeta$ such that 
\begin{align*}
	D_\zeta\coeq\frac{\partial}{\partial \zeta}+\zeta\frac{\partial}{\partial z}
\end{align*}
generates the distribution $\mathcal{D}_\Sigma$. In more generality for a family, we can find relative superconformal coordinates $z|\zeta$ such that $D_\zeta$ generates $\mathcal{D}_{X/S}$.
Naturally, we call a change of coordinates $z|\zeta\mapsto x|\xi$ superconformal if $D_\zeta$ and $D_\xi$ are multiples of one another. Alternatively, the condition 
\begin{align}\label{superconformal local change}
	D_\zeta x=\xi D_\zeta \xi
\end{align}
is equivalent to the change of coordinates being superconformal.

\begin{defn}\label{Mg def}
	The moduli space of super Riemann surfaces $\mathfrak{M}_g$ is the superspace representing the functor 
	\begin{align*}
		S\mapsto\left\{\textrm{families of SUSY curves } \pi\co X\to S\right\}.
	\end{align*}
\end{defn}
Similar to the moduli space of classical Riemann surfaces $\mathcal{M}_g$, the moduli space of SRSs is representable as a stack, specifically an orbifold or Deligne-Mumford stack for $g\geq 2$. The DM stack structure of $\mathfrak{M}_g$ and its compactification has recently been studied in papers \cite{Codogni-Viviani--2019,Felder-Kazhdan-Polishchuk--2020X,Moosavian-Zhou--2019X}.

The following exact sequence, or equivalently its dual also listed below, provides much information about the structure of SUSY curves.
\begin{align*}
	0\to \mathcal{D}_{X/S}\to \mathcal{T}_{X/S}\to \mathcal{T}_{X/S}/\mathcal{D}_{X/S}\to 0 && 0\to \mathcal{D}^{\perp}_{X/S}\to \Omega^1_{X/S}\to \mathcal{D}^{-1}_{X/S}\to 0
\end{align*}
Since $\mathcal{D}^{\perp}_{X/S}$ is dual to the quotient $\mathcal{T}_{X/S}/\mathcal{D}_{X/S}$, then we find that $\mathcal{D}^{\perp}_{X/S}$ is generated in local superconformal coordinates by $dz-\zeta d\zeta$. 
Further, the SRS structure gives isomorphisms with powers of the distribution $\mathcal{D}$ so that these exact sequences are equivalent to
\begin{align}\label{SRS SES}
	0\to \mathcal{D}_{X/S}\to \mathcal{T}_{X/S}\to \mathcal{D}^{\otimes 2}_{X/S}\to 0 && 0\to \mathcal{D}^{\otimes -2}_{X/S}\to \Omega^1_{X/S}\to \mathcal{D}^{-1}_{X/S}\to 0.
\end{align}
To be precise, the projection above is the composition $\mathcal{T}_{X/S}\to \mathcal{T}_{X/S}/\mathcal{D}_{X/S} \xleftarrow[]{\sim}\mathcal{D}^2_{X/S}$ using the isomorphism induced by the Lie bracket.

We recall the definition of the Berezinian, the supergeneralization of the determinant. Let $A$ be a superalgebra. Consider $\mathrm{GL}(m|n,A)$, the supergroup of automorphisms of the free $A$-module $A^{m|n}$. For $X\in \mathrm{GL}(m|n,A)$ expressed in as an $(m+n)\times (m+n)$ matrix the Berezinian is defined as
\begin{align}\label{Ber def}
X=\begin{bmatrix}
A & B\\ C & D
\end{bmatrix} && \Ber X \coeq \det(A-BD^{-1}C)\textstyle\det^{-1}D.
\end{align}
We note that $\Ber(\mathrm{exp} X)=\mathrm{exp}(\str X)$, where $\str$ is the supertrace. Written in terms of the matrix in \cref{Ber def}, $\str X\coeq\tr A-\tr D$.

The parity reversing operator denoted $\Pi$ maps a super vector space $V=V_0\oplus V_1$ to the super vector space with even and odd components swapped: $(\Pi V)_0\coeq V_1$ and $(\Pi V)_1\coeq V_0$. Since $\Pi$ has no analogous classical operator, neither does the $\Pi$-transpose given below.
\begin{align}\label{Pi transpose}
	X^\Pi=\begin{bmatrix}
		D & C\\ B & A
	\end{bmatrix}
\end{align}

Let $\mathcal{F}$ be a vector bundle of rank $m|n$ with associated transition functions $\{f_{ij}\}$. Then we define $\Ber\mathcal{F}$ to be the bundle with transition functions $\{\Ber f_{ij}\}$. We enforce that $\Ber\mathcal F$ is rank $1|0$ if $n$ is even and rank $0|1$ if $n$ is odd. We denote a trivializing section of $\Ber \mathcal{F}$ as $[e_1, \ldots, e_m|f_1,\ldots ,f_n]$, where $e_1,\ldots, e_m|f_1,\ldots, f_n$ are trivializing local coordinates for $\mathcal{F}$.

The Berezinian of the exact sequence \cref{SRS SES} induces an isomorphism 
\begin{align}\label{Ber sheaf iso}
\Ber\Omega_{X/S}^1\cong \Ber\mathcal{D}^{-2}_{X/S}\otimes \Ber \mathcal{D}^{-1}_{X/S}\cong \mathcal{D}^{-2}_{X/S} \otimes \mathcal{D}_{X/S}\cong \mathcal{D}^{-1}_{X/S}.
\end{align}
We define $\omega_{X/S}\coeq \Ber\Omega_{X/S}^1$, which we call the relative Berezinian of $X$ over $S$. Therefore in local superconformal coordinates, we may describe the relative Berezinian as relative one-forms modulo $dz-\zeta d\zeta$.

Further, properties of the Berezinian of a SRS give an interesting construction of a relative superconformal differential. We define $\delta$ as the composition of the classical relative differential $d$ and the quotient $q$
\begin{equation}\label{delta diagram}
	\begin{tikzcd}
		\mathcal{O}_{X} \arrow{r}{d} \arrow{rd}[swap]{\delta}&\Omega_{X/S}\arrow{d}{q}\\
		& \omega_{X/S}
	\end{tikzcd}
\end{equation}
where $q$ is the map $\Omega_{X/S}\to \mathcal{D}^{-1}_{X/S}$ in \cref{SRS SES} combined with the isomorphism in \cref{Ber sheaf iso}. In local superconformal coordinates, the superconformal differential $\delta$ is given by
\begin{align}\label{delta in local coor}
	\delta f = (D_\zeta f) \;[dz|d\zeta].
\end{align}
We remark that an equivalent condition for a local coordinate system $(x|\xi)$ to be superconformal is 
\begin{align}\label{delta superconformal}
	\delta x=\xi \,\delta \xi
\end{align}
which is a coordinate-independent version of \cref{superconformal local change}.

The vector fields which preserve the superconformal structure $\mathcal{D}_{X/S}$ are called superconformal. The sheaf of superconformal vector fields $\mathcal{T}^s_{X/S}$ on a family of super Riemann surfaces is the largest subsheaf of $\mathcal{T}_{X/S}$ such that $[\mathcal{T}^s_{X/S},\mathcal{D}_{X/S}]=\mathcal{D}_{X/S}$. 
In local superconformal coordinates $z|\zeta$, every superconformal vector field has a unique function $f$ such that the vector field is given by
\begin{align}\label{superconformal V}
	[ f D_\zeta,D_\zeta]= f(z|\zeta)\frac{\partial}{\partial z}+\frac{(-1)^{|f|}D_\zeta f(z|\zeta)}{2}D_\zeta.
\end{align}

This expression in local coordinates shows two facts. Firstly, $\mathcal{T}^s_{X/S}$ is not an $\mathcal{O}_{X}$ submodule of $\mathcal{T}_{X/S}$; it is only a subsheaf of $\mathbb{C}$ super vector spaces. However, via the map $[\;,\;]\co\mathcal{D}_{X/S}\otimes\mathcal{D}_{X/S}\to \mathcal{T}^s_{X/S}$, the sheaf $\mathcal{T}^s_{X/S}$ has an $\mathcal{O}_{X}$ module structure.
Secondly, the map $\mathcal{T}^s_{X/S}\to \mathcal{T}_{X/S}/\mathcal{D}_{X/S}$ given locally by $[fD_\zeta,D_\zeta]\mapsto f\pp{z} \mod D_\zeta$ is a canonical isomorphism (of sheaves of $\mathbb{C}$ super vector spaces). Therefore we have $\mathcal{T}_{X/S}\cong\mathcal{T}^s_{X/S}\oplus \mathcal{D}_{X/S}$ (as sheaves of $\mathbb{C}$ super vector spaces).

\subsection{Definition}\label{Mg subsection}

Recall the definition of $\mathfrak{M}_g$ in \cref{Mg def}. In order to consider the relationship of SRSs with the Sato Grassmannian, we must consider SRSs with extra information, as described in this section.

Consider a family of SUSY curves $\pi\co X\to S$. 
A section $\sigma:S\to X$ of the morphism $\pi$ is a coherent way to choose a Neveu-Schwarz (NS) puncture in each fiber. Since the section $\sigma$ is locally $\Spec A\to \Spec A[z|\zeta]$, these NS punctures  may be described by $z=z_0$ and $\zeta=\zeta_0$ for some even function $z_0$ and some odd function $\zeta_0$ in $\mathcal{O}_S$. Note that the image of $\sigma$ is a codimension $1|1$ subsuperspace of $X$. However, the orbit generated by the $\mathcal{D}_{X/S}$ canonically associates a divisor of $X$ to the section $\sigma$. Expressed with an odd parameter $\alpha$, the divisor, denoted $P=\div(s)$, is the codimension $1|0$ subsuperspace \cite{Witten--2019}
\begin{align*}
	z&=z_0+\alpha \zeta\\
	\zeta&=\zeta_0+\alpha.
\end{align*}

\begin{defn}
	The moduli stack $\mathfrak{M}_{g,1^{k}_{\mathrm{NS}}}$ is the moduli space of triples $(\Sigma,p,z|\zeta)$, where $\Sigma$ is a genus $g$ SUSY curve, $p$ is the divisor associated to a NS puncture on $\Sigma$, both $z$ and $\zeta$ are a $k$-jet equivalence class of an even (resp. odd) formal parameter vanishing at $p$. 
\end{defn}

Specifically, we have the parameters $z,\zeta\in \hat{\mathcal{O}}_{p}$, the formal neighborhood at $p$, and $z$ and $\zeta$ each have a zero of order one at $p$.
And since $z|\zeta$ is a $k$-jet, the coordinate system is equivalent to any $y|\gamma$ such that $y|\gamma$ equals $z|\zeta$ modulo $\mathfrak{m}^{k+1}_p|\mathfrak{m}^{k+1}_p$. Note that $\mathfrak{M}_{g,1^{k}_{\mathrm{NS}}}$ has dimension ${3g-3+1+k}\,|\,{2g-2+1+k}$. For $g\geq 1$, the moduli space of SRSs with one NS puncture $\mathfrak{M}_{g,1_{\text{NS}}}$ is a Deligne-Mumford stack. As well, $\mathfrak{M}_{0,1^{3}_{\mathrm{NS}}}$ is a Deligne-Mumford stack. So for high enough $k$, we have that $\mathfrak{M}_{g,1^k_{\text{NS}}}$ is a Deligne-Mumford stack.

\begin{defn}
	By taking the projective limit of $\mathfrak{M}_{g,1^{k}_{\mathrm{NS}}}$ as $k\to \infty$, we construct the pro-Deligne-Mumford stack 
	\begin{align*}
		\mathfrak{M}_{g,1^{\infty}_{\mathrm{NS}}}=\lim\limits_{\longleftarrow}\mathfrak{M}_{g,1^{k}_{\mathrm{NS}}}.
	\end{align*} 
\end{defn}
Clearly, $\mathfrak{M}_{g,1^{\infty}_{\mathrm{NS}}}$ infinite dimensional (in both even and odd dimensions) since the local coordinate system $(z|\zeta)$ now contains an infinite dimensional amount of information.

Simply by forgetting the coordinate system $z|\zeta$ and the NS puncture $P$, we have the projection
\begin{align*}
	\mathfrak{M}_{g,1^{\infty}_{\mathrm{NS}}} &\to \mathfrak{M}_{g}\\
	(X/S,P,z|\zeta)&\mapsto X/S
\end{align*}
where $\mathfrak{M}_{g}$ is the supermoduli space of SUSY curves of genus $g$.

\subsection{Superconformal Noether normalization and a perfect Lie superalgebroid}

A technical lemma needed to show our main result in \cref{flat} is the perfectness of the Lie superalgebroid of superconformal vector fields on a family of open super Riemann surfaces, which holds locally on the base. The analogous result in classical geometry was proven in \cite[Lemma 4.3]{Arbarello-DeConcini-Kac-Procesi--1988} and later in \cite[Proposition 6.2]{Kawazumi--1993}. Compared to the classical case, we need to consider the Lie superalgebroid (over a family of SRSs) as opposed to the Lie algebra over a single Riemann surface. 

A super Noether normalization was shown in \cite[Proposition 4.3]{Masuoka-Zubkov--2020}. However, we require a stronger version: relative superconformal Noether normalization, avoiding a divisor. Our normalization lemma below is formulated in the complex analytic category.

\begin{lem}[Relative Superconformal Noether normalization]\label{superconformalNoether}
	Let $\pi\co X\to S$ be a family of closed SRSs with distribution $\mathcal{D}_{X/S}$. Consider the family of open SRSs $X\setminus P\to S$, where $P$ is a divisor representing NS punctures. Let $(x|\xi)$ be relative superconformal coordinates on $\AA^{1|1}\times S$ over $S$.

	Then for every $s\in S$, there exists a small enough open supersubspace $s\in U\subseteq S$ and a finite surjective morphism of families over $U$ 
	\begin{align*}
		p\co \pi^{-1}(U)\setminus P\to \AA^{1|1}\times U
	\end{align*}
  such that $p^*\left(\pp{\xi}+\xi\pp{x}\right)$ generates $\mathcal{D}_{X/S}\big|_{\pi^{-1}(U)\setminus P}(R)$, where $R$ is the relative ramification divisor of $p$ over $U$.
	
	Further for every small enough neighborhood $U\subset S$, we may find another morphism $p'$ as above such that the relative ramification divisors of $p$ and $p'$ are disjoint.
\end{lem}

\begin{proof}
	See \cref{Noether Normalization proof appendix}.
\end{proof}

\color{black}
\begin{prop}\label{commutant}
	Let $\pi\co X\to S$ be a family of closed SRSs with distribution $\mathcal{D}_{X/S}$. Consider the family of open SRSs $X\setminus P\to S$, where $P$ is a divisor representing NS punctures. 
	
	Consider a small enough open supersubspace $U\subseteq S$ and denote the Lie superalgebroid of global superconformal vector fields over the base $U$ as $\mathcal{K}\coeq \pi_*\big(\mathcal{T}^s_{X/S}\big|_{\pi^{-1}(U)}\big)$. Then the Lie superalgebroid $\mathcal{K}$ is perfect, that is to say, $\mathcal{K}/[\mathcal{K},\mathcal{K}] = 0$.
\end{prop}

\begin{proof}
	
	Firstly we consider genus $g=0$. In this case, we may consider an injective map $\pi^{-1}(U)\to \AA^{1|1}\times U$, and we can use the standard (global) superconformal coordinates $z|\zeta$ such that $\mathcal{D}_{X/S}\big|_{\pi^{-1}(U)}$ is generated by $D_\zeta=\frac{\partial}{\partial \zeta}+\zeta\frac{\partial}{\partial z}$. Then elements of $\Gamma(V,\mathcal{K})$, where $V$ is open in $U$, are of the form $[fD_\zeta,D_\zeta]$, 
	where $f=f(z|\zeta)$ is a regular function on $\pi^{-1}(V)$. 	
	The supercommutator of superconformal vector fields gives
	\begin{align*}  
		\Big[[fD_\zeta,D_\zeta],[gD_\zeta,D_\zeta]\Big]=\left[\left(2f(\partial_z g)+(-1)^{|f|}(D_\zeta f)(D_\zeta g)-2(\partial_z f) g \right)D_\zeta,D_\zeta\right]
	\end{align*}
	which was stated in \cite[Section 3.6]{Manin--1988}.
	For an arbitrary function $h$, we find $[hD_\zeta,D_\zeta]\in [\mathcal{K},\mathcal{K}]$ since 
	\begin{align*}
		[hD_\zeta,D_\zeta]=\frac{1}{3}\Big[[hD_\zeta,D_\zeta],[zD_\zeta,D_\zeta]\Big]+\frac{1}{3}\Big[[D_\zeta,D_\zeta],[hzD_\zeta,D_\zeta]\big]+\frac{1}{3}\Big[[h\zeta D_\zeta,D_\zeta],[\zeta D_\zeta,D_\zeta]\Big].
	\end{align*}

	For affine SRSs of genus $g\geq 1$, we will follow the idea of the argument for the classical statement in \cite{Arbarello-DeConcini-Kac-Procesi--1988}, which uses Noether's normalization lemma.

	For the family $ \pi^{-1}(U)\to U$, denote the distribution by $\mathcal{D}\coeq\mathcal{D}_{X/S}\big|_{\pi^{-1}(U)} $ and the sheaf of functions by $\mathcal{O}\coeq\mathcal{O}_{X}\big|_{\pi^{-1}(U)}$. 
	Consider two finite surjective morphisms $p_i\co \pi^{-1}(U)\to\AA^{1|1}\times U$ as in \cref{superconformalNoether} such that their relative ramification divisors over $U$ are disjoint.

	Since $\pi^{-1}(U)$ is a family of open SRSs over a small enough base, there exists a global nowhere zero odd holomorphic section $D\in\Gamma(\pi^{-1}(U),\mathcal{D})$. Since $\mathcal{D}$ is an $\mathcal{O}$-module of relative rank $0|1$, then any global section of $\mathcal{D}$ can be written uniquely as $fD$ for some $f\in\Gamma(\pi^{-1}(U),\mathcal{O})$.

	Further, we have the global nowhere zero even section $D\otimes D\in\Gamma(\pi^{-1}(U),\mathcal{D}^{\otimes 2})$, which corresponds to a global nowhere zero even superconformal vector field $[D,D]\in\Gamma(U,\mathcal{K})$. Note that every element of $\Gamma(V,\mathcal{K})$ can be written uniquely as $[hD,D]$ for $h\in \Gamma(\pi^{-1}(V),\mathcal{O})$.
	We note that the bracket of global superconformal vector fields $\mathcal{K}$ is given by the following.
	\begin{align}  \label{global bracket}
		\bigg[[gD,D],[hD,D]\bigg]
		&=\bigg[\left(2g\left(D^2h\right)+(-1)^{|g|}(D g)(D h)-2\left(D^2g\right) h\right)D,D\bigg]
	\end{align}

	We define the following objects on $\pi^{-1}(U)$.
	\begin{align*}
	z_i\coeq p_i^*(z) &&\zeta_i\coeq p_i^*(\zeta) && D_i\coeq p_i^*\left(\pdz+\zeta\pdz\right)
\end{align*}
	We have that $z_i\in \Gamma(\pi^{-1}(U),\mathcal{O})$ is a global regular even function, $\zeta_i\in \Gamma(\pi^{-1}(U),\mathcal{O})$ is a global regular odd function, and $D_i\in \Gamma(\pi^{-1}(U),\mathcal{D}(R_i))$ is a global rational odd vector field, where $R_i$ is the relative ramification divisor of $p_i$ over $U$. 
	By \cite[Theorem 26.5]{Forster--1991} generalized to our case of SRSs over a small enough base $U$, there exists a global regular even function $f_i\in \Gamma(\pi^{-1}(U),\mathcal{O})$ with relative divisor exactly $-R_i$. Thus, $f_i D_i$ is a global regular vector field on $\pi^{-1}(U)$, nonzero outside of the ramification locus $[R_i]$. 
	We therefore have 
	\begin{align*}
		f_iD_i=a_iD
	\end{align*}
	for some holomorphic even function $a_i\in \Gamma(\pi^{-1}(U),\mathcal{O})$.

	We list the following expressions and their simplifications for use later in the final calculation.
	\begin{align*}
	a_iD(z_i)=f_iD_i(z_i)&=f_i\zeta_i\\
	a_iD(\zeta_i)=f_iD_i(\zeta_i)&=f_i\\
	a_i\Big((Dz_i)-\zeta_i(D\zeta_i)\Big)&=0\\
	a_i^2\Big((D^2z_i)+\zeta_i(D^2\zeta_i)\Big)&=f_i^2 
\end{align*}
The last equality follows from simplifying the below expression in two different ways.
\begin{align*}
	\Big([a_i^2D,D]z_i\Big)+\zeta_i\Big([a_i^2D,D]\zeta_i\Big)&=2a_i^2\Big((D^2z_i)+\zeta_i(D^2\zeta_i)\Big)+(Da_i^2)(Dz_i)+\zeta_i(D a_i^2)(D\zeta_i)\\
	&=2a_i^2\Big((D^2z_i)+\zeta_i(D^2\zeta_i)\Big)+2(Da_i)f_i\zeta_i+2\zeta_i(Da_i)f_i\\
	&=2a_i^2\Big((D^2z_i)+\zeta_i(D^2\zeta_i)\Big)\\
	\Big([a_i^2D,D]z_i\Big)+\zeta_i\Big([a_i^2D,D]\zeta_i\Big)&=\Big([f_i^2D_i,D_i]z_i\Big)+\zeta_i\Big([f_i^2D_i,D_i]\zeta_i\Big)\\
	&=2f_i^2(D_i^2z_i)+2\zeta_if_i^2(D_i^2\zeta_i)\\
	&=2f_i^2
\end{align*}

	Since $f_1$ and $f_2$ are even and relatively prime, we can find holomorphic functions $c_i\in \Gamma(\pi^{-1}(U),\mathcal{O})$ such that
	\begin{align*}
		h=c_1f^2_1+c_2f^2_2
	\end{align*}
	for any given holomorphic function $h\in \Gamma(\pi^{-1}(U),\mathcal{O})$. We claim that $[hD,D]\in\mathcal{K}(U)$ is equal to 
	\begin{align*}
		[hD,D]=\sum_{i=1,2}\frac{1}{3}\bigg[[c_ia_i^2D,D],[z_iD,D]\bigg]+\frac{1}{3}\bigg[[D,D],[c_ia_i^2z_iD,D]\bigg]+\frac{1}{3}\bigg[[c_ia_i^2\zeta_i D,D],[\zeta_i D,D]\bigg].
	\end{align*}
	The following calculation confirms this claim. The first equality below follows from  formula \cref{global bracket}.
	\begin{align*}
		\bigg[[c_ia_i^2D,D],[z_iD,D]\bigg]&+\bigg[[D,D],[c_ia_i^2z_iD,D]\bigg]+\bigg[[c_ia_i^2\zeta_i D,D],[\zeta_i D,D]\bigg]\\
		&=\bigg[\Big(4c_ia_i^2\left(D^2z_i\right)+(D c_ia_i^2)(D z_i)\Big)D,D\bigg]
		\\&\qquad\qquad\qquad
		+\bigg[\Big(4c_ia_i^2\zeta_i\left(D^2\zeta_i\right)-c_ia_i^2(D \zeta_i)(D \zeta_i)-(D c_ia_i^2)\zeta_i(D \zeta_i)\Big)D,D\bigg]\\
		&=4\Big[ c_ia_i^2\Big((D^2z_i)+\zeta_i(D^2\zeta_i)\Big)D,D\Big]-\Big[c_ia_i^2(D\zeta_i)(D\zeta_i) D,D\Big]
		\\&\qquad\qquad\qquad\qquad\qquad\qquad\qquad\qquad\quad
		+\bigg[(D c_ia_i^2)\Big((D z_i)-\zeta_i(D \zeta_i)\Big)D,D\bigg]\\
		&=4\Big[ c_if_i^2D,D\Big]-\Big[c_if_i^2 D,D\Big]\\
		&=3\big[c_if_i^2D,D\big]
	\end{align*}

\end{proof}

	\begin{cor}
	Let $\Sigma$ be a split super Riemann surface. The Lie superalgebra $\mathfrak{k}\coeq \Gamma(\Sigma,\mathcal{T}^s_\Sigma)$ is perfect, that is to say $H_1(\mathfrak{k};\CC)=\mathfrak{k}/[\mathfrak{k},\mathfrak{k}]=0$.
\end{cor}
\begin{proof}
	This is \cref{commutant} for $S$ a point.
\end{proof}

\subsection{Action of the super Witt algebra}\label{witt action subsection}

The classical Witt algebra has two super generalizations, denoted $K(1;\epsilon=0)$ and $K(1;\epsilon=1)$ in \cite{Kac-vandeLeur--1989}, corresponding to the Neveu-Schwarz and Ramond superalgebras respectively. We are only concerned with the Neveu-Schwarz case in this paper, which we recall in the definition below.

\begin{defn}
The super Witt algebra $\mathfrak{switt}$ is the Lie superalgebra of superconformal vector fields on a punctured formal neighborhood of a point in $\mathbb{C}^{1|1}$. Explicitly, the elements of $\mathfrak{switt}$ are of the form \cref{superconformal V} with $f$ a  formal super Laurent series:
\begin{align*}
\mathfrak{switt}=\left\{f(z|\zeta)\frac{\partial}{\partial z}+\frac{(-1)^{|f|}D_\zeta f(z|\zeta)}{2}D_\zeta\co f(z|\zeta)\in \mathbb{C}(\!(z)\!)[\zeta]\right\}.
\end{align*}
\end{defn}
Expressed as a basis of even and odd vector fields:
\begin{align*}
f&=z^{-n+1} &L_n&=z^{-n+1}\frac{\partial}{\partial z}+\frac{-n+1}{2}z^{-n}\zeta\frac{\partial}{\partial\zeta} & n&\in \mathbb{Z}\\
f&=2i\zeta z^{-r+\frac{1}{2}} &G_r&=i\,z^{-r+\frac{1}{2}}\left(\zeta\frac{\partial}{\partial z}-\frac{\partial}{\partial\zeta}\right) & r&\in \mathbb{Z}+\frac{1}{2}
\end{align*}
\begin{align*}
[L_m,L_n]=(m-n)L_{m+n} && [L_n,G_r]=\left(\frac{n}{2}-r\right)G_{n+r} && [G_r,G_s]=2L_{r+s}
\end{align*}

Now, we use the sheaf of superconformal vector fields to analyze the tangent spaces of the moduli spaces. For $\pi:X\to \mathfrak{M}_{g}$ the universal family of super Riemann surfaces, we have the short exact sequence
\begin{align*}
0 \to \mathcal{T}^s_{X/\mathfrak{M}_{g}} \to \mathcal{T}^s_X \to \pi^*(\mathcal{T}_{\mathfrak{M}_{g}}) \to 0.
\end{align*}
In the long exact sequence of higher direct images, we find the Kodaira-Spencer map \cite{Huybrechts--2005} is
\begin{align*}
\delta: \mathcal{T}_{\mathfrak{M}_{g}}  \xrightarrow{\sim} R^1 \pi _*\mathcal{T}^s_{X/\mathfrak{M}_{g}},
\end{align*}
where the source $\mathcal{T}_{\mathfrak{M}_{g}}$ follows from the projection formula $
\pi_* (\pi^*(\mathcal{T}_{\mathfrak{M}_{g}})) = \mathcal{T}_{\mathfrak{M}_{g}} \otimes \pi_* (\mathcal{O}_X) = \mathcal{T}_{\mathfrak{M}_{g}}$. 
The restriction to a SRS $\Sigma$ is
\begin{align*}
\delta_\Sigma\co T_{\Sigma}\mathfrak{M}_{g} \xrightarrow{\sim} H^1(\Sigma,\mathcal{T}^s_\Sigma).
\end{align*}
This map is a natural bijection between the isomorphism classes of infinitesimal deformations $v\in T_{\Sigma}\mathfrak{M}_g$ and elements of $H^1(\Sigma,\mathcal{T}^s_\Sigma)$. 
From this we can find the dimension of $\mathfrak{M}_g$ is ${3g-3}\,|\,{2g-2}$ for $g\geq 2$.

Further, for $\pi:X\to \mathfrak{M}_{g,1^{k}_{\mathrm{NS}}}$ the family of super Riemann surfaces with NS punctures represented by the divisor $P$ and a $k$-jet coordinate system near the punctures, the Kodaira-Spencer map preserving this extra structure is 
\begin{align*}
\delta: \mathcal{T}_{\mathfrak{M}_{g,1^{k}_{\mathrm{NS}}} } \xrightarrow{\sim} R^1 \pi _*\mathcal{T}^s_{X/\mathfrak{M}_{g,1^{k}_{\mathrm{NS}}}}(-(k+1)P),
\end{align*}
which locally for $(\Sigma,p,z|\zeta)$ is 
\begin{align*}
\delta_{(\Sigma,p,z|\zeta)}\co T_{(\Sigma,p,z|\zeta)}\mathfrak{M}_{g,1^{k}_{\mathrm{NS}}} \xrightarrow{\sim} H^1(\Sigma,\mathcal{T}^s_\Sigma(-(k+1)p))
\end{align*}
where $P|_\Sigma = p$ is the NS puncture on $\Sigma$.
From this we can find the dimension of $\mathfrak{M}_{g,1^k_{\mathrm{NS}}}$ is ${3g-2+k} \, | \, {2g-1+k}$ for $g\geq 2$.

\begin{prop}\label{switt action}
The super Witt algebra acts on the moduli space $\mathfrak{M}_{g,1^\infty_{\mathrm{NS}}}$ by vector fields. That is, there exists a Lie superalgebra antihomomorphism
\begin{align*}
	\Lambda\co\mathfrak{switt}
	\to \Gamma(\mathfrak{M}_{g,1^\infty_{\mathrm{NS}}},\mathcal{T}_{\mathfrak{M}}).
\end{align*}
\end{prop}
\begin{proof} Consider $\pi\co X\to \mathfrak{M}_{g,1^k_{\mathrm{NS}}}$.
	Let $U$ be a tubular formal neighborhood of the NS puncture divisor $P$ in $X$. Choosing the relative superconformal formal parameters $z|\zeta$ on $U$ over $S$ such that the divisor $P$ is given by $z=0$ gives a trivialization of $\mathcal{T}^s_{U/S}$. 
	
	Using relative \v{C}ech cohomology, then 
\begin{align*}
	\pi_*\mathcal{T}_{X/\mathfrak{M}}^{s}(-(k+1)P)\to \pi_*\mathcal{T}_{(X\setminus P)/\mathfrak{M}}^{s}\to \pi_*\mathcal{T}^{s}_{(U\setminus P)/\mathfrak{M}}\Big/\pi_*\mathcal{T}^{s}_{U/\mathfrak{M}}(-(k+1)P)\to R^1\pi_*\mathcal{T}_{X/\mathfrak{M}}^{s}(-(k+1)P).
\end{align*}
We can identify $\pi_*\mathcal{T}^{s}_{(U\setminus P)/\mathfrak{M}}\cong \mathfrak{switt}\hat{\otimes} \mathcal{O}_\mathfrak{M}$. Taking the projective limit over $k$,
\begin{align}\label{witt superalgebroid sequence}
	0\to
	\pi_*\mathcal{T}_{(X\setminus P)/\mathfrak{M}}^{s}\to \mathfrak{switt}\hat{\otimes} \mathcal{O}_\mathfrak{M}
	\to \mathcal{T}_{\mathfrak{M}}.
\end{align}
Further, taking global sections over $\mathfrak{M}_{g,1^\infty_{\mathrm{NS}}}$
\begin{align*}
	0\to
	\Gamma(\mathfrak{M}_{g,1^\infty_{\mathrm{NS}}},\pi_*\mathcal{T}_{(X\setminus P)/\mathfrak{M}}^{s})\to \Gamma(\mathfrak{M}_{g,1^\infty_{\mathrm{NS}}},\mathfrak{switt}\hat{\otimes} \mathcal{O}_\mathfrak{M})
	\to \Gamma(\mathfrak{M}_{g,1^\infty_{\mathrm{NS}}},\mathcal{T}_{\mathfrak{M}}).
\end{align*}
Lastly, taking only constant sections in $\Gamma(\mathfrak{M}_{g,1^\infty_{\mathrm{NS}}},\mathfrak{switt}\hat{\otimes} \mathcal{O}_\mathfrak{M})$, we obtain the map $\Lambda$ as stated in the proposition.

Checking that $\Lambda$ is a Lie superalgebra antihomomorphism follows from the naturality of the Kodaira-Spencer map and can be checked using relative \v{C}ech cohomology based on the covering $U$ and $X\setminus P$.
\end{proof}

We note that the tangent space fibers can be identified as
\begin{align*}
	T_{(\Sigma,p,z|\zeta)}(\mathfrak{M}_{g,1^\infty_{\mathrm{NS}}})
	& \cong \mathfrak{switt}\big/\Gamma(\Sigma\setminus p, \mathcal{T}^s_\Sigma).
\end{align*}

For use later, we note that $\mathfrak{switt}\hat{\otimes}\mathcal{O}_\mathfrak{M}$ in \cref{witt superalgebroid sequence} is essentially the action Lie superalgebroid $\mathcal{W}$ defined below.

We denote by $\mathfrak{g}^\text{op}$ the opposite Lie superalgebra, that is the Lie superalgebra with the negative bracket of $\mathfrak{g}$, that is $[x,y]_\mathfrak{g}=-[x,y]_{\mathfrak{g}^\text{op}}$. We denote the Lie superalgebra antihomomorphism between these as $\text{op}\co \mathfrak{g}\to \mathfrak{g}^\text{op}$.

\begin{defn}\label{switt superalgebroid defs}
	According to \cref{action superalgebroid def}, define $(\mathcal{W},a_\Lambda)$ to be the action Lie superalgebroid associated to $\Lambda\circ \text{op}\co \mathfrak{switt}^\text{op}\to \Gamma(\mathfrak{M}_{g,1^\infty_{\mathrm{NS}}},\mathcal{T}_{\mathfrak{M}})$ as in \cref{switt action}.
\end{defn}

\subsection{The super Mumford isomorphism}

Note that the Berezinian of cohomology used below is the natural super generalization of the determinant of cohomology defined in \cite{Deligne--1987}.

\begin{defn} \label{lambda j/2 def}
	Let $\pi: X\to S$ be a proper family of complex supermanifolds of dimension $1|1$. Let $\mathcal{F}$ be a locally free sheaf on $X$. Then the Berezinian of cohomology of $\mathcal{F}$ is a sheaf on $S$ of rank $1|0$ or $0|1$ given by
	\begin{align*}
		B(\mathcal{F})\coeq \otimes_i\left(\Ber R^i\pi_*\mathcal{F}\right)^{(-1)^i}.
	\end{align*} 
	We define the Berezinian line bundles $\lambda_{j/2}$ for the universal family $\pi\co X\to \mathfrak{M}_{g,1^{\infty}_{\mathrm{NS}}}$ as
	\begin{align*}
		\lambda_{j/2}\coeq B(\omega_{X/\mathfrak{M}}^{\otimes j})
	\end{align*}
	where $\omega_{X/\mathfrak{M}}\coeq \Ber \Omega_{X/\mathfrak{M}}^1$ is the relative Berezinian. 
\end{defn}

\begin{thm}[\textcite{Voronov--1988}; \textcite{Deligne--1988U}]
The super Mumford isomorphism is the collection of canonical isomorphisms
\begin{align*}
\lambda_{j/2}\cong \lambda_{1/2}^{-(-1)^{j}(2j-1)}, \qquad\textrm{ in particular } \qquad \lambda_{3/2}\cong \lambda_{1/2}^{5}.
\end{align*}
\end{thm}

Since the super Mumford isomorphism is canonical, the super Mumford form $\mu$ may be defined as the image of $1\in\mathcal{O}_S$ under the isomorphism $\mathcal{O}_S \cong \lambda_{3/2}\otimes \lambda_{1/2}^{-5}$.

\section{The semi-infinite super Grassmannian of $\mathbb{C}(\!(z)\!)[\zeta]\,[dz|d\zeta]^{\otimes j}$}\label{Gr section}

The semi-infinite super Grassmannian or super Sato Grassmannian has been discussed in \cite{Mulase--1991,Mulase-Rabin--1991} using an alternate construction. See \cref{unsuper Gr} for a discussion of the classical Sato Grassmannian.

\subsection{Background material on finite dimensional super Grassmannians}

Much of what we describe in this section is based on \cite[Chapter 4 Section 3]{Manin--1988e}. A recent account of super Grassmannians can be found in \cite{Bruzzo-HernandezRuiperez-Polishchuk--2020X}. We provide this background to give context for the super Sato Grassmannian construction.

\begin{defn}
	
	Let $V$ be a super vector space. For $c|d\leq \dim(V)$, let $\Gr(c|d,V)$ be the functor that takes a superspace $S$ to the family of dimension $c|d$ super subbundles of $V\otimes\mathcal{O}_S$.

\end{defn}

The bosonization of $\Gr(c|d,V)$ is given by $\Gr(c,V_0)\times \Gr(d,V_1)$, the product of two classical Grassmannians.

It is known that super Grassmannians may not to superprojected. This is an interesting feature of supergeometry. Super Grassmannians seem to be important since most superspaces with projected bosonization embed into a super Grassmannian. \cite{Bruzzo-HernandezRuiperez-Polishchuk--2020X}

We now show that $\Gr(c|d,V)$ is representable. 

Let $L\in\Gr(c|d,V)$. Then to each $L$ we can associate a superdomain $U_L$ of the Grassmannian. Consider all possible morphisms 
\begin{align*}
	\begin{pmatrix}
		\id_L\\
		Z
	\end{pmatrix}\co L\to V
\end{align*}
where $Z\co L\to V/L$ may be written as a $(m-c|n-d) \times c|d$ matrix with entries from $\CC$. We define $U_L\subseteq \Gr(c|d,V)$ to be the superdomain with $Z$ representing its coordinates. We notice that $U_L\cong  \CC^{(m-c)c+(n-d)d|(m-c)d+(n-d)c}$.

We describe the gluing of these superdomains. Let $U_{L}$ and $U_{L'}$ be two superdomains as above. We write the identity map of $V$ as the block matrix
\begin{align*}
	\id_\CC=\begin{pmatrix}
		T^{L'L} & T^{L'K}\\
		T^{K'L} & T^{K'K}
	\end{pmatrix}
\end{align*}
where we denote $K\coeq V/L$ and $K'\coeq V/L'$, and $T^{L'K}\co K\to L'$, etc. Then we define the change of coordinates from $Z_L$ to $Z_{L'}$ as 
\begin{align}\label{change of coords finite}
	Z_{L'}=\left(T^{K'L}+T^{K'K}Z_L\right)\left(T^{L'L}+T^{L'K}Z_L\right)^{-1}.
\end{align}
Clearly this is only defined on the subdomain of $U_L$ where $T^{L'L}+T^{L'K}Z_L$ is invertible, which is an open condition. As well, this definition clearly gives transition functions which are rational functions, which are a part of the functions in any of the categories we are considering. One can check that this satisfies the cocycle condition on the intersection on three superdomains. 

The proof that this construction represents the functor $\Gr(c|d,V)$ is analogous to the proof in classical geometry. See \cite[Chapter 4 Section 3 Theorem 10]{Manin--1988e}. 

To elucidate the change of coordinates definition and to define a group action, consider  any automorphism $T$ of $V$ instead of the identity. A natural left action of the supergroup $\GL(V)\coeq \intAut(V)$ is defined by the matrix multiplication 
\begin{align*}
	\begin{pmatrix}
		T^{L'L} & T^{L'K}\\
		T^{K'L} & T^{K'K}
	\end{pmatrix} \begin{pmatrix}
		I \\
		Z _L
	\end{pmatrix}=\begin{pmatrix}
		I \\
		Z_{L'}
	\end{pmatrix} Q
\end{align*}
where $Q\co L\to L'$ is some isomorphism.

As for the classical Grassmannian, there is a tautological sheaf over the super Grassmannian. Over the superdomain $U_L$, we define the tautological bundle to be the trivial bundle over $U_L$ with supervector space the column space of $\left(\begin{smallmatrix}
	\id_L\\
	Z_L
\end{smallmatrix}\right)$. This free bundle of rank $c|d$ is then glued across superdomains $U_L$ using the same change of coordinates as in \cref{change of coords finite}.

\subsection{Definition} \label{Gr subsection}

Many of the definitions in this section are a generalization of the analogous definitions regarding the classical Sato Grassmannian by \textcite{AlvarezVazquez-MunozPorras-PlazaMartin--1998}.

Consider $\mathbb{C}(\!(z)\!)$, the infinite-dimensional topological space of formal Laurent series with the $z$-adic topology. Define the super vector space $H=\mathbb{C}(\!(z)\!)[\zeta]=\mathbb{C}(\!(z)\!)\zeta+\mathbb{C}(\!(z)\!)$ where $\zeta$ is a Grassmann variable, that is $\zeta^2=0$. 

In fact, for compatibility later on, we wish to consider $H_j=\mathbb{C}(\!(z)\!)[\zeta]\,[dz|d\zeta]^{\otimes j}$. As a topological space $H_j$ and $H$ are homeomorphic; and as super vector spaces, $H_j\cong H$ for $j$ an even integer, and  $H_j\cong\Pi H$ for $j$ an odd integer. We define the distinguished decomposition $H_j=H^-_j\oplus H^+_j$ by
\begin{align}\label{distinguished decomposition}
		H^-_j\coeq z^{-1}\mathbb{C}[z^{-1}|\zeta]\,[dz|d\zeta]^{\otimes j} && H^+_j\coeq \mathbb{C}[\![z]\!][\zeta]\,[dz|d\zeta]^{\otimes j}.
\end{align}
The Grassmannian we consider essentially parameterizes all super subspaces of $H_j$ close enough to $H^-_j$, which is defined precisely below.

Firstly, we note that two super subspaces $V,W\subseteq H_j$ are \emph{commensurable} if both $V/(V\cap W)$ and $W/(V\cap W)$ have both finite even dimension and finite odd dimension. We define the notation 
\begin{align*}
	V_S\coeq V\otimes_\mathbb{C} \mathcal{O}_S   && \hat{V}_S\coeq V\hat{\otimes}_\mathbb{C} \mathcal{O}_S=\lim\limits_{\longleftarrow}(V/K\otimes_\mathbb{C}\mathcal{O}_S)
\end{align*}
where $S$ is a superspace with structure sheaf $\mathcal{O}_S$. 

\begin{defn}
	Define a super subspace of $H_j$ to be \emph{compact} if it is commensurable with $H_j^+$.
	Define a super subspace $D$ of $H_j$ to be \emph{discrete} if there exists a compact subspace $K$ such that the natural map $D\oplus K\to H_j$ is an isomorphism of super vector spaces.\footnote{We have defined discrete and compact subspaces purely algebraically. The topological names are motivated by the $z$-adic topology on $H_j$.}
	
	More generally, for a superspace $S$ with structure sheaf $\mathcal{O}_S$, a super $\mathcal{O}_S$-submodule $L\subset (\hat{H}_j)_S$ is \emph{discrete} if 
	for every $s\in S$ there exists a neighborhood $U$ of $s$ and a compact $K$ such that the natural map $L_U\oplus  \hat{K}_U\to (\hat{H}_j)_U$ is an isomorphism, where $L_U\coeq L\otimes_{\mathcal{O}_S}\mathcal{O}_U$.
\end{defn}

\begin{lem}\label{CompactDiscreteDef} 
	A super $\mathcal{O}_S$-submodule $L\subset (\hat{H}_j)_S$ is discrete if and only if 
	for every $s\in S$ there exists a neighborhood $U$ of $s$ and a compact $K$ such that the kernel is a free finite type super $\mathcal{O}_S$-submodule and the cokernel is trivial in the exact sequence below
	\begin{align}\label{LK exact sequence}
		0\to L_U\cap \hat{K}_U\to L_U\oplus \hat{K}_U\to (\hat{H}_j)_U\to (\hat{H}_j)_U/(L_U+\hat{K}_U)\to 0
	\end{align}
	where $L_U\coeq L\otimes_{\mathcal{O}_S}\mathcal{O}_U$.
\end{lem}
\begin{proof}
 The proof is given in \cref{alt discrete def proof}. 
\end{proof}

\begin{defn}\label{GrDef}
	We define the (semi-infinite or Sato) super Grassmannian $\Gr(H_j)$ as the infinite-dimensional complex supermanifold representing the functor of the discrete super subspaces:
	\begin{align*}
		S \mapsto \left\{\text{discrete super }\mathcal{O}_S\text{-submodules } L\subseteq (\hat{H}_j)_S \right\}.
	\end{align*}

\end{defn}

The fact that this functor is representable by an infinite-dimensional complex supermanifold may be proven following the style of \cite{AlvarezVazquez-MunozPorras-PlazaMartin--1998} and \cite{MunozPorras-PlazaMartin--1999}, which prove the analogous nonsuper result for the infinite-dimensional Grassmannian of discrete subspaces in $\CC(\!(z)\!)$. We explain a modified version of this argument in the remainder of this section.

Observe that the underlying manifold of $\Gr(H_j)$ is $(\Gr(H_j))_{\text{red}}$ which is the set of points
\begin{align*}
	(\Gr(H_j))_{\text{red}}\cong\{\text{discrete super subspaces } D\subset H_j\}.
\end{align*}
We may identify this with the product of two classical Sato Grassmannians: $\Gr(\CC(\!(z)\!))\times \Gr(\Pi\zeta\CC(\!(z)\!))$. The super Sato Grassmannian therefore has $\ZZ\times \ZZ$ connected components, which we label with the virtual dimension: $m|n$. A discrete subspace $D\subset H_j$ is in the virtual dimension $m|n$ component if the Fredholm index of the natural map $D\oplus H_j^+\to H_j$ is $m|n$, explicitly, that is if $\dim(D\cap H_j^+)-\dim(H_j/(D+H_j^+))=m|n$.

\begin{defn}
	Given discrete $D$ and compact $K$ such that the natural map $D\oplus K\to H_j$ is an isomorphism, define an \emph{open chart $U_{D,K}$} as the subsupermanifold of the super Grassmannian representing the planes $L$ which project isomorphically onto $D$ along $K$. To make this precise using the functor of points, the supermanifold $U_{D,K}$ is 
	\begin{align*}
		S\mapsto \Big\{\text{discrete super }\mathcal{O}_S\text{-submodules } L\subseteq (\hat{H}_j)_S \textrm{ such that }
		L\cap\hat{K}_S=(\hat{H}_j)_S/(L+\hat{K}_S)=0\Big\}.
	\end{align*}
\end{defn}
Equivalently, $U_{D,K}$ is those $L$ such that $L\xrightarrow{p_D} D_S$ is an isomorphism.
Taking the inverse $D_S \to L$ of the isomorphism and composing it with the inclusion $L \hookrightarrow (\hat{H}_j)_S$, we get a map $A: D_S \to (\hat{H}_j)_S$. On the other hand, for any $\mathcal{O}_S$-module map $B\co D_S\to \hat{K}_S$, we have that the natural map $\textrm{graph}(B)\oplus \hat{K}_S\to (\hat{H}_j)_S$ is an isomorphism, so $L' = \textrm{graph}(B)$ is an $S$-point of $U_{D,K}$. We say that $L'$ is represented by the coordinates $B \co D_S \to \hat{K}_S$. 
We therefore identify the affine coordinate charts $U_{D,K}$ with the superspaces $\uline{\Hom}_\mathbb{C}(D,K)$. 
Where we use the notation $\uline\Hom_\mathbb{C}(D,K)$ for the corresponding (infinite dimensional) affine supermanifold in addition to the super vector space, by a slight abuse of notation.

We check that the functor $U_{D,K}$ does define an open chart of $\Gr(H_j)$. Consider an $S$-point $L$ of $\Gr(H_j)$. Taking the pullback, we have a morphism $S\times_{\Gr(H_j)} U_{D,K} \to U_{D,K}$. We need to show that this $S\times_{\Gr(H_j)} U_{D,K}$-point is represented by an open subspace of $S$.

Let $s\in S$ be a superpoint such that $L_s$ satisfies $L_s\cap\hat{K}_s=(\hat{H}_j)_s/(L_s+\hat{K}_s)=0$. Then since $(\hat{H}_j)_S/(L+\hat{K}_S)$ is locally of finite presentation by \cref{quotient finite pres}, we may apply the super Nakayama lemma \cite[Lemma 4.7.1 ii]{Varadarajan--2004} to conclude that there exists a neighborhood $U\ni s$ such that $(\hat{H}_j)_U/(L_U+\hat{K}_U)=0$. Similarly, since $L\cap\hat{K}_S$ is locally finitely generated, there exists $U'\ni s$ such that $L_{U'}\cap \hat{K}_{U'}=0$. Therefore the $S\times_{\Gr(H_j)} U_{D,K}$-point must be open, because we have shown that any $s$ contained in $U_{D,K}$ is contained in an open neighborhood of $U_{D,K}$.

The fact that the charts $U_{D,K}$ cover $\Gr(H_j)$ follows directly from their definition. Therefore, since each chart $U_{D,K}$ is representable as an infinite dimensional supermanifold, then $\Gr(H_j)$ is representable as well.

Lastly we describe the gluing of the charts $U_{D,K}$ following the style of \cite[Proposition (7.1.2)]{Pressley-Segal--1986}.
Consider two such charts: $U_{D,K}$ and $U_{D',K'}$. Let $I\subseteq \uline{\Hom}_\mathbb{C}(D,K)$ and $I'\subseteq\uline{\Hom}_\mathbb{C}(D',K')$ each correspond to $U_{D,K}\cap U_{D',K'}$. We wish to show $I$ and $I'$ are open and that the change of coordinates $I\to I'$ is holomorphic.

Let
\begin{align*}
	\begin{pmatrix}
		T^{D'D} & T^{D'K}\\
		T^{K'D} & T^{K'K}
	\end{pmatrix}
\end{align*}
be the identity map expressed as $D\oplus K \to D'\oplus K'$. Since $K'$ and $K$ are commensurable, then $T^{K'K}$ is Fredholm. 
Suppose $W=\textrm{graph}(A\co D\to K)=\textrm{graph}(B\co D'\to K')$, i.e. $W$ corresponds to a point in $U_{D,K}\cap U_{D',K'}$. Then we must have 
\begin{align*}
	\begin{pmatrix}
		T^{D'D} & T^{D'K}\\
		T^{K'D} & T^{K'K}
	\end{pmatrix} \begin{pmatrix}
		I \\
		A 
	\end{pmatrix}=\begin{pmatrix}
		I \\
		B
	\end{pmatrix} Q
\end{align*}
where $Q\co D\to D'$ is some isomorphism. We find
\begin{align*}
	B=(T^{K'D}+T^{K'K}A)(T^{D'D}+T^{D'K}A)^{-1}.
\end{align*}
So $B$ is a algebraic function of $A$ and $I=\{A\in\uline{\Hom}_\mathbb{C}(D,K)\co T^{D'D}+T^{D'K}A \textrm{ is invertible}\}$ and so $I$ is indeed open; and similarly for $I'$.

\subsection{Action of the general linear Lie group and algebra}

We describe the action of $\GL\big(\mathbb{C}(\!(z)\!)[\zeta]\,[dz|d\zeta]^{\otimes j}\big)$ and $\mathfrak{gl}\big(\mathbb{C}(\!(z)\!)[\zeta]\,[dz|d\zeta]^{\otimes j}\big)$ on the super Sato Grassmannian.

Firstly, we recall the general linear superalgebra and supergroup in the finite dimensional setting for reference.

For a finite dimensional supervector space $V$, the general linear superalgebra $\mathfrak{gl}(V)$ is defined by the internal Hom, denoted $\intHom(V,V)\eqco\intEnd(V)$. The even endomorphisms are those which are parity preserving, and the odd endomorphisms are those which are parity reversing. 

For a finite dimensional super vector space $V$, the general linear supergroup $\GL(V)$ is defined to be the superspace corresponding to the functor of points
\begin{align*}
	S\mapsto \Aut_{\mathcal{O}_S}(V\otimes\mathcal{O}_S).
\end{align*}
So $\GL(V)$ is the split supermanifold whose bosonization is the classical manifold $\GL(V_0)\times\GL(V_1)$ and which has odd dimension $2\dim(V_0)\cdot\dim(V_1)$. Group operations are given by the usual matrix operations.

We now return to our infinite dimensional super vector space $H_j$.
\begin{defn}
	Define the general linear superalgebra of $H_j$, denoted $\mathfrak{gl}(H_j)$, to be the subalgebra of endomorphisms $\intEnd(H_j)$ which are continuous with respect to the $z$-adic topology on $H_j$.
	\end{defn}
For any endomorphism of $H_j$, we may write it as
\begin{align*}
	F=\begin{pmatrix}
		F^{--} & F^{-+}\\
		F^{+-} & F^{++}
	\end{pmatrix}
\end{align*}
where $F^{+-}: H_j^-\to H_j^+$ etc., where $H_j^-$ and $H^+_j$ are defined in \cref{distinguished decomposition}. Define an endomorphism of a discrete subspace to be \emph{supertrace class} if it factors through some compact subspace.

We make the following definition in the style of \cite[Definition 4.1]{AlvarezVazquez-MunozPorras-PlazaMartin--1998}.
\begin{defn}
Define the general linear supergroup $\GL(H_j)$ to be the supergroup of homeomorphic (i.e. bicontinuous) linear isomorphisms. 
Precisely, this supergroup is given by the functor of points below.
\begin{align*}
	S\mapsto \Big\{G\in\Aut_{\mathcal{O}_S}(H_j\hat{\otimes}\mathcal{O}_S)\text{ such that for every $s\in S$ there exists a neighborhood $U\ni s$ and}\\\text{compact super subspaces $K',K''$ such that } G(\hat{K}_S)_U\cong \hat{K}'_U \text{ and } G^{-1}(\hat{K}_S)_U\cong \hat{K}''_U \Big\}
\end{align*}
\end{defn}

\begin{prop} \label{GL acts}
	The group $\GL(H_j)$ acts transitively on $\Gr(H_j)$, and the stabilizer of $D$ is $\left(\begin{smallmatrix}* & *\\0 &*\end{smallmatrix}\right)$ expressed in a decomposition $D\oplus K$.
\end{prop}
\begin{proof} 
	
	Consider $G\in\GL(H_j)$ and discrete $L\subset \hat{H}_S$. Consider $U\subset S$ and compact $K$ such that the natural map $L_U\oplus \hat{K}_U\to \hat{H}_U$ is an isomorphism. Then $G(L_U)\oplus G(\hat{K}_U)\cong \hat{H}_U$ since $G$ is an automorphism. Further, possibly shrinking $U$, by definition $G(\hat{K}_U)\cong \hat{K}'_U$. Therefore $G(L)$ is discrete by the definition of a discrete submodule.

	To show the transitivity of the action, first notice that every discrete subspace $D$ is the image of an operator $w=\left(\begin{smallmatrix}w_-\\w_+\end{smallmatrix}\right)\co H_j^-\to H_j$ such that $p_-\circ w=w_-$ is Fredholm. As well, notice that there exists a map $v=\left(\begin{smallmatrix}v_-\\v_+\end{smallmatrix}\right)\co H_j^+\to H_j$ where the columns of $v$ are a basis for $K$ such that the natural map $D\oplus K\to H_j$ is an isomorphism. Since $K$ and $H^+$ are commensurable, we may assume eventually $v(z^i)=z^{i-m}$ for some $m$ and $v(z^i\zeta)=z^{i-n}\zeta$ for some $n$. 
	Then we claim that 
	\begin{align*}
		G=\begin{pmatrix}
			w_- & v_-\\
			w_+ & v_+
		\end{pmatrix}
	\end{align*}
	is in $\GL(H_j)$ and $G(H_j^-)=D$. The fact that $G$ is invertible follows from assuming that $D\oplus K\to H$ is an isomorphism, and the bicontinuity of $G$ follows from the condition that eventually $v(z^i)=z^{i-m}$ for some $m$ and $v(z^i\zeta)=z^{i-n}\zeta$ for some $n$.
	
	Further, consider $L\in\hat{H}_S$. By \cref{quotient zero locally free kernel} and since $L$ is discrete, there exists a large enough compact $K$ such that $L\cap \hat{K}_S$ is locally free and $\hat{H}_S/(L+\hat{K}_S)=0$. Choose a discrete subspace $D$ such that the natural map $D\oplus K\to H$ is an isomorphism. Then $L$ is the image of an operator $w'=\left(\begin{smallmatrix}w'_D\\w'_K\end{smallmatrix}\right)\co D\to H_j$ such that $p_D\circ w'=w'_D$ has locally free finite type kernel and trivial cokernel.
	Then we claim that 
	\begin{align*}
		G'=\begin{pmatrix}
			w'_D & 0\\
			w'_K & \id_K
		\end{pmatrix}
	\end{align*}
	is in $\GL(H_j)$ and $G'(D)=L$.

	The stabilizer is obvious. 
\end{proof}

Therefore we may describe $\Gr(H_j)$ as the homogeneous superspace
\begin{align*}
	\Gr(H_j)\cong \GL(H_j)/P
\end{align*} 
where $P=\left(\begin{smallmatrix}* & *\\0 &*\end{smallmatrix}\right)$ in the $H_j^-\oplus H_j^+$ decomposition.

\begin{prop}\label{vector fields on Gr}
	The Lie superalgebra $\mathfrak{gl}(H_j)$ acts by vector fields on $\Gr(H_j)$. Explicitly, $F\mapsto L_F$ is a Lie superalgebra antihomomorphism $L\co\mathfrak{gl}(H_j)\to \Gamma(\Gr(H_j),\mathcal{T}_{\Gr(H_j)})$ sending $[F_1,F_2]$ to $[L_{F_2},L_{F_1}]$. In the chart $U_{D,K}$, this action is given by the formula
	\begin{align*}
		L_F(A)=F^{KD}+F^{KK}A-AF^{DD}-AF^{DK}A
	\end{align*} 
	where $L_F
	\in \uline{\Hom}_\mathbb{C}(S^\bullet(\uline{\Hom}_\mathbb{C}(D,K)),\uline{\Hom}_\mathbb{C}(D,K))$ acts on functions by 
	\begin{align*}
		L_F(f(A))=\lim_{\epsilon \to 0}\frac{f(A+\epsilon L_FA)-f(A)}{\epsilon}.
	\end{align*}
\end{prop}
\begin{proof}
	Consider the action of $G\in \GL(H_j)$ on the point in $U_{D,K}$ represented by coordinates $A\co D\to K$.
	\begin{align*}
		\begin{bmatrix}
			G^{DD} & G^{DK} \\ 
			G^{KD} & G^{KK}
		\end{bmatrix}\begin{bmatrix}
			I & 0\\ 
			A & I
		\end{bmatrix}
		&=	\begin{bmatrix}
			G^{DD}+G^{DK}A & G^{DK} \\ 
			G^{KD}+G^{KK}A & G^{KK}
		\end{bmatrix}\\
		&\sim\begin{bmatrix}
			I & 0 \\ 
			(G^{KD}+G^{KK}A)(G^{DD}+G^{DK}A)^{-1} & I
		\end{bmatrix}
	\end{align*}
	The equivalence relation is given by multiplication on the right by the stabilizer of $D$, which in the $D\oplus K$ decomposition is $\left(\begin{smallmatrix} * & * \\ 0 & *\end{smallmatrix}\right)$. We calculate $L_F(A)$ as the derivative at $t=0$ of the lower left block with $G=I+tF$.
\end{proof}

\subsection{The Berezinian line bundle}

The Berezinian line bundle on the super Grassmannian is a generalization to the super and infinite dimensional setting of the determinant line bundle.

\begin{lem}\label{perfect complex}
	Let $L$ be a discrete $\mathcal{O}_S$-module, and let $K$ be a compact subspace. Then the complex
	\begin{align*}
		0\to L\oplus \hat{K}_S\to (\hat{H}_j)_S \to 0
	\end{align*}
	is perfect, that is locally quasi-isomorphic to a bounded complex of free finite-type modules.
\end{lem}
\begin{proof}
	Let $s\in S$. There exists a neighborhood $U$ of $s$ and compact $K'$ such that 
	\begin{align*}
		K\subseteq K' && (\hat{H}_j)_U/(L_U+\hat{K}'_U)=0 && L_U\cap \hat{K}'_U \text{ is free of finite type}.
	\end{align*}
Then we have the exact sequence
\begin{align*}
	0\to L_U\cap \hat{K}_U\to L_U\cap \hat{K}'_U\to (K'/K)_U\to (\hat{H}_j)_U/(L_U+\hat{K}_U)\to 0.
\end{align*}
	Thus the original complex restricted to $U$ is quasi-isomorphic to $0\to L_U\cap \hat{K}'_U\to (K'/K)_U\to 0$, which is a bounded complex of free finite-type modules.
\end{proof}

For a discrete $L\in (\hat{H}_j)_S$, we define the $\mathcal{O}_S$-module
\begin{align}\label{Ber line def}
	\textstyle\Ber_K(L)\coeq \Ber(L\oplus \hat{K}_S\to (\hat{H}_j)_S)
\end{align}
which is well defined for any compact $K$ by \cref{perfect complex}. 

We now provide motivation for this choice of line bundle by observing the properties of the construction $Ber_K(L)$ over a finite Grassmannian $\Gr(c|d,V)$.  
For $L$, $K$, and $V=H_j$ of finite even and odd dimension, considering the exact sequence \cref{LK exact sequence} and taking its Berezinian, locally we have that $\Ber_K(L_{U})\cong\frac{\Ber(L_{U})\Ber(\hat{K}_U)}{\Ber(\hat{V}_U)}$. Since this finite-dimensional expression is $\Ber(L_U)$ multiplied by other constant\footnote{If we fix any compact $K$, for example $K=H_j^+$.} factors, it is an appropriate choice of the Berezinian line bundle over the $U$-point corresponding to $L$.

\begin{defn}
	The tautological bundle $I$ on $\Gr(H_j)$ is the $\mathcal{O}_{\Gr(H_j)}$-module which corresponds to $\Gr(H_j)\subseteq (\hat{H}_j)_{\Gr(H_j)}$.
\end{defn}
\begin{lem}
	The tautological bundle $I$ is a discrete $\mathcal{O}_{\Gr(H_j)}$-submodule.
\end{lem}
\begin{proof}
	Let $D\in\Gr(H_j)$ and consider the neighborhood $U_{D,K}$ for some compact $K$ such that the natural map $D\oplus K\to H_j$ is an isomorphism. Then by definition of $U_{D,K}\eqco U$, we have that the natural map $I_U\oplus \hat{K}_U\to (\hat{H}_j)_U$ is an isomorphism.
\end{proof}
In order to consider an analogous object to a determinant line bundle, we consider the line bundles $\Ber_K(I)$. The downside with this construction is that now our line bundle depends on a choice of compact subspace $K$. However, any two such Berezinian line bundles are isomorphic by 
\begin{align}\label{isom det}
	\textstyle \Ber_K(I)\cong \displaystyle\frac{\Ber(K/(K\cap K'))}{\Ber(K'/(K\cap K'))}\otimes\textstyle\Ber_{K'}(I).
\end{align} 

\begin{defn}\label{Ber on Gr def}
	In the spirit of \cite{Arbarello-DeConcini-Kac-Procesi--1988} and \cite{AlvarezVazquez-MunozPorras-PlazaMartin--1998}, we define the Berezinian line bundle on $\Gr(H_j)$ as the locally free sheaf given by
	\begin{align*}
		\mathcal{B\!e\!r}_{\Gr(H_j)}\coeq\textstyle\Ber_{H_j^+}(I).
	\end{align*} In other words, \cref{Ber line def} with the distinguished compact subspace $K=H_j^+$ and with the tautological discrete submodule $L=I\subset (\hat{H}_j)_{\Gr(H_j)}$.
\end{defn}
We can check this construction on the geometrical points $D\in \Gr(H_j)$. We have the fiber over $D$ given by
\begin{align*}
	\mathcal{B\!e\!r}_{\Gr(H_j)}(D)=\dfrac{\Ber(D\cap H_j^+)}{\Ber(H/(D+H_j^+))}
\end{align*}
which can be seen as a natural generalization of the determinant line bundle defined over the classical Sato Grassmannian in \cref{Gr det bundle}.

We see that $\mathcal{B\!e\!r}_{\Gr(H_j)}$ is of rank $1|0$ over the connected components of virtual dimension $m|2n$ and of rank $0|1$ over the connected components of virtual dimension $m|2n+1$, for $m,n\in \ZZ$.

\subsection{Action of the central extension}

It is known that $H^2(\mathfrak{gl}(H))$ is one-dimensional \cite[Section 2]{Kac-vandeLeur--1989}.
\begin{defn}\label{J cocycle}
For any choice of discrete $D$ and compact $K$ such that the natural map $D\oplus K\xrightarrow{\sim}H_j$ is an isomorphism, we define a 2-cocycle on $\mathfrak{gl}(H_j)$ as
\begin{align*}
	\eta_{D,K}(F,G):=\str(F^{DK}G^{KD}-(-1)^{|F||G|}F^{DK}G^{KD}).
\end{align*}
We choose a distinguished 2-cocycle to correspond the distinguished decomposition $H_j^-\oplus H_j^+\cong H_j$:\footnote{The cocycle defined by \textcite{Ueno-Yamada--1988} is $\eta_{D,K}$ with choice $D=\mathbb{C}[z^{-1}|\zeta]$ and $K= z\mathbb{C}[\![z]\!][\zeta]$.}
\begin{align*}
	\eta(F,G):=
	\str(F^{-+}G^{+-}-(-1)^{|F||G|}F^{-+}G^{+-}),
\end{align*}
which we call the super Japanese cocycle. 
The unique Lie superalgebra central extension defined by the super Japanese cocycle will be denoted $\widetilde{\mathfrak{gl}}(H_j)$.
\end{defn}
We denote the bracket on $\widetilde{\mathfrak{gl}}(H_j)$ as $[F,G]^{\sim}=[F,G]+\eta(F,G)$.

\begin{prop}\label{diff operators on det} 
	The Lie superalgebra $\widetilde{\mathfrak{gl}}(H_j)$ 
	acts by first order differential operators on ${\mathcal{B\!e\!r}}_{\Gr(H_j)}$. 
	Explicitly, $F+c\mapsto \widetilde{L}_F+c$ is a Lie superalgebra antihomomorphism $\tilde{L}\co\widetilde{\mathfrak{gl}}(H_j)\to \Gamma(\Gr(H_j),\mathcal{A}_{\mathcal{B\!e\!r}})$ 
	sending $[F_1,F_2]^\sim$ to $[\widetilde{L}_{F_2},\widetilde{L}_{F_1}]$. In 
	the chart $U_{D,K}$, 
	this action is given by the formula
	\begin{align*}
		\widetilde{L}_F(A)=L_F(A)+\str(F^{DK}A)+\alpha(F)
	\end{align*} 
	where $\alpha\in C^1(\mathfrak{gl})$ is the unique 1-cochain such that
	\begin{align*}
		d\alpha(F,G)=\alpha([F,G])=\eta_{D,K}(F,G)-\eta(F,G).
	\end{align*}
\end{prop}

Before providing the proof, we first describe the case that $D,K,H_j$ were finite dimensional as motivation for the definition of the Lie superalgebra action in the infinite dimensional case. For this paragraph only, assume $H_j$ is finite dimensional, so therefore ${\mathcal{B\!e\!r}}_{\Gr(H_j)}\cong\Ber(D)\Ber(K)/\Ber(H_j)$.
Consider the natural action of $G\in \GL(H_j)$ on $\Ber_K(D_A)$ where $D_A\in U_{D,K}$ is represented by coordinates $A\co D\to K$. 
\begin{equation*}
	\begin{tikzcd}[column sep=1.5in]
		\dfrac{\Ber(D)\Ber(K)}{\Ber(H_j)} 
		\arrow{r}{\frac{\Ber(G^{DD}+G^{DK}A)\Ber(G^{KK})}{\Ber(G)}} \arrow{d}{\textrm{graph}(A)} &
		\dfrac{\Ber(D)\Ber(K)}{\Ber(H_j)} 
		\\
		\dfrac{\Ber(D_A)\Ber(K)}{\Ber(H_j)}
		\arrow{r}{{G}}&
		\dfrac{\Ber(D_A)\Ber(K)}{\Ber(H_j)} 
		\arrow{u}{{\pi_D,\pi_K}}
	\end{tikzcd}
\end{equation*}
Using the canonical isomorphism of the fiber $\Ber_K(D_A)$ with $\Ber_K(D)$, we find the multiplicative factor $\frac{\Ber(G^{DD}+G^{DK}A)\Ber(G^{KK})}{\Ber(G)}$. Then for $G=I+tF$, we derive the Lie superalgebra action of $\str(F^{DK}A)$.

\begin{proof}
	Direct computation shows that $\str([F_1,F_2]^{DK}A)=\eta_{D_A,K}(F_1,F_2)-\eta_{D,K}(F_1,F_2)$. Similarly, if $D_A\in U_{D,K}$ and $D_A\in U_{D',K'}$ then the change of coordinates is given by the unique 1-cochain whose differential is $\eta_{D_A,K'}(F_1,F_2)-\eta_{D_A,K}(F_1,F_2)$. Thus, this Lie superalgebra action glues between charts.
	
	The Lie superalgebra antihomomorphism follows from:
	\begin{align*}
		[\widetilde{L}_{F_1},\widetilde{L}_{F_2}]A
		&= [L_{F_1},L_{F_2}]A+\str(F_2^{DK}L_{F_1}A)-\str(F_1^{DK}L_{F_2}A) \\
		&=  L_{[F_2,F_1]}A+\eta_{D,K}(F_2,F_1)+\str([F_2,F_1]^{DK}A)\\
		&=L_{[F_2,F_1]}A+\eta_{D,K}(F_2,F_1)-\eta(F_2,F_1)+\eta(F_2,F_1)+\str([F_2,F_1]^{DK}A)\\
		&=L_{[F_2,F_1]}A+\alpha([F_2,F_1])+\eta(F_2,F_1)+\str([F_2,F_1]^{DK}A)\\
		&=\widetilde{L}_{[F_2,F_1]^\sim} A
	\end{align*}
\end{proof}

We now summarize the above Lie superalgebra action using action Lie superalgebroids. We denote by $\mathfrak{g}^\text{op}$ the opposite Lie superalgebra, that is the Lie superalgebra with the negative bracket of $\mathfrak{g}$, that is $[x,y]_\mathfrak{g}=-[x,y]_{\mathfrak{g}^\text{op}}$. We denote the Lie superalgebra antihomomorphism between these as $\text{op}\co \mathfrak{g}\to \mathfrak{g}^\text{op}$.

\begin{defn}\label{gl superalgebroid defs}
	According to \cref{action superalgebroid def}, define $(\mathcal{G}_j,a_L)$ to be the action Lie superalgebroid associated to $L\circ \text{op}\co \mathfrak{gl}(H_j)^\text{op}\to \Gamma(\Gr(H_j),\mathcal{T}_{\Gr(H_j)})$ as in \cref{vector fields on Gr}. 
	
	Similarly, define $\widetilde{\mathcal{G}}_j$ to be the action Lie superalgebroid associated to $\sym\circ\tilde{L}\circ\text{op}\co\widetilde{\mathfrak{gl}}(H_j)^\text{op}\to \Gamma(\Gr(H_j),\mathcal{T}_{\Gr(H_j)})$ as in \cref{diff operators on det}.
\end{defn}

By the correspondence in \cref{Lie alg algoid rep corrsp}, the Lie superalgebra action $\tilde{L}\circ\text{op}\co\widetilde{\mathfrak{gl}}(H_j)^\text{op}\to \Gamma(\Gr(H_j),\mathcal{A}_{\mathcal{B\!e\!r}})$ may be used to define a morphism of Lie superalgebroids $b_{\tilde{L}}\co\widetilde{\mathcal{G}}_j\to\mathcal{A}_{\mathcal{B\!e\!r}}$.

In summary, we have the commutative diagram of Lie superalgebroids below.
\begin{equation} \label{gl tilde action square}
	\begin{tikzcd}
		0 \arrow{r} & \mathcal{O}_{\Gr(H_j)} \arrow{r}\arrow{d}{\id} & \widetilde{\mathcal{G}}_j \arrow{r}\arrow{d}{b_{\widetilde{L}}} & \mathcal{G}_j\arrow{r}\arrow{d}{a_L} & 0\\
		0\arrow{r} & \mathcal{O}_{\Gr(H_j)} \arrow{r} & \mathcal{A}_{\mathcal{B\!e\!r}} \arrow{r} & \mathcal{T}_{\Gr(H_j)} \arrow{r} & 0
	\end{tikzcd}
\end{equation}

\section{A flat holomorphic connection via the super Krichever map}\label{flat section}

The super Krichever map has been discussed in \cite{Mulase--1991,Mulase-Rabin--1991}. A flat connection essentially the same as the flat connection in this paper has been described in \cite{Manin--1988}. The relationship between the super Krichever map and the flat connection is new. See \cref{unsuper krichever} for a discussion of the classical Krichever map and flat connection.

\subsection{Definition of the super Krichever map}

The super Krichever map has been defined for split super Riemann surfaces by \textcite[(2.5)]{Mulase-Rabin--1991}.

Consider the super vector subspace of sections $\Gamma(\Sigma\setminus p,\omega_\Sigma^{\otimes j})\subset \mathbb{C}(\!(z)\!)[\zeta]\,[dz|d\zeta]^{\otimes j}=H_j$.   

\begin{defn}[{\textcite[(2.5)]{Mulase-Rabin--1991}}]
	The super Krichever map $\left(\mathfrak{M}_{g,1^\infty_{\mathrm{NS}}}\right)_{\red} \to \big(\Gr(H_j)\big)_{\red}$ is given by
	\begin{align*}
		\kappa_j(\Sigma,p,z|\zeta)= \Gamma(\Sigma\setminus p,\omega_\Sigma^{\otimes j})\subset H_j
	\end{align*}
	where $\omega_\Sigma\coeq \Ber(\Omega_\Sigma^1)$ is rank $0|1$. 
\end{defn}
See \cite[Theorem 4.2]{Mulase-Rabin--1991} for a proof that the map above is injective analytic.

We consider this map extended to families of SRSs.  
\begin{prop}
Consider a family of super Riemann surfaces $\pi\co X\to S$ with NS punctures given by the divisor $P$. The sheaf $\pi_*\omega_{(X\setminus P)/S}^{\otimes j}$  is a discrete $\mathcal{O}_S$-submodule.
\end{prop}
\begin{proof}
	We follow the style of \cite[proof of Proposition 6.3]{MunozPorras-PlazaMartin--1999}.
	
	Define the compact subspaces $K^n_j\coeq z^{-n}\CC[[z]][\zeta]\,[dz|d\zeta]^{\otimes j}\subset H_j$. Notice $(\hat{H}_j)_S/(\hat{K}_j)^n_S\cong (H_j/K_j^n)_S$ is a discrete $\mathcal{O}_S$-submodule.
	
	Consider formal superconformal coordinates $(z|\zeta)$ around the NS puncture divisor $P$. Then we may identify $\pi_*\omega_{(X\setminus P)/S}^{\otimes j}\eqco L\subset (\hat{H}_j)_S$.
We have the exact sequence 
\begin{align*}
	0\to L\cap (\hat{K}^n_j)_S\to L\to (\hat{H}_j)_S/(\hat{K}^n_j)_S\to (\hat{H}_j)_S/(L+(\hat{K}^n_j)_S).
\end{align*}
  Using relative \v{C}ech cohomology, where $U$ is a formal tubular neighborhood of $P$, we can identify the previous exact sequence as the (twisted higher) direct image sheaves as below. 
\begin{align*}
	0\to \pi_*\omega_{X/S}^{\otimes j}(nP)\to \pi_*\omega_{(X\setminus P)/S}^{\otimes j}\to \pi_*\omega^{\otimes j}_{(U\setminus P)/S}\Big/\pi_*\omega^{\otimes j}_{U/S}(nP)\to R^1\pi_*\omega_{X/S}^{\otimes j}(nP)
\end{align*}
We note that we have used $\pi_*\left(\omega^{\otimes j}_{X/S}(mP)\big/\omega^{\otimes j}_{X/S}(nP)\right)\cong \pi_*\omega^{\otimes j}_{U/S}(mP)\big/\pi_*\omega^{\otimes j}_{U/S}(nP)$ and that $\lim\limits_{\longrightarrow}\pi_*\omega^{\otimes j}_{U/S}(mP)\cong \pi_*\omega^{\otimes j}_{(U\setminus P)/S}$.

Then it suffices to notice that for large enough twisting $n$, then locally  $R^1\pi_*\omega_{X/S}^{\otimes j}(nP)=0$ and $\pi_*\omega_{X/S}^{\otimes j}(nP)$ is locally free and finitely generated. The result follows from \cref{CompactDiscreteDef}.

\end{proof}

We therefore have a well-defined super Krichever map that maps to the super Sato Grassmannian.
\begin{defn}\label{super Krichever}
The super Krichever map $\mathfrak{M}_{g,1^\infty_{\mathrm{NS}}} \to \Gr(H_j)$ is the map given by
\begin{align*}
 \kappa_j(X/S,P,z|\zeta)=\pi_*\omega_{(X\setminus P)/S}^{\otimes j}\subset
 (\hat{H}_j)_S
\end{align*}
where $\omega_{(X\setminus P)/S}\coeq \Ber\left(\Omega_{(X\setminus P)/S}^1\right)$.
\end{defn}
Since $\omega_{(X\setminus P)/S}\coeq \Ber\left(\Omega_{(X\setminus P)/S}^1\right)$ is rank $0|1$, we note that for $j$ an even integer, the super Krichever map $\kappa_j$ has image in the $\big((j-1)(g-1)\big|j(g-1)\big)$ virtual dimension component of $\Gr(H_j)$, whereas for $j$ an odd integer, the super Krichever map $\kappa_j$ has image in the $\big(j(g-1)\big|(j-1)(g-1)\big)$ virtual dimension component of $\Gr(H_j)$. This follows from super Riemann-Roch, which states that the Euler characteristic $h^0-h^1$ of a rank $1|0$ sheaf $\mathcal{L}$ is $(\deg\mathcal{L}-g+1|\deg\mathcal{L})$ and of a rank $0|1$ sheaf $\mathcal{L}$ is $(\deg\mathcal{L}|\deg\mathcal{L}-g+1)$.

\color{black}

\subsection{Representations of $\mathfrak{switt}$ and $\mathfrak{ns}$ on $\mathbb{C}(\!(z)\!)[\zeta]\,[dz|d\zeta]^{\otimes j}$}

Consider $\pi\co X\to S$ a family of SRSs. Superconformal vector fields act on the sections of $\omega_{(X/S)}^{\otimes j}$ by Lie derivative \cite[Section 3.11]{Deligne-Morgan--1999}
\begin{align}\label{Lie derivative action}
	\mathcal{L}_{[fD_\zeta,D_\zeta]}\big(g\;[dz|d\zeta]^{\otimes j}\big)=\bigg(\Big([fD_\zeta,D_\zeta]\,g\Big)+\frac{j}{2}\;\frac{\partial f}{\partial z}\;g\bigg)\;[dz|d\zeta]^{\otimes j}
\end{align} 
where $f,g\in \Gamma(X,\mathcal{O}_X)$ and $(z|\zeta)$ are local superconformal coordinates.
Further, for $U$ a formal neighborhood of a point $p\in\CC^{1|1}$, we see that $\Gamma(U\setminus p,\omega^{\otimes j}_{\CC^{1|1}})$ is identified with $\CC(\!(z)\!)[\zeta]\,[dz|d\zeta]^{\otimes j}=H_j$, and $\Gamma(U\setminus p,\mathcal{T}^s_{\CC^{1|1}})$ is identified with $\mathfrak{switt}$.

\begin{defn}
Define the Lie superalgebra morphism $\varrho_j: \mathfrak{switt}\to \mathfrak{gl}(H_j)$ as the natural Lie derivative action of $j/2$-differentials.
\begin{align*}
\varrho_j\big([fD_\zeta,D_\zeta]\big)\coeq \mathcal{L}_{[fD_\zeta,D_\zeta]}
\end{align*}
\end{defn}

We remark to compare our definition above with the definition of \textcite{Ueno-Yamada--1988}.  
\textcite{Ueno-Yamada--1988} define the representations $[fD_\zeta,D_\zeta]\mapsto \left(g\mapsto \Big([fD_\zeta,D_\zeta]\,g\Big)+\frac{j}{2}\;\frac{\partial f}{\partial z}\;g\right)$ without the factor $[dz|d\zeta]^{\otimes j}$, and therefore are representations on $\mathfrak{gl}(H)$ for every $j$.

\begin{prop}[cf. {\textcite[Proposition 4]{Ueno-Yamada--1988}}]
The pullbacks of the super Japanese cocycle $\eta$ as in \cref{J cocycle} along the representations $\varrho_j$ satisfies:
\begin{align*}
\varrho_j^*(\eta)=-(-1)^{j}(2j-1)\varrho_1^*(\eta) &&\varrho_j^*(\eta)=\varrho_{1-j}^*(\eta).
\end{align*}

\end{prop}
\begin{proof}
The proof may be done by direct computation. This result 
was stated in \cite[Proposition 4]{Ueno-Yamada--1988} with the multiplicative factor $(2j-1)$ instead of $-(-1)^{j}(2j-1)$. Further, using the $\Pi$ transpose as in \cref{Pi transpose}, we have isomorphisms
\begin{align*}
	\mathfrak{gl}(H_j)=\begin{dcases*}
		\mathfrak{gl}(H)& for $j$ an even integer\\
		\big(\mathfrak{gl}(H)\big)^\Pi& for $j$ an odd integer
	\end{dcases*}
\end{align*}
which may be used to complete the computation. 
Since the supertrace is defined as the trace of the even-even block minus the trace of the odd-odd block, we have that $\str(F^\Pi)=-(-1)^{|F|}\str(F)$ \cite[165]{Manin--1988e}. 
\end{proof}

It is known that $H^2(\mathfrak{switt})$ is one dimensional \cite[Section 4]{Kac-vandeLeur--1989}.
\begin{defn}
Let $c_j=-(-1)^{j}(2j-1)$. Define $\mathfrak{ns}_j$ to be the central extension of $\mathfrak{switt}$ by $\varrho^*_j(\eta)=c_j\varrho^*_0(\eta)$. When $j=0$ or $j=1$, we recover the standard definition of the Neveu-Schwarz algebra, denoted $\mathfrak{ns}$. On the basis of $\mathfrak{switt}$ the standard cocycle is given by
\begin{align*}
\varrho^*_0(\eta)(L_m,L_n)=\frac{m^3-m}{4}\delta_{m+n,0} && 
\varrho^*_0(\eta)(L_n,G_r)=0&&
\varrho^*_0(\eta)(G_r,G_s)=\frac{4r^2-1}{4}\delta_{r+s,0}
\end{align*}
\end{defn}
 The following commutative diagram of Lie superalgebras summarizes the relationship between them.
\begin{equation}\label{nsj}
\begin{tikzcd}
0 \arrow{r} & \mathbb{C} \arrow{r}\arrow{d}{\cdot c_j} & \mathfrak{ns} \ar{r}\arrow{d}{\wr} & \mathfrak{switt} \arrow{r}\arrow{d}{\id} & 0\\
 0 \arrow{r} & \mathbb{C} \arrow{r}\arrow{d}{\id} & \mathfrak{ns}_j \arrow{r}\arrow[hook]{d}\arrow[ very near start, phantom]{rd}{\lrcorner} & \mathfrak{switt} \arrow{r}\arrow[hook]{d}{\varrho_j} & 0\\
 0 \arrow{r} & \mathbb{C} \arrow{r} & \widetilde{\mathfrak{gl}}(H_j) \arrow{r} & \mathfrak{gl}(H_j) \arrow{r} & 0
\end{tikzcd}
\end{equation}

\subsection{Compatibility of the $\mathfrak{switt}$ representation and the super Krichever map}\label{compatibility}

\begin{prop}\label{krichever rep compat} 
	The representation $\varrho_j$ and the super Krichever map $\kappa_j$ are compatible, that is the diagram below of Lie superalgebras commutes.
	\begin{equation*}
		\begin{tikzcd}
			\mathfrak{switt} \arrow{d}{\varrho_j}\arrow{r}{\Lambda} & \Gamma(\mathfrak{M}_{g,1^\infty_{\mathrm{NS}}},\mathcal{T}_{\mathfrak{M}})\arrow{d}{d\kappa_j} \\
			\mathfrak{gl}(H_j)  \arrow{r}{L}& \Gamma(\mathfrak{M}_{g,1^\infty_{\mathrm{NS}}},\kappa_j^*\mathcal{T}_{\Gr(H_j)})  
		\end{tikzcd}
	\end{equation*}
\end{prop}

\begin{proof}
	The representation $\varrho_j$ is based on the Lie derivative of $j/2$-differentials, while $d\kappa_j$ is also given by Lie derivative action of $j/2$-differentials. So essentially, the diagram is commutative by construction.
\end{proof}

Recall \cref{gl superalgebroid defs} of $(\mathcal{G}_j,a_L)$ and \cref{switt superalgebroid defs} of $(\mathcal{W},a_\Lambda)$. 
 Then \cref{krichever rep compat} implies the existence of a morphism of these Lie superalgebroids on $\mathfrak{M}_{g,1^\infty_{\mathrm{NS}}}$ as below. 
\begin{equation}\label{gl witt algebroid diagram}
b_{\varrho_j}\co	 \mathcal{W} \to
 {\kappa_j}^!\mathcal{G}_j
\end{equation}

\subsection{Action of the Neveu-Schwarz algebra}

We now combine the $L$ and $\widetilde{L}$ action as in \cref{gl tilde action square} with the $\varrho_j$ compatibility morphism \cref{gl witt algebroid diagram}. This allows us to derive a canonical action of the Neveu-Schwarz Lie superalgebra. The general strategy may be summarized as pulling back the diagram \cref{gl tilde action square} along the Krichever map to $\mathfrak{M}_{g,1^\infty_{\mathrm{NS}}}$, and further restricting along $b_{\varrho_j}$ to $\mathcal{W}$.

We work with the action Lie superalgebroids $\mathcal{W}$, $\mathcal{G}_j$, $\tilde{\mathcal{G}}_j$ associated to the action of $\mathfrak{switt}$ on $\mathfrak{M}_{,1^\infty_{\mathrm{NS}}}$ and $\mathfrak{gl}(H_j)$ and $\tilde{\mathfrak{gl}}(H_j)$ on $\Gr(H_j)$ respectively, as in \cref{switt superalgebroid defs} and \cref{gl superalgebroid defs}.
Since $\widetilde{\mathcal{G}}_j$, $\mathcal{G}_j$, and $\mathcal{A}_{\mathcal{B\!e\!r}}$ are all transitive, we can consider their pullbacks along $\kappa_j$. Further by \cref{pullback morphism prop}, the morphisms between them also pullback.

\begin{equation*}
	\begin{tikzcd}
		0\arrow{r}& \mathcal{O}_{\mathfrak{M}}\arrow{d}{\id}\arrow[hook]{r}  &{\kappa_j}^!\widetilde{\mathcal{G}}_j \arrow[two heads]{d}{\kappa_j^!b_{\widetilde{L}}}\arrow[two heads]{r}   & {\kappa_j}^!\mathcal{G}_j\arrow[two heads]{d}{\kappa_j^!a_L} \arrow{r}&0 \\
		0\arrow{r}&\mathcal{O}_{\mathfrak{M}}\arrow[hook]{r} 
		& {\kappa_j}^!\mathcal{A}_{\mathcal{B\!e\!r}} \arrow[two heads]{r} & \mathcal{T}_{\mathfrak{M}} \arrow{r}&0
	\end{tikzcd}
\end{equation*}

\begin{cor} 
	There is an isomorphism of Lie superalgebroids $\mathcal{A}_{{{\kappa_j}^*}\mathcal{B\!e\!r}} \cong {\kappa_j}^!\mathcal{A}_{\mathcal{B\!e\!r}} $.
\end{cor}
\begin{proof}
	Apply \cref{pullback of atiyah prop}.
\end{proof}

\begin{defn}\label{nsj superalgebroid defs}
	According to \cref{action superalgebroid def}, define $\mathcal{N}_j$ to be the action Lie superalgebroid associated to the composition of $\text{op}\co \mathfrak{ns}_j^\text{op}\to\mathfrak{ns}_j$, the projection $\mathfrak{ns}_j\to \mathfrak{switt}$ as in \cref{nsj}, and $\Lambda\co \mathfrak{switt}\to \Gamma(\mathfrak{M}_{g,1^\infty_{\mathrm{NS}}},\mathcal{T}_{\mathfrak{M}})$ as in \cref{switt action}.

		Similarly, define $\mathcal{N}$ for $\mathfrak{ns}$.
\end{defn}

Then the diagram \cref{nsj} gives the diagram below of Lie superalgebroids on $\mathfrak{M}_{g,1^\infty_{\mathrm{NS}}}$.
\begin{equation*}
	\begin{tikzcd}
		0 \arrow{r} & \mathcal{O}_\mathfrak{M} \arrow{r}\arrow{d}{\cdot c_j} & \mathcal{N} \ar{r}\arrow{d}{\wr} & \mathcal{W}\arrow{r}\arrow{d}{\id} & 0\\
		0 \arrow{r} & \mathcal{O}_\mathfrak{M} \arrow{r}\arrow{d}{\id} & \mathcal{N}_j \arrow{r}\arrow[hook]{d}\arrow[ very near start, phantom]{rd}{\lrcorner} & \mathcal{W} \arrow{r}\arrow[hook]{d}{b_{\varrho_j}} & 0\\
		0 \arrow{r} & \mathcal{O}_\mathfrak{M} \arrow{r} & {\kappa_j}^!\widetilde{\mathcal{G}}_j \arrow{r} & {\kappa_j}^!\mathcal{G}_j \arrow{r} & 0
	\end{tikzcd}
\end{equation*}
Thus, we arrive at the main diagram below showing the action of $\mathcal{N}$ on an Atiyah Lie superalgebroid.
The crucial property of the action of the Neveu-Schwarz Lie superalgebroid $\mathcal{N}$ is that the central charge acts by $c_j$.
\begin{equation}\label{ns action diagram}
	\begin{tikzcd}
		\mathcal{O}_\mathfrak{M} \arrow{r}\arrow{d}{\cdot c_j} & \mathcal{N} \arrow{r}\arrow[hook]{d} & \mathcal{W} \arrow[hook]{d}{b_{\varrho_j}}\arrow[two heads,bend left=80]{dd}{a_\Lambda} \\
		\mathcal{O}_\mathfrak{M}\arrow{d}{\id}\arrow{r}&{\kappa_j}^!\widetilde{\mathcal{G}}_j \arrow[two heads]{d}{\kappa_j^!b_{\widetilde{L}}}\arrow[two heads]{r}   & {\kappa_j}^!\mathcal{G}_j\arrow[two heads]{d}{\kappa_j^!a_L} \\
		\mathcal{O}_\mathfrak{M}\arrow{r}&\mathcal{A}_{{{\kappa_j}^*}\mathcal{B\!e\!r}} \arrow[two heads]{r} & \mathcal{T}_{\mathfrak{M}} 
	\end{tikzcd}
\end{equation}

\begin{prop}\label{ber lambda isom}
	The pullback of the Berezinian line bundle from the super Grassmannian along the $j^\textrm{th}$ super Krichever map is canonically isomorphic to $\lambda_{j/2}$.
	\begin{equation*}
		{\kappa_j}^*\mathcal{B\!e\!r}_{\Gr(H_j)}\cong \lambda_{j/2}
	\end{equation*}
\end{prop}
\begin{proof}
	We see $\kappa_j(X/S,p,z|\zeta)=\pi_*\omega_{X/S}^{\otimes j}$, which has virtual dimension $m|2n$ for some $m,n\in \ZZ$ in $\Gr(H_j)$. Using \cref{Ber on Gr def} of the Berezinian line bundle on the super Grassmannian, we note that $\mathcal{B\!e\!r}_{\Gr(H_j)}$ is rank $1|0$ over virtual dimension $m|2n$ components, and further we have 
	\begin{align*}
		\mathcal{B\!e\!r}_{\Gr(H_j)}\Big|_{\textrm{Im}(\kappa_j(X/\mathfrak{M}))}
		&=\Ber\bigg(I\oplus (\hat{H}_j^+)_{\Gr(H_j)}\to (\hat{H}_j)_{\Gr(H_j)}\bigg)\bigg|_{\textrm{Im}(\kappa_j(X/\mathfrak{M}))}\\
		&=\Ber\bigg(\pi_*\omega_{X/\mathfrak{M}}^{\otimes j}\oplus (\hat{H}^+_j)_{X/\mathfrak{M}}\to (\hat{H}_j)_{X/\mathfrak{M}}\bigg)\\
		&=\Ber R^0\pi_*{\omega_{X/\mathfrak{M}}^j}\bigg/\Ber R^1\pi_*\omega_{X/\mathfrak{M}}^j
	\end{align*}
	which can be seen using relative \v{C}ech cohomology. This is exactly \cref{lambda j/2 def} of $\lambda_{j/2}$.
\end{proof}

\begin{prop}\label{ns action}
	The Neveu-Schwarz superalgebra acts by first order differential operators on $\lambda_{j/2}$ with central charge $c_j$. Precisely, there is a Lie superalgebra antihomomorphism $\mathfrak{ns}\to \Gamma(\mathfrak{M}_{g,1^\infty_{\mathrm{NS}}},\mathcal{A}_{\lambda_{j/2}})$ which is compatible with the action $\Lambda$ of $\mathfrak{switt}$ given in \cref{switt action}.
	\end{prop}
\begin{proof}
	The Lie superalgebra map follows from \cref{Lie alg algoid rep corrsp} applied to the Lie superalgebroid morphism $\mathcal{N}\to\mathcal{A}_{{{\kappa_j}^*}\mathcal{B\!e\!r}}$ in \cref{ns action diagram} combined with the isomorphism of \cref{ber lambda isom}.
\end{proof}

\subsection{A flat holomorphic connection}\label{flat subsection}

Using the properties of Atiyah algebras described in \cref{Atiyah section}, we describe the action of the Neveu-Schwarz Lie superalgebroid on the line bundles of the super Mumford isomorphism.  

Consider the Atiyah superalgebra of the line bundle $\lambda_{j/2}\otimes \lambda_{1/2}^{-c_j}$:
\begin{align*}
	\mathcal{A}_{j}\coeq \mathcal{A}_{\lambda_{j/2}\otimes \lambda_{1/2}^{-c_j}}.
\end{align*} 
From \cref{ns action}, $\mathcal{N}$ acts on $\mathcal{A}_{\lambda_{j/2}}$ with central charge $c_j$, and further by \cref{Atiyah defs} combined with \cref{Atiyah defs iso lemma}, $\mathcal{N}$ acts on $-c_j\mathcal{A}_{\lambda_{1/2}}$ with central charge $-c_j$. Then the action of $\mathcal{N}$ on $\mathcal{A}_{j}$ is defined via the action on $\mathcal{A}_{\lambda_{j/2}}\times_{\mathcal{T}_X} -c_j\mathcal{A}_{\lambda_{1/2}}$. Importantly, we see that $\mathcal{N}$ acts with central charge $c_j-c_j=0$, as shown in the commutative diagram of Lie superalgebroids below.
\begin{equation}\label{central charge cancel action diagram}
\begin{tikzcd}
	&&&0\arrow{d}&\\
 & &  & \mathcal{K}\ar{d} & \\
 0 \arrow{r} & \mathcal{O}_{\mathfrak{M}} \arrow{r}\arrow{d}{}{0} & \mathcal{N} \arrow{r}\arrow{d} & \mathcal{W} \arrow{r}\arrow{d}{a_\Lambda}\arrow[dashed]{dl}{\alpha_j}  & 0\\
 0\arrow{r} & \mathcal{O}_{\mathfrak{M}} \arrow{r} & \mathcal{A}_{j} \arrow{r} & \mathcal{T}_\mathfrak{M}  \arrow{r}\arrow{d} & 0\\
 &&& 0 
\end{tikzcd}
\end{equation}
We remind the reader that $\mathcal{N}$ and $\mathcal{W}$ are defined as the action Lie superalgebroids of $\mathfrak{ns}$ and $\mathfrak{switt}$ on $\mathfrak{M}_{g,1^\infty_{\mathrm{NS}}}$ respectively, according to \cref{action superalgebroid def}.
\begin{lem}\label{aplha map lemma}
	There exists a unique morphism of Lie superalgebroids $\alpha_j\co \mathcal{W}\to \mathcal{A}_{j}$ which commutes with the Neveu-Schwarz action $\mathcal{N}\to \mathcal{A}_j$ and the projection $\mathcal{N}\to\mathcal{W}$. In other words, the Neveu-Schwarz action on $\mathcal{A}_j$ descends to a super Witt action.
	\end{lem}
\begin{proof}
The proof is a straightforward diagram chase.

Let $w\in\Gamma(U,\mathcal{W})$ for some open $U$. Then there exists an element $n\in\Gamma(U,\mathcal{N})$ such that $n\mapsto w$. Then we define $\alpha_j(w)$ to be the image of $n$ in $\Gamma(U,\mathcal{A}_j)$.

We now check that this map is well-defined. Consider $n_1,n_2\in \Gamma(U,\mathcal{N})$ such that $n_1\mapsto w$ and $n_2\mapsto w$. Then the difference must be in the kernel: $n_1-n_2 \in\Gamma(U,\mathcal{O}_\mathfrak{M})$. Since the central charge acts by zero, we have that the image of $n_1-n_2$ in $\Gamma(U,\mathcal{A}_j)$ is $0$.
\end{proof}

\begin{thm}[{\textcite[Theorem 3.3]{Manin--1988}}]\label{flat}
	There exists a flat holomorphic connection on the line bundle $\lambda_{j/2}\otimes\lambda_{1/2}^{-c_j}$. 
\end{thm}
\begin{proof}
	We show the existence of a section $\nabla\co\mathcal{T}_\mathfrak{M}\to \mathcal{A}_j$ via a diagram chase.

	Define $\mathcal{K}$ to be the kernel of the Witt superalgebroid action $a_\Lambda\co \mathcal{W}\to \mathcal{T}_\mathfrak{M}$. 
	In particular, $\mathcal{K}\cong \pi_*\mathcal{T}^s_{(X\setminus P)/\mathfrak{M}}$ is a bundle with fiber the Lie algebra $\mathfrak{k}=\Gamma(\Sigma\setminus p,\mathcal{T}^s_\Sigma)$ and an anchor map which is trivial.
	
	Now, let $k\in \Gamma(U,\mathcal{K})$ for some small enough open $U\subseteq\mathfrak{M}_{g,1^\infty_{\mathrm{NS}}}$. We have shown by \cref{commutant} that $k=[k_1,k_2]$ for some $k_1,k_2\in\Gamma(U,\mathcal{K})$. Then notice that since $a_\Lambda(k_i)=0\in\Gamma(U,\mathcal{T}_\mathfrak{M})$, then $\alpha_j(k_i)\in\ker(\Gamma(U,\mathcal{A}_{j})\to\Gamma(U, \mathcal{T}_\mathfrak{M}))=\Gamma(U,\mathcal{O}_\mathfrak{M})$. Thus $\alpha_j(k)=[\alpha_j(k_1),\alpha_j(k_2)]=0$ as an element of $\Gamma(U,\mathcal{A}_j)$. Further, for any $k'\in \Gamma(U',\mathcal{K})$ where $U'\subseteq\mathfrak{M}_{g,1^\infty_{\mathrm{NS}}}$ may not be small enough, we still find that $\alpha_j(k')=0$ since $\alpha_j$ must commute with restriction to an open covering of $\mathfrak{M}_{g,1^\infty_{\mathrm{NS}}}$ by small enough superspaces.
	
	Note that $\mathcal{T}_\mathfrak{M}=\coker(\mathcal{K}\to\mathcal{W})$. Since $\alpha_j$ maps $\mathcal{K}$ to zero, then by the universal property of the cokernel, then $\alpha_j$ must factor through $\mathcal{T}_\mathfrak{M}$ uniquely. That is, we have a unique morphism of Lie superalgebroids $\nabla\co\mathcal{T}_\mathfrak{M}\to \mathcal{A}_{j}$.
\end{proof}

\appendix

\section{Summary of the analogous classical result}\label{classical}

This section contains a summary of the previously established findings regarding the moduli space $\mathcal{M}_g$ of classical Riemann surfaces, the Sato Grassmannian, and the Krichever map. 
While several papers \cite{Arbarello-DeConcini-Kac-Procesi--1988,Beilinson-Shechtman--1988,Kawamoto-Namikawa-Tsuchiya-Yamada--1988,Kontsevich--1987e} provide a version of this story, the summary presented here is in analogy with the super case described in this paper.

\subsection{The moduli space of triples $\mathcal{M}_{g,1^{\infty}}$} \label{classical MG1infty}
Let a \emph{Riemann surface} refer to a compact complex manifold of dimension $1$, equivalently a complete nonsingular curve over $\mathbb{C}$.  A family of Riemann surfaces of genus $g$ is a proper submersion $\pi\co X\to S$ of relative dimension $1$, where each fiber is a genus $g$ Riemann surface, and $S$ is a complex manifold.
A section $\sigma:S\to X$ of the morphism $\pi$ is a coherent way to choose a puncture in each fiber, which we sometimes denote by a divisor $P=\textrm{div}(\sigma)$ of $X$.

\begin{prop}[{\textcite[Lemma 4.3]{Arbarello-DeConcini-Kac-Procesi--1988}}]\label{perfect classical prop}
	Let $C$ be an \emph{open} Riemann surface. Denote the Lie algebra of global vector fields as $\mathfrak{k}\coeq \Gamma(C,\mathcal{T}_C)$. Then the Lie algebra $\mathfrak{k}$ is perfect, that is to say, $H_1(\mathfrak{k}; \mathbb{C}) = \mathfrak{k}/[\mathfrak{k},\mathfrak{k}] = 0$.
\end{prop}
\begin{proof}We repeat the proof of Arbarello, De Concini, Kac, and Procesi here for reference.

	When $C$ is an open subset of $\mathbb{C}^1$, we have the global nonvanishing vector field $\partial\coeq \frac{\partial}{\partial z}$. Then since $\mathfrak{k}$ is generated by $\partial$ as an $\mathcal{O}_C$ module, it suffices to notice:
	\begin{align*}
		h\partial=\frac{1}{2}[h\partial,z\partial]+\frac{1}{2}[\partial,hz\partial].
	\end{align*}

	By \cite[Theorem 30.3]{Forster--1991}, there exists a global everywhere nonzero vector field $\partial\in\mathfrak{k}$. So any global vector field can be written uniquely as $h\partial$ for $h\in \Gamma(C,\mathcal{O}_C)$.
	
	By Noether's normalization lemma \cite[Theorem 13.3]{Eisenbud.1995.ca}, there exist finite projections $p_i\co C\to \mathbb{A}^1=\Spec(\mathbb{C}[z])$ such that the ramification divisors of $p_1$ and $p_2$ are disjoint.
	
	We define the following objects on $C$. 
	\begin{align*}
		z_i\coeq p_i^*(z) && \partial_i\coeq p_i^*\left(\pp{z}\right)
	\end{align*}
Note that $\partial_i$ are a rational vector fields with poles only at the ramification of $p_i$.

	By \cite[Theorem 26.5]{Forster--1991}, there exists a global regular function $f_i$ with divisor exactly the negative of the divisor of $\partial_i$. Thus, $f_i\partial_i$ is a global everywhere nonzero regular vector field. Therefore there exist everywhere nonzero functions $a_i$ such that $f_i\partial_i=a_i\partial$.
	
	Since $f_1$ and $f_2$ are relatively prime, then there exist functions $c_i$ such that 
	\begin{align*}
		h=c_1f_1+c_2f_2.
	\end{align*}
	We then claim that $h\partial\in\mathfrak{k}$ is equal to 
	\begin{align*}
		h\partial=\sum_{i=1,2}\frac{1}{2}[c_i\partial,z_i\partial]+\frac{1}{2}[\partial,c_iz_i\partial].
	\end{align*}
\end{proof}

Define the moduli stack $\mathcal{M}_{g,1^{k}}$ to be the moduli space of triples $(C,p,z)$, where $C$ is a genus $g$ Riemann surface, $p\in C$, and $z$ is a $k$-jet equivalence class of a formal parameter vanishing at $p$.
Specifically, we have the parameter $z\in \hat{\mathcal{O}}_{p}$, the formal neighborhood at $p$, and $z$ has a zero of order one at $p$. For high enough $k$, $\mathcal{M}_{g,1^{k}}$ is known to be a Deligne-Mumford stack.
By taking the projective limit over $k$, we construct the pro-Deligne-Mumford stack
\begin{align*}
	\mathcal{M}_{g,1^{\infty}}=\lim\limits_{\longleftarrow}\mathcal{M}_{g,1^{k}}.
\end{align*} 
Simply by forgetting the parameter $z$ and puncture $p$, we have the projection
\begin{align*}
	\mathcal{M}_{g,1^{\infty}} &\to \mathcal{M}_{g}\\
	(C,p,z)&\mapsto C
\end{align*}
where $\mathcal{M}_{g}$ is the moduli space of genus $g$ Riemann surfaces.

The Witt algebra $\mathfrak{witt}$ is the Lie algebra of vector fields on a punctured formal neighborhood of a point in $\mathbb{C}^{1}$. Explicitly, $
\mathfrak{witt}\cong\mathbb{C}(\!(z)\!)\frac{\partial}{\partial z}$.
In terms of a basis, we have $L_n=z^{-n+1}\frac{\partial}{\partial z}$ for $n\in\mathbb{Z}$ with $[L_m,L_n]=(m-n)L_{m+n}$.

For $\pi:X\to \mathcal{M}_{g,1^{k}}$ the universal family of Riemann surfaces with punctures represented by the divisor $P$ and a $k$-jet of a coordinate near the punctures, we have the short exact sequence
\begin{align*}
	0 \to \mathcal{T}_{X/\mathcal{M}_{g,1^{k}}}(-(k+1)P) \to \mathcal{T}_X(-(k+1)P) \to \pi^*(\mathcal{T}_{\mathcal{M}_{g,1^{k}}}) \to 0.
\end{align*}
The resulting Kodaira-Spencer map preserving this structure is 
\begin{align*}
	\delta: \mathcal{T}_{\mathcal{M}_{g,1^{k}}}  \xrightarrow{\sim} R^1 \pi _*\mathcal{T}_{X/\mathcal{M}_{g,1^{k}}}(-(k+1)P),
\end{align*}
which locally for $(C,p,z)$ is 
\begin{align*}
	\delta_{(C,p,z)}\co T_{(C,p,z)}\mathcal{M}_{g,1^k} \xrightarrow{\sim} H^1(C,\mathcal{T}_C(-(k+1)p)),
\end{align*}
where $P|_C=p$ is the puncture on $C$.
From this we find the dimension of $\mathcal{M}_{g,1^k}$ is $3g-2+k$ for $g\geq 2$.

Let $U$ be a formal neighborhood of a puncture $p\in C$. Choosing the formal parameter $z$ on $U$ such that the divisor $p$ is given by $z=0$ gives a trivialization of $\mathcal{T}_C$ on $U$. Using \v{C}ech cohomology, then 
\begin{align*}
	T_{(C,p,z)}(\mathcal{M}_{g,1^k})&\cong \Gamma(U\setminus p, \mathcal{T}_C)\big/\Big(\Gamma(U, \mathcal{T}_C(-(k+1)p))+\Gamma(C\setminus p, \mathcal{T}_C)\Big)\\
	&\cong \mathfrak{witt}\Big/\Big((z^{k+1}\mathbb{C}[\![z]\!]\textstyle\frac{\partial}{\partial z} )+\Gamma(C\setminus p, \mathcal{T}_C)\Big).
\end{align*}
Further, taking the projective limit over $k$ gives a description of the tangent space of $\mathcal{M}_{g,1^\infty}$:
\begin{align*}
	T_{(C,p,z)}(\mathcal{M}_{g,1^\infty})
	& \cong \mathfrak{witt}\big/D_{C,p,z}.
\end{align*}

\begin{prop}[\textcite{Kontsevich--1987e}; {\textcite[Proposition (3.19)]{Arbarello-DeConcini-Kac-Procesi--1988}}; {\textcite[Proposition (2.19)]{Kawamoto-Namikawa-Tsuchiya-Yamada--1988}}]\label{witt action on Mg}
	The Witt algebra acts on the moduli space $\mathcal{M}_{g,1^\infty}$ by vector fields, that is there exists a Lie algebra antihomomorphism $\Lambda\co\mathfrak{witt}\to\Gamma(\mathcal{M}_{g,1^\infty},\mathcal{T}_\mathcal{M})$.
\end{prop}

\begin{defn}\label{witt algebroid def}
	According to \cite[last paragraph page 8]{KosmannSchwarzbach-Mackenzie--2002}, define $(\mathcal{W},a_\Lambda)$ to be the action Lie algebroid associated to $\Lambda\circ \text{op}\co \mathfrak{witt}^\text{op}\to \Gamma(\mathcal{M}_{g,1^\infty},\mathcal{T}_{\mathcal{M}})$ as in \cref{witt action on Mg}.
\end{defn}

\begin{defn}[{\textcite[Definition 5.9]{Mumford--1977}}; \textcite{Deligne--1987}]\label{line bundle lambda j unsuper}
	Let $\pi: X\to S$ be a proper family of complex manifolds of relative dimension $1$. Let $\mathcal{F}$ be a locally free sheaf on $X$. Then the determinant of cohomology of $\mathcal{F}$ is an invertible sheaf on $S$ given by
	\begin{align*}
		D(\mathcal{F})\coeq \otimes_i\left(\det R^i\pi_*\mathcal{F}\right)^{(-1)^i}.
	\end{align*} 
	We define the determinant line bundles $\lambda_{j}$ for the universal family $\pi:X\to \mathcal{M}_{g,1^{\infty}}$ as
	\begin{align*}
		\lambda_{j}\coeq D(\omega_{X/\mathcal{M}}^{\otimes j})
	\end{align*}
	where $\omega_{X/\mathcal{M}}\coeq \Omega_{X/\mathcal{M}}^1$. 
\end{defn}

\begin{thm}[{\textcite[Theorem 5.10]{Mumford--1977}}]
	The Mumford isomorphism is the collection of canonical isomorphisms 
	\begin{align*}
		\lambda_{j}\cong \lambda_{1}^{(6j^2-6j+1)}, \qquad\textrm{ in particular } \qquad \lambda_{2}\cong \lambda_{1}^{13}.
	\end{align*}
\end{thm}

\subsection{The semi-infinite Grassmannian of $\mathbb{C}(\!(z)\!)$}\label{unsuper Gr}

The semi-infinite Grassmannian, defined originally by Sato in \cite{Sato--1981}, 
	originated via the study of soliton equations and KP equations (see also \cite{Sato-Sato--1983}). 
	In this section, we mostly follow the spirit of Kontsevich \cite{Kontsevich--1987e}, however the descriptions in \cite{Segal-Wilson--1985,Pressley-Segal--1986} and \cite{AlvarezVazquez-MunozPorras-PlazaMartin--1998,PlazaMartin.2000.tGok} are alternatives.

Consider $\mathbb{C}(\!(z)\!)$, the infinite-dimensional topological space of formal Laurent series with the $z$-adic topology. Define the vector spaces $H=\mathbb{C}(\!(z)\!)$, $H^+=\mathbb{C}[\![z]\!]$, and $H^-=z^{-1}\mathbb{C}[z^{-1}]$.
\begin{defn}
	Define a subspace of $H$ to be \emph{compact} if it is commensurable with $H^+$. Explicitly, subspaces $K$ and $H^+$ are commensurable when $K/(K\cap H^+)$ and $H^+/(K\cap H^+)$ are both finite-dimensional.
	
	Define a subspace $D$ of $H$ to be \emph{discrete} if there exists a compact subspace $K$ such that the natural map $D\oplus K\to H$ is an isomorphism.
\end{defn}

\begin{lem}\label{classical discrete altdef}
	Let $K$ be any compact subspace. Then $D$ is discrete if and only if the natural map $D\oplus K \to H$ is Fredholm, that is when the kernel and cokernel are finite dimensional.
	\begin{align*}
		0\to D\cap K\to D\oplus K\to H\to H/(D+K)\to 0
	\end{align*}
\end{lem}
\begin{proof}
	Firstly we show that commensurability is transitive, i.e. any two compact subspaces are commensurable to each other. For $K$ and $K'$ compact, using the exact sequence
	\begin{align*}
		0\to  \frac{(K\cap H^+)+(K\cap K')}{K\cap K'}\to \frac{K}{K\cap K'}\to \frac{K}{(K\cap H^+)+(K\cap K')}\to 0,
	\end{align*}
	where the first term is isomorphic to $(K\cap H^+)/(K\cap K'\cap H^+)$, so we see $\dim(K/(K\cap K'))\leq \dim(H^+/(K'\cap H^+))+\dim(K/(K\cap H^+))<\infty$. And $\dim(K/(K\cap K'))<\infty$ similarly.

	Next we show that if $D\cap K$ and $H/(D+K)$ are finite-dimensional for some compact $K$, then $D\cap K'$ and $H/(D+K')$ are finite-dimensional for any other compact $K'$. Using the exact sequence
	\begin{align*}
		0\to D\cap K\cap K'\to D \cap K'\to D\cap K'/(D\cap K\cap K')\to 0,
	\end{align*}
	we see $\dim(D \cap K')\leq \dim(D\cap K)+\dim(K'/(K\cap K'))<\infty$. Using the exact sequence
	\begin{align*}
		0\to \frac{K}{(D+K')\cap K}\cong \frac{D+K'+K}{D+K'}\to \frac{H}{D +K'}\to \frac{H}{D+K+K'}\to 0,
	\end{align*}
	we see $\dim(H/(D +K'))\leq \dim(K/(K'\cap K))+\dim(H/(D +K))<\infty$.

	Assume $D$ is discrete. Then $D\oplus K'=H$ for some compact $K'$. Thus $D\cap K'=0$ and $H/(D+K')=0$. Then $D\cap K$ and $H/(D+K)$ are finite-dimensional by the previous paragraph.
	
	Assume $D\cap K$ and $H/(D+K)$ are finite-dimensional. We show $D$ is discrete by showing $H/D$ is compact. It suffices to notice
	\begin{align*}
		\frac{K}{K\cap (H/D)}\cong D\cap K, && \frac{H/D}{K\cap (H/D)}\cong \frac{H}{D+K}.
	\end{align*}
\end{proof}

\begin{defn}\label{GrDef unsuper}
	The Grassmannian $\Gr(H)$ is the set of all discrete subspaces $D\subset H$.
\end{defn}

\begin{prop}[\textcite{Kontsevich--1987e}]
	$\Gr(H)$ is locally modeled on the vector space $\Hom_\mathbb{C}(H^-,H^+)$. That is, $\Gr(H)$ is covered by open charts $U_{D,K}\cong \Hom_\mathbb{C}(D,K)$. 
	And thus $\Gr(H)$ is an infinite-dimensional complex manifold.
\end{prop} 
\begin{proof}
	This proof is essentially \cite[Proposition (7.1.2)]{Pressley-Segal--1986} slightly modified to our setting.
\end{proof}

Define $\mathfrak{gl}(H)=\Hom(H,H)$ to be the space of continuous linear maps on $H$. For any endomorphism of $H$ we may write it as
\begin{align*}
	F=\begin{pmatrix}
		F^{--} & F^{-+}\\
		F^{+-} & F^{++}
	\end{pmatrix}
\end{align*}
where $F^{+-}: H^-\to H^+$ etc.
Define the group $\GL(H)\subseteq \Aut(H)$ to be those invertible maps which are homeomorphisms (i.e. bicontinuous).\footnote{The general linear group $\GL(H)$ we have defined is analogous to the group denoted $\GL_F(H)$ in \cite{Pressley-Segal--1986} and \cite{Segal-Wilson--1985}.} That is $G\in \GL(H)$ iff for every compact $K$, $G(K)$ and $G^{-1}(K)$ are also compact. Since all the compact subspaces are commensurable, this implies that  $G^{--}$ and $G^{++}$ are Fredholm.

Define an endomorphism of a discrete subspace to be trace class if it factors through some compact subspace.

\begin{prop}[cf. {\textcite[Proposition 7.1.3]{Pressley-Segal--1986}}]
	The group $\GL(H)$ acts transitively on $\Gr(H)$, and the stabilizer of $D$ is $\left(\begin{smallmatrix}* & *\\0 &*\end{smallmatrix}\right)$ expressed in a decomposition $D\oplus K$. 
\end{prop}
\begin{proof}
	
		Consider $G\in\GL(H)$ and discrete $D$. Consider compact $K$ such that the natural map $D\oplus K\to H$ is an isomorphism. Then $G(D)\oplus G(K)\cong H$ since $G$ is an automorphism. Further, since $G(K)$ is compact, then $G(D)$ is discrete.

	To show the transitivity of the action, consider an operator $w=\left(\begin{smallmatrix}w_-\\w_+\end{smallmatrix}\right)\co H^-\to H$ such that $p_-\circ w=w_-$ is Fredholm by \cref{classical discrete altdef}, and a map $v=\left(\begin{smallmatrix}v_-\\v_+\end{smallmatrix}\right)\co H^+\to H$ where the columns of $v$ are a basis for $K$ such that the natural map $D\oplus K\to H$ is an isomorphism. Since $K$ and $H^+$ are commensurable, we may assume eventually $v(z^i)=z^{i-n}$ for some $n$. We claim that
	\begin{align*}
		G=\begin{pmatrix}
			w_- & v_-\\
			w_+ & v_+
		\end{pmatrix}
	\end{align*}
	is in $\GL(H)$ and $G(H^-)=D$. The fact that $G$ is invertible follows from assuming that $D\oplus K\to H$ is an isomorphism, and the bicontinuity of $G$ follows from the condition that eventually $v(z^i)=z^{i-n}$ for some $n$.
	
	The stabilizer is obvious.
\end{proof} 

Therefore we may describe $\Gr(H)$ as the homogeneous space
\begin{align*}\label{homogeneousGr unsuper}
	\Gr(H)\cong \GL(H)/P
\end{align*} 
where $P=\left(\begin{smallmatrix}* & *\\0 &*\end{smallmatrix}\right)$ in the $H^-\oplus H^+$ decomposition.

\begin{prop}[Kontsevich \cite{Kontsevich--1987e}; {\textcite[Proposition (1.14)]{Kawamoto-Namikawa-Tsuchiya-Yamada--1988}}]\label{vector fields on Gr unsuper}
	The Lie algebra $\mathfrak{gl}(H)$ acts by vector fields on $\Gr(H)$. Explicitly, $F\mapsto L_F$ is a Lie algebra antihomomorphism $L\co\mathfrak{gl}(H)\to \mathcal{T}_{\Gr(H)}$ sending $[F_1,F_2]$ ot $[L_{F_2},L_{F_1}]$. In the chart $U_{D,K}$, this action is given by the formula
	\begin{align*}
		L_F(A)=F^{KD}+F^{KK}A-AF^{DD}-AF^{DK}A
	\end{align*} 
	where $L_F
	\in \Hom_\mathbb{C}(S^\bullet(\Hom_\mathbb{C}(D,K)),\Hom_\mathbb{C}(D,K))$ acts on functions by 
		\begin{align*}
		L_F(f(A))=\lim_{\epsilon \to 0}\frac{f(A+\epsilon L_FA)-f(A)}{\epsilon}.
	\end{align*}
\end{prop}
\begin{proof}
	The proof is the same as the proof of \cref{vector fields on Gr}.
\end{proof}

We can take as definition
\begin{align*}
	\textstyle\det_K(D)\coeq \dfrac{\det(D\cap K)}{\det(H/(D+K))},
\end{align*}
which is well-defined for discrete $D$ and compact $K$ by \cref{classical discrete altdef}.
\begin{defn}\label{Gr det bundle}
	As in \cite{Arbarello-DeConcini-Kac-Procesi--1988}, we define the determinant line bundle on $\Gr(H)$ as $\mathcal{d\!e\!t}_{\Gr(H)}\coeq\det_{H^+}$, in other words, $\det_K$ with $K=H^+$, the distinguished compact subspace.
\end{defn}

It is known that $H^2(\mathfrak{gl}(H))$ is one-dimensional \cite{Kac.Peterson.1981.sawroidLaag}.
For any choice of discrete $D$ and compact $K$ such that the natural map $D\oplus K\xrightarrow{\sim}H$ is an isomorphism, we define a 2-cocycle on $\mathfrak{gl}(H)$ as
\begin{align*}
	\eta_{D,K}(F,G):=\tr(F^{DK}G^{KD}-F^{DK}G^{KD}).
\end{align*}
We choose a distinguished 2-cocycle:
\begin{align*}
	\eta(F,G):=
	\tr(F^{-+}G^{+-}-F^{-+}G^{+-}),
\end{align*}
which is known as the Japanese cocycle \cite{Khesin-Wendt--2009}. 
The unique Lie algebra central extension defined by the Japanese cocycle will be denoted $\widetilde{\mathfrak{gl}}(H)$. 
 We denote the bracket on $\widetilde{\mathfrak{gl}}(H)$ as $[F,G]^{\sim}=[F,G]+\eta(F,G)$.

\begin{thm}[\textcite{Kontsevich--1987e}] \label{Kontsevich gl tilde action}
	The Lie algebra $\widetilde{\mathfrak{gl}}(H)$ 
	acts by first order differential operators on $\mathcal{d\!e\!t}_{\Gr(H)}$. Explicitly, $F+c\mapsto \widetilde{L}_F+c$ is a Lie algebra antihomomorphism $\tilde{L}\co\widetilde{\mathfrak{gl}}(H)\to \mathcal{A}_{\mathcal{d\!e\!t}}(\Gr(H))$ sending $[F_1,F_2]^\sim$ to $[\widetilde{L}_{F_2},\widetilde{L}_{F_1}]$. In the chart $U_{D,K}$, this action is given by the formula
	\begin{align*}
		\widetilde{L}_F(A)=L_F(A)+\tr(F^{DK}A)+\alpha(F)
	\end{align*} 
	where $\alpha\in C^1(\mathfrak{gl}(H))$ is the unique 1-cochain such that
	\begin{align*}
		d\alpha(F,G)=\alpha([F,G])=\eta_{D,K}(F,G)-\eta(F,G).
	\end{align*}
\end{thm}

\begin{proof}
	The proof is analogous to the proof of \cref{diff operators on det}.
\end{proof}

\begin{defn}\label{gl algebroid defs}
	According to \cite[last paragraph page 8]{KosmannSchwarzbach-Mackenzie--2002}, define $(\mathcal{G},a_L)$ to be the action Lie algebroid associated to $L\circ \text{op}\co \mathfrak{gl}(H)^\text{op}\to \Gamma(\Gr(H),\mathcal{T}_{\Gr(H)})$ as in \cref{vector fields on Gr unsuper}. 
	
	Similarly, define $\widetilde{\mathcal{G}}$ to be the action Lie algebroid associated to $\sym\circ\tilde{L}\circ\text{op}\co\widetilde{\mathfrak{gl}}(H)^\text{op}\to \Gamma(\Gr(H),\mathcal{T}_{\Gr(H)})$ as in \cref{Kontsevich gl tilde action}.
\end{defn}

By the correspondence in \cite[Theorem 2.4]{KosmannSchwarzbach-Mackenzie--2002}, the Lie algebra action $\tilde{L}\circ\text{op}\co\widetilde{\mathfrak{gl}}(H)^\text{op}\to \Gamma(\Gr(H),\mathcal{A}_{\mathcal{d\!e\!t}})$ may be used to define a morphism of Lie algebroids $b_{\tilde{L}}\co\widetilde{\mathcal{G}}\to\mathcal{A}_{\mathcal{d\!e\!t}}$.

In summary, we have the commutative diagram of Lie algebroids below.
\begin{equation} \label{gl tilde action square unsuper}
	\begin{tikzcd}
		0 \arrow{r} & \mathcal{O}_{\Gr(H)} \arrow{r}\arrow{d}{\id} & \widetilde{\mathcal{G}} \arrow{r}\arrow{d}{b_{\widetilde{L}}} & \mathcal{G}\arrow{r}\arrow{d}{a_L} & 0\\
		0\arrow{r} & \mathcal{O}_{\Gr(H)} \arrow{r} & \mathcal{A}_{\mathcal{d\!e\!t}} \arrow{r} & \mathcal{T}_{\Gr(H)} \arrow{r} & 0
	\end{tikzcd}
\end{equation}

\subsection{The Krichever map and a flat connection}
\label{unsuper krichever}

\begin{prop}[{\textcite[Proposition 6.1]{Segal-Wilson--1985}}]
	The vector subspace $\Gamma(C\setminus p,\omega_C^{\otimes j})\subset \mathbb{C}(\!(z)\!)\;dz^{\otimes j}\cong H$ is a discrete subspace.
\end{prop}
\begin{proof} 
	Notice the map $p_-$ restricted to $\Gamma(C\setminus p,\omega_C^{\otimes j})$ is Fredholm, so the result follows from \cref{classical discrete altdef}.
	\begin{align*}
		0\to \Gamma(C,\omega_C^{\otimes j})\to \Gamma(C\setminus p,\omega_C^{\otimes j})\xrightarrow{p_-} H^-\to H^1(C,\omega_C^{\otimes j})\to 0.
	\end{align*}
\end{proof}

\begin{defn}[{\textcite[Proposition 6.2]{Segal-Wilson--1985}}]
	The Krichever map $\mathcal{M}_{g,1^\infty} \to \Gr$ defined by
	\begin{align*}
		\kappa_j(C,p,z)= \Gamma(C\setminus p,\omega_C^{\otimes j})\subset \mathbb{C}(\!(z)\!)\;dz^{\otimes j}.
	\end{align*}
	where $\omega_C\coeq \Omega_C^1$, is an injective map.
\end{defn}

The vector fields act on sections of $\omega_C^{\otimes j}$ by Lie derivative
\begin{align*}
	\mathcal{L}_{f(z)\frac{\partial}{\partial z}}\big(g(z)\;dz^{\otimes j}\big)=\Big(f(z)g'(z)+{j}f'(z)g(z)\Big)\;dz^{\otimes j}.
\end{align*} 
For $U$ a formal neighborhood of $p\in \CC^1$, we see that $\Gamma(U\setminus p,\omega^{\otimes j}_{C})$ is canonically identified with $\CC(\!(z)\!)\,dz^{\otimes j}\cong H$. 
We define the map $\varrho_j: \mathfrak{witt}\to \mathfrak{gl}(H)$ as
\begin{align*}
	\varrho_j\bigg(f(z)\frac{\partial}{\partial z}\bigg)\big(g(z)\big)=f(z)g'(z)+{j}f'(z)g(z).
\end{align*}

\begin{prop}[{\textcite[(2.24)]{Arbarello-DeConcini-Kac-Procesi--1988}}]
	The pullbacks of the Japanese cocycle $\eta$ along the representations $\varrho_j$ satisfies:
	\begin{align*}
		\varrho_j^*(\eta)=(6j^2-6j+1)\varrho_1^*(\eta) &&\varrho_j^*(\eta)=\varrho_{1-j}^*(\eta).
	\end{align*}
\end{prop}

It is known that $H^2(\mathfrak{witt})$ is one dimensional \cite[Proposition 2.1 (2)]{Arbarello-DeConcini-Kac-Procesi--1988}. 
Let $c_j=(6j^2-6j+1)$ and define $\mathfrak{vir}_j$ to be the central extension of $\mathfrak{witt}$ by $\varrho^*_j(\eta)=c_j\varrho^*_1(\eta)$. When $j=0$ or $j=1$, we recover the standard definition of the Virasoro algebra, denoted $\mathfrak{vir}$. On the basis of $\mathfrak{witt}$ the standard cocycle is given by
\begin{align*}
	\varrho^*_0(\eta)(L_m,L_n)=\frac{m^3-m}{4}\;\delta_{m+n,0} .
\end{align*}
This is summarized in the following commutative diagram of Lie algebras.
\begin{equation}\label{vir j diagram}
	\begin{tikzcd}
		0 \arrow{r} & \mathbb{C} \arrow{r}\arrow{d}{\cdot c_j} & \mathfrak{vir} \arrow{r}\arrow{d}{\wr} & \mathfrak{witt} \arrow{r}\arrow{d}{\id} & 0\\
		0 \arrow{r} & \mathbb{C} \arrow{r}\arrow{d}{\id} & \mathfrak{vir}_j \arrow{r}\arrow[hook]{d}\arrow[very near start, phantom]{rd}{\lrcorner} & \mathfrak{witt} \arrow{r}\arrow[hook]{d}{\varrho_j} & 0\\
		0 \arrow{r} & \mathbb{C} \arrow{r} & \widetilde{\mathfrak{gl}}(H) \arrow{r} & \mathfrak{gl}(H) \arrow{r} & 0
	\end{tikzcd}
\end{equation}

\begin{prop}[{\textcite[Proposition 3.24]{Arbarello-DeConcini-Kac-Procesi--1988}}]\label{gl witt algebra diagram unsuper}
	The representation $\varrho_j$ and the Krichever map $\kappa_j$ are compatible, that is the diagram below of Lie algebras commutes.
	\begin{equation*}
		\begin{tikzcd}
			\mathfrak{witt} \arrow{d}{\varrho_j}\arrow{r}{\Lambda} & \Gamma(\mathcal{M}_{g,1^\infty},\mathcal{T}_{\mathcal{M}})\arrow{d}{d\kappa_j} \\
			\mathfrak{gl}(H_j)  \arrow{r}{L}& \Gamma(\mathcal{M}_{g,1^\infty},\kappa_j^*\mathcal{T}_{\Gr(H)})  
		\end{tikzcd}
	\end{equation*}
\end{prop}

\begin{proof}
	The representation $\varrho_j$ is based on the Lie derivative of $j$-differentials, while $d\kappa_j$ is also given by Lie derivative action of $j$-differentials. So essentially, the diagram is commutative by construction.
\end{proof}

Recall \cref{gl algebroid defs} of $(\mathcal{G},a_L)$ and \cref{witt algebroid def} of $(\mathcal{W},a_\Lambda)$. Then \cref{gl witt algebra diagram unsuper} implies the existence of a morphism of these Lie algebroids on $\mathcal{M}_{g,1^\infty}$ as below. 
\begin{equation}\label{gl witt algebroid diagram unsuper}
	b_{\varrho_j}\co	\mathcal{W} \to
		{\kappa_j}^!\mathcal{G} 
\end{equation}

We now combine the $L$ and $\widetilde{L}$ action as in \cref{gl tilde action square unsuper} with the $\varrho_j$ compatibility diagram \cref{gl witt algebroid diagram unsuper}. This allows us to derive a canonical action of the Virasoro algebra. The general strategy may be summarized as pulling back the diagram \cref{gl tilde action square unsuper} along the Krichever map to $\mathcal{M}_{g,1^\infty}$, and further restricting along $b_{\varrho_j}$ to $\mathcal{W}$.

We work with the action Lie algebroids $\mathcal{W}$, $\mathcal{G}$, $\tilde{\mathcal{G}}$ associated to the action of $\mathfrak{witt}$ on $\mathcal{M}_{,1^\infty}$ and $\mathfrak{gl}$ and $\tilde{\mathfrak{gl}}$ on $\Gr(H)$ respectively, as in \cref{witt algebroid def} and \cref{gl algebroid defs}.  
Since $\widetilde{\mathcal{G}}$, $\mathcal{G}$, and $\mathcal{A}_{\mathcal{d\!e\!t}}$ are all transitive, we can consider their pullbacks along $\kappa_j$. Further by \cite[Proposition 1.8]{Higgins-Mackenzie--1990}, the morphisms between them also pullback.

\begin{equation*}
	\begin{tikzcd}
		0\arrow{r}& \mathcal{O}_{\mathcal{M}}\arrow{d}{\id}\arrow[hook]{r}  &{\kappa_j}^!\widetilde{\mathcal{G}} \arrow[two heads]{d}{\kappa_j^!b_{\widetilde{L}}}\arrow[two heads]{r}   & {\kappa_j}^!\mathcal{G}\arrow[two heads]{d}{\kappa_j^!a_L} \arrow{r}&0 \\
		0\arrow{r}&\mathcal{O}_{\mathcal{M}}\arrow[hook]{r} 
		& {\kappa_j}^!\mathcal{A}_{\mathcal{d\!e\!t}} \arrow[two heads]{r} & \mathcal{T}_{\mathcal{M}} \arrow{r}&0
	\end{tikzcd}
\end{equation*}

\begin{cor}[{\textcite[Lemma 4.7]{Arbarello-DeConcini-Kac-Procesi--1988}}]
	There is an isomorphism of Lie algebroids $\mathcal{A}_{{{\kappa_j}^*}\mathcal{d\!e\!t}} \cong {\kappa_j}^!\mathcal{A}_{\mathcal{d\!e\!t}} $.
\end{cor}
\begin{proof}
	Apply \cref{pullback of atiyah prop}.
\end{proof}

\begin{defn}\label{nsj algebroid defs}
	According to \cite[last paragraph page 8]{KosmannSchwarzbach-Mackenzie--2002}, define $\mathcal{V}_j$ to be the action Lie algebroid associated to the composition of $\text{op}\co \mathfrak{vir}_j^\text{op}\to\mathfrak{vir}_j$, the projection $\mathfrak{vir}_j\to \mathfrak{witt}$ as in \cref{vir j diagram}, and $\Lambda\co \mathfrak{witt}\to \Gamma(\mathcal{M}_{g,1^\infty},\mathcal{T}_{\mathcal{M}})$ as in \cref{witt action on Mg}.

	Similarly, define $\mathcal{V}$ for $\mathfrak{vir}$.
\end{defn}

Then the diagram \cref{vir j diagram} gives the diagram below of Lie algebroids on $\mathcal{M}_{g,1^\infty}$.
\begin{equation*}
	\begin{tikzcd}
		0 \arrow{r} & \mathcal{O}_\mathcal{M} \arrow{r}\arrow{d}{\cdot c_j} & \mathcal{V} \ar{r}\arrow{d}{\wr} & \mathcal{W}\arrow{r}\arrow{d}{\id} & 0\\
		0 \arrow{r} & \mathcal{O}_\mathcal{M} \arrow{r}\arrow{d}{\id} & \mathcal{V}_j \arrow{r}\arrow[hook]{d}\arrow[ very near start, phantom]{rd}{\lrcorner} & \mathcal{W} \arrow{r}\arrow[hook]{d}{b_{\varrho_j}} & 0\\
		0 \arrow{r} & \mathcal{O}_\mathcal{M} \arrow{r} & {\kappa_j}^!\widetilde{\mathcal{G}} \arrow{r} & {\kappa_j}^!\mathcal{G} \arrow{r} & 0
	\end{tikzcd}
\end{equation*}

Thus, we arrive at the main diagram below showing the action of $\mathcal{V}$ on an Atiyah Lie algebroid.
The crucial property of the action of the Virasoro Lie algebroid $\mathcal{V}$ is that the central charge acts by $c_j$.
\begin{equation}
	\begin{tikzcd}
		\mathcal{O}_\mathcal{M} \arrow{r}\arrow{d}{\cdot c_j} & \mathcal{V} \arrow{r}\arrow[hook]{d} & \mathcal{W}\arrow{d}{\varrho_j} \arrow[hook]{d}{b_{\varrho_j}}\arrow[two heads,bend left=60]{dd}{a_\Lambda} \\
		\mathcal{O}_\mathcal{M}\arrow{d}{\id}\arrow{r}&{\kappa_j}^!\widetilde{\mathcal{G}} \arrow[two heads]{d}{\kappa_j^!b_{\widetilde{L}}}\arrow[two heads]{r}   & {\kappa_j}^!\mathcal{G}\arrow[two heads]{d}{\kappa_j^!a_L} \\
		\mathcal{O}_\mathcal{M}\arrow{r}&\mathcal{A}_{{{\kappa_j}^*}\mathcal{d\!e\!t}} \arrow[two heads]{r} & \mathcal{T}_{\mathcal{M}} 
	\end{tikzcd}
\end{equation}

\begin{prop}[\textcite{Kontsevich--1987e}]
	The pullback of the determinant line bundle from the Grassmannian along the $j^\textrm{th}$ Krichever map is canonically isomorphic to $\lambda_{j}$.
	\begin{equation*}
		{\kappa_j}^*{\mathcal{d\!e\!t}}_{\Gr(H)}\cong \lambda_{j}
	\end{equation*}
\end{prop}
\begin{proof}
	We see $\kappa_j(X/S,P,z)=\pi_*\omega_{(X\setminus P)/S}^j$. Using the \cref{Gr det bundle} of the determinant line bundle on the Grassmannian, we have 
	\begin{align*}
		\textstyle\det_{\Gr(H)}\Big(\pi_*\omega_{(X\setminus P)/S}^j\Big)=\det R^0\pi_* \omega_{X/S}^j\Big/\det R^1\pi_*\omega_{X/S}^j
	\end{align*}
	which can be seen using relative \v{C}ech cohomology. This is exactly the \cref{line bundle lambda j unsuper} of $\lambda_{j}$.
\end{proof}

\begin{prop}\label{virasoro action on Mg}
	The Virasoro algebra acts by first order differential operators on $\lambda_j$ with central charge $c_j$. Precisely, there is a Lie algebra antihomomorphism $\mathfrak{vir}\to \Gamma(\mathcal{M}_{g,1^\infty},\mathcal{A}_{\lambda_j})$ which is compatible with the action $\Lambda$ of $\mathfrak{witt}$ given in \cref{witt action on Mg}.
\end{prop}

Using these standard properties of Atiyah algebras described in \cref{Atiyah section}, we describe the action of the Virasoro Lie algebroid on the line bundles of the Mumford isomorphism. 

Denote $\mathcal{A}_{j}\coeq \mathcal{A}_{\lambda_{j}\otimes \lambda_{1}^{-c_j}}$.
From \cref{virasoro action on Mg}, $\mathcal{V}$ acts on $\mathcal{A}_{\lambda_{j/2}}$ with central charge $c_j$, and further by \cref{Atiyah defs} combined with \cref{Atiyah defs iso lemma}, $\mathcal{V}$ acts on $\mathcal{A}_{\lambda_{1/2}^{-c_j}}$ with central charge $-c_j$. Then the action of $\mathcal{V}$ on $\mathcal{A}_{j}$ is defined via the action on $\mathcal{A}_{\lambda_{j/2}}\times_{\mathcal{T}_X} -c_j\mathcal{A}_{\lambda_{1/2}}$. Importantly, we see that $\mathcal{V}$ acts with central charge $c_j-c_j=0$, as shown in the commutative diagram of Lie algebroids below.
\begin{equation*}
	\begin{tikzcd}
		& &  & \mathcal{K}\ar{d} & \\
		0 \arrow{r} & \mathcal{O}_{\mathcal{M}} \arrow{r}\arrow{d}{}{0} & \mathcal{V} \arrow{r}\arrow{d} & \mathcal{W} \arrow{r}\arrow{d}{a_\Lambda}\arrow[dashed]{dl}{\alpha_j} & 0\\
		0\arrow{r} & \mathcal{O}_{\mathcal{M}} \arrow{r} & \mathcal{A}_{j} \arrow{r} & \mathcal{T}_\mathcal{M}  \arrow{r} & 0
	\end{tikzcd}
\end{equation*}

\begin{lem}
	There exists a unique morphism of Lie algebroids $\alpha_j\co \mathcal{W}\to \mathcal{A}_{j}$ which commutes with the Virasoro action $\mathcal{V}\to \mathcal{A}_j$ and the projection $\mathcal{V}\to\mathcal{W}$. In other words, the Virasoro action on $\mathcal{A}_j$ descends to a Witt action.
\end{lem}
\begin{proof}
	The proof is the same as for \cref{aplha map lemma}.
\end{proof}

\begin{thm}[\textcite{Kontsevich--1987e}; {\textcite{Arbarello-DeConcini-Kac-Procesi--1988}}]
	There exists a flat holomorphic connection on the line bundle $\lambda_1^{-c_j}\otimes\lambda_j$. 
\end{thm}
\begin{proof}
	Define $\mathcal{K}$ to be the kernel of the Witt algebroid action $a_\Lambda \co \mathcal{W} \to \mathcal{T}_\mathcal{M}$. Since the action is surjective, by \cite[Theorem 1.4]{Mackenzie.1987.LgaLaidg}, $\mathcal{K}$ is a Lie algebra bundle.
	In particular, $\mathcal{K}\cong\pi_*\mathcal{T}_{(X\setminus P)/\mathcal{M}}$ is a bundle with fiber the Lie algebra $\mathfrak{k} = \Gamma(C \setminus p, \mathcal{T}_C )$ and an anchor map which is trivial. Therefore the bracket on $\mathcal{K}$ is given by the pointwise bracket of its fibers. By \cref{perfect classical prop}, we know every fiber is equal to its commutant $\mathfrak{k} = [\mathfrak{k}, \mathfrak{k}]$, therefore $\mathcal{K} = [\mathcal{K} , \mathcal{K} ]$.
	
	The rest of the proof is identical to the proof of \cref{flat}.
\end{proof}

\section{Proof of \cref{pullback diagram lemma} (a condition to be a pullback)}\label{pullback proof}

We label the morphisms in the diagram as follows.
\begin{equation*}
	\begin{tikzcd}
		0 \arrow{r} & A \arrow{r}{\iota}\arrow{d}{a}[swap]{\wr} & B \arrow{r}{\pi}\arrow{d}{b} & C \arrow{r}\arrow{d}{c} & 0\\
		0 \arrow{r}& A' \arrow{r}{}{\iota} & B'\arrow{r}{}{\pi} & C' \arrow{r} & 0
	\end{tikzcd}
\end{equation*}

\begin{proof}
	Suppose we have $X$ such that the following square commutes.
	\begin{equation*}
		\begin{tikzcd}
			X \arrow{r}{g}\arrow{d}{f} & C \arrow{d}{c}\\
			B'\arrow{r}{}{\pi} & C'
		\end{tikzcd}
	\end{equation*}
	Let $x\in X$. Choose $y\in B$ such that $\pi(y)=g(x)$. Then $b(y)-f(x)\in B'$, and further $\pi'(b(y)-f(x))=\pi'\circ b(y)-\pi'\circ f(x)=c\circ \pi(y)-c\circ g(x)=c(\pi(y)-g(x))=0$. Therefore we have $b(y)-f(x)\in A'$. And using the isomorphism $a$ and inclusion into $B$, we have $a^{-1}(b(y)-f(x))\in B$.
	
	Define $\phi\co X\to B$ as $\phi(x)=y-a^{-1}(b(y)-f(x))$. We first check this map commutes, then check it does not depend on the choice of $y$.
	\begin{align*} 
		&\pi\circ\phi(x)=\pi(y-a^{-1}(b(y)-f(x)))=\pi(y)-\pi(a^{-1}(b(y)-f(x)))=\pi(y)=g(x)\\
		& b\circ\phi(x)=b(y-a^{-1}(b(y)-f(x)))=b(y)-b(a^{-1}(b(y)-f(x)))=b(y)-(b(y)-f(x))=f(x)
	\end{align*}
	Now consider two lifts $y,y'\in B$ such that $\pi(y)=\pi(y')=g(x)$. Consider the subtraction of the resulting maps 
	\begin{align*}
		\phi(x)-\phi'(x)&=[y-a^{-1}(b(y)-f(x))]-[y'-a^{-1}(b(y')-f(x))]=(y-y')-a^{-1}(b(y-y'))=0
	\end{align*}
	where the last equality results from noticing that $y-y'\in A$ and $b(y-y')\in A'$.
\end{proof}

\section{Proof of \cref{superconformalNoether} (superconformal Noether normalization)}\label{Noether Normalization proof appendix}

Before considering a family of open SRSs as in \cref{superconformalNoether}, we first provide the required preliminaries from the setting of spin curves over a point, which is the setting studied by \textcite{Artebani-Pirola--2005}.

\begin{lem}[{\textcite[Lemma 3]{Artebani-Pirola--2005}}]\label{technical}
	Let $A$ and $B$ be effective divisors on a compact Riemann surface $C$ such that $[A]\cap [B]=\emptyset$ and $\deg(A-B)-\deg[A]-\deg[B]> 2g-2$. Then
	\begin{align*}
		H^0(C,\Omega^1_C(A-B))\to H^1_{dR}\big(C\setminus[A]\big) && \sigma\to[\sigma]
	\end{align*}
	is onto.
\end{lem}

	\begin{prop}[{\textcite[Proposition 2]{Artebani-Pirola--2005}}] \label{uplsion surjective}
	Let $C$ be a compact Riemann surface, let $\Omega_C^{1/2}$ be a spin structure on $C$, let $p\in C$, and let $n\geq 0$. Consider the analytic map below. 
	\begin{align*}
		\Upsilon_C\co H^0\big(C,\Omega_C^{1/2}(np)\big)\to H^1_{dR}(C\setminus p) && \sigma\mapsto [\sigma^2]_{dR}
	\end{align*}
	If $n>6g-2$, then $\dim \Upsilon_C^{-1}(0)=n-2g$.
\end{prop}

\begin{proof}
	The proof can be found in \cite{Artebani-Pirola--2005} and is very similar to the proof of \cref{uplsion Q surjective} below. In particular, \cref{technical} is used in the proof.
\end{proof}

\begin{prop} \label{uplsion Q surjective}
	Let $C$ be a compact Riemann surface, let $\Omega_C^{1/2}$ be a spin structure on $C$, let $p,q\in C$ be distinct, and let $n\geq 0$. Consider the analytic map below. 
	\begin{align*}
		\Upsilon_{q,C}\co H^0\big(C,\Omega_C^{1/2}(np-q)\big)\to H^1_{dR}(C\setminus p) && \sigma\mapsto [\sigma^2]_{dR}
	\end{align*}
	If $n>6g$, then $\dim \Upsilon_{q,C}^{-1}(0)=n-2g-1$.
\end{prop}

\begin{proof}

	Let $\sigma\in \Upsilon_{q,C}^{-1}(0)$ be nontrivial. Let the zero divisor of $\sigma$ be $mp+rq+E$, where $n\geq m\geq 0$, $r\geq 1$, $E\geq 0$, $[E]\cap p=\emptyset$, and $[E]\cap q=\emptyset$.

	We use the technical lemma \cref{technical} above. 	
	Set $A=(2n-m)p$ and $B=rq+E$. 
	\begin{equation*}
		\begin{tikzcd}[row sep=small]
			H^0\big(C,\Omega_C^{1/2}(np-q)\big)\arrow{r}& H^0\big(C,\Omega^1_C((2n-m)p-2q-E)\big)\arrow{r}& H^1_{dR}(C\setminus p)\\
			\upsilon \arrow[mapsto]{r}& \sigma\upsilon \arrow[mapsto]{r}& \left[\sigma\upsilon\right]_{dR}
		\end{tikzcd}
	\end{equation*}
	The first map above is an isomorphism since it is multiplication by $\sigma$. The second map above is surjective by \cref{technical} when 
	\begin{enumerate}
		\item[1)] $m\leq 2n$ and
		\item[2)] $\deg(E)+\deg[E]\leq 2n-m-r-2g$.
	\end{enumerate}
	We notice that the composite map is $\upsilon\mapsto [\sigma\upsilon]_{dR}$, which is $d\Upsilon_{q,C,\sigma}$ up to a constant factor.

	Let $t\in C$ be a generic point, distinct from $p$ and $q$. Let $s\geq 1$. 	Set 
	\begin{align*}
		\PP_{\Upsilon^{-1}(0)}\coeq \left\{ (\sigma)\in\PP H^0(C,\Omega_C^{1/2}(np-q)) \text{ such that } \sigma\in \Upsilon_{q,C}^{-1}(0) \right\}, \\
		\PP_{st}\coeq \left\{ (\sigma)\in\PP H^0(C,\Omega_C^{1/2}(np-q)) \text{ such that } \div(\sigma)\geq st \right\}.
	\end{align*}
	Notice that $\dim \PP H^0(C,\Omega_C^{1/2}(np-q))=n-2$, $\dim \PP_{\Upsilon^{-1}(0)} \geq n-2g-2$, and $\dim \PP_{st}=n-s-2$. Therefore, if $n-2g-2\geq s$, then $\PP_{\Upsilon^{-1}(0)} \cap \PP_{st} \neq\emptyset$.

	By the last paragraph, we there exists $\sigma\in\Upsilon_{q,C}^{-1}(0)$ nontrivial with $mp+rq+E\geq (n-2g-2)t$. This implies that $n-2g-2 \leq m+r+\deg E \leq g-1+n$.
	
	We see that $r+\deg E-(n-2g-2)\geq m$. Then since we have assumed $n>6g$, we have that $2n>3g+1=(g-1+n)-(n-2g-2)\geq m$ which is condition 1) above. 
	
	Further, since we have assumed $n>6g$, we have that $n-3g+1\geq 3g+1\geq (r+m) -(n-2g-2)\geq \deg E\geq \deg [E]$. This implies that $(g-1+n)+\deg [E]\leq 2n-2g$, which implies condition 2) above. 	
\end{proof}

We now return to the setting where $\pi\co X\to S$ is a family of closed SRSs, $P$ is a relative divisor representing NS punctures, and $(x|\xi)$ are relative superconformal coordinates on the trivial family $\AA^{1|1}\times S\to S$. Also assume that $U$ is a small enough open supersubspace of $S$. Let $Y\to U$ be the restriction of the family of open SRSs $X\setminus P\to S$ to the base $U$.

We next prove three claims before combining the results in the proof of \cref{superconformalNoether} at the end of this section.

\begin{clm}\label{claim A} Consider $\pi\co X\to S$ a family of closed SRSs with $P$ a relative divisor representing NS punctures. 
	
	If $n>2g$, then for every $s\in S$ there exists a small enough open $U\ni s$ such that there exist nontrivial $(f|\phi)\in\Gamma(\pi^{-1}(U),\mathcal{O}_X(nP))$ such that $\delta f=\phi\,\delta\phi$.
\end{clm}

	\begin{proof}
	
	Since we are interested in odd sections $\phi$ of $\mathcal{O}_X(nP)$ which are relative de Rham exact, we consider $H^0\big(\pi^{-1}(U),\Pi\mathcal{O}_X(nP)\big)$, where $\Pi$ is the parity change operator, as a super vector bundle over $U$. For high enough $n$, this super vector bundle is locally free. 
	The relative de Rham cohomology $H^1_{dR/U}(Y)$ results from the relative superconformal differential $\delta$. Considered over $U$, this is a super vector bundle since the de Rham cohomology is topological. 
	
	Consider the map between the relative affine superspaces over $U$ 
	\begin{align*}
		\Upsilon\co \mathbb{H}^0\big(\pi^{-1}(U),\Pi\mathcal{O}_X(nP)\big)\to \mathbb{H}^1_{dR/U}(Y) && \sigma\mapsto [\sigma\,\delta\sigma]_{dR}.
	\end{align*}
	We wish to show that there exists an everywhere nonzero section $\phi$ in $\Upsilon^{-1}(0)$. 
	
	By super Riemann-Roch, we have that $\dim_U \mathbb{H}^0\big(\pi^{-1}(U),\Pi\mathcal{O}_X(nP)\big)=n\,|\,n-g+1$ since $n>2g-1$. Further by de Rham theory, $\dim_U \mathbb{H}^1_{dR/U}(Y)=2g\,|\,0$, since the cohomology classes are the same as for the underlying Riemann surface.
	Therefore, $\Upsilon$ is a map between superspaces of relative dimension $n\,|\,n+g-1$ and $2g\,|\,0$ over $U$.

	Since the domain has higher even dimension than the even dimension of the codomain, we see that there must exist a nontrivial section $\phi\in H^0\big(\pi^{-1}(U),\Pi\mathcal{O}_X(nP)\big)$ such that the de Rham class $[\phi\,\delta\phi]$ is $0$. We have assumed that $U$ is small enough, so we are allowed to shrink $U$ more if needed to make $\phi$ everywhere nonzero over $U$.
\end{proof}

\begin{clm}\label{claim B}
	Consider $\pi\co X\to S$ a family of closed SRSs with $P$ a relative divisor representing NS punctures. 
	
	If $n>6g-2$, then for every $s\in S$ there exists a small enough open $U\ni s$ such that the superspace of odd $\phi\in\Gamma(\pi^{-1}(U),\mathcal{O}_X(nP))$ such that there exists an $f$ so that $\delta f=\phi\,\delta \phi$ has dimension  $n-2g\,|\,n-g+1$.
\end{clm}

\begin{proof}

	Consider the differential map between relative tangent bundles
	\begin{equation*}
		\begin{tikzcd}[row sep=small]
			d\Upsilon_\sigma\co\Big(H^0\big(\pi^{-1}(U),\Pi\mathcal{O}_X(nP)\big)\Big)\Big|_u\arrow{r}&  \Big(H^1_{dR/U}(Y)\Big)\Big|_u\\
			\upsilon \arrow[mapsto]{r}&  \left[\upsilon\,\delta\sigma+\sigma\,\delta\upsilon\right]_{dR}
		\end{tikzcd}
	\end{equation*}
	where $\sigma\in\Upsilon^{-1}(0)$ is a point, with $\pi(\sigma)=u\in U$.

	We wish to show that $d\Upsilon_\sigma$ is surjective for some nontrivial point $\sigma\in\Upsilon^{-1}(0)$ in each fiber over $U$, and therefore by the implicit function theorem, $\Upsilon^{-1}(0)$ has relative dimension over $U$ the difference $(n\,|\,n-g+1)-(2g\,|\,0)$.

	We recall that a family of SRSs have as fiber over each point of the base a split SRS, which is known to be exactly a spin curve. For $\Sigma$ a split SRS, $\Pi\mathcal{O}_\Sigma\cong \Pi\mathcal{J}_{\Sigma_{\red}}\oplus \Pi\mathcal{O}_{\Sigma_{\red}}$, where $\Pi\mathcal{J}_{\Sigma_{\red}}$ is the spin structure on $\Sigma_{\red}$. So the domain of $d\Upsilon_\sigma$ has decomposition as a supervector space as $H^0\big(\Sigma,\Pi\mathcal{J}_{\Sigma_{\red}}(np)\big)\oplus H^0\big(\Sigma,\Pi\mathcal{O}_{\Sigma_{\red}}(np)\big)$.

	Since the codomain of $d\Upsilon_\sigma$ is of relative dimension $(2g\,|\,0)$, then the surjectivity of $d\Upsilon_\sigma$ just means the surjectivity of the even-even block of $H^0\big(\Sigma,\Pi\mathcal{J}_{\Sigma_{\red}}(np)\big)\xrightarrow{d\Upsilon_\sigma} H^1_{dR}\big(\Sigma\setminus p\big)$.

	The situation of a single spin curve was studied by \textcite{Artebani-Pirola--2005}. 
	We restate their result in \cref{uplsion surjective}. This proposition implies the even dimension of $\Upsilon^{-1}(0)$ is $n-2g$. 
\end{proof}

\begin{clm}\label{claim C} Consider $\pi\co X\to S$ a family of closed SRSs with $P$ a relative divisor representing NS punctures. Let $R$ be the zero divisor of $\delta\phi$, for a $\phi$ such that there exists $f$ where $\delta f=\phi\,\delta \phi$. Let $Q$ be a relative irreducible divisor (of degree $1$) contained in $R$.
	
	If $n>6g$, then for every $s\in S$ there exists a small enough open $U\ni s$ such that the superspace of odd $\phi\in \Gamma(\pi^{-1}(U),\mathcal{O}_X(nP-Q))$ such that there exists an $f$ so that $\delta f=\phi\,\delta \phi$ has dimension $n-2g-1\,|\,n-g+1$.
\end{clm}

\begin{proof}
	We assume that $Q$ and $P$ are disjoint without loss of generality.
				
	Consider the slightly modified map
	\begin{align*}
			\Upsilon_Q\co \mathbb{H}^0\big(\pi^{-1}(U),\Pi\mathcal{O}_X(nP-Q)\big)\to \mathbb{H}^1_{dR/U}(Y) && \sigma\mapsto [\sigma\,\delta\sigma]_{dR}.
	\end{align*} 
	
	Consider the differential map between relative tangent bundles
	\begin{equation*}
		\begin{tikzcd}[row sep=small]
			d\Upsilon_{Q,\sigma}\co\Big(H^0\big(\pi^{-1}(U),\Pi\mathcal{O}_X(nP-Q)\big)\Big)\Big|_u\arrow{r}&  \Big(H^1_{dR/U}(Y)\Big)\Big|_u\\
			\upsilon \arrow[mapsto]{r}&  \left[\upsilon\,\delta\sigma+\sigma\,\delta\upsilon\right]_{dR}
		\end{tikzcd}
	\end{equation*}
	where $\sigma\in\Upsilon^{-1}(0)$ is a point, with $\pi(\sigma)=u\in U$.
	
	As in the proof of \cref{claim B}, we wish to show that $d\Upsilon_{Q,\sigma}$ is surjective for some nontrivial point $\sigma\in \Upsilon^{-1}(0)$ in each fiber over $U$. 
	As in the proof of \cref{claim B}, the result follows from \cref{uplsion Q surjective} since we have assumed $n>6g$.
\end{proof}

Lastly we combine the above three claims to show the below lemma which implies \cref{superconformalNoether}.

	\begin{lem}\label{Noether lemma lemma} Consider $\pi\co X\to S$ a family of closed SRSs with $P$ a relative divisor representing NS punctures. 	
		Let $Y\to U$ be the restriction of the family of open SRSs $X\setminus P\to S$ to the base $U$.
		
		For every $s\in S$ there exists a small enough open $U\ni s$ such that there exist global functions $f|\phi\in\Gamma(Y,\mathcal{O}_{X\setminus P})$ where $f$ is even and $\phi$ odd such that 
	\begin{align}\label{delta condition}
		\delta f=\phi\,\delta \phi
	\end{align}
	and the relative zero divisor of $\delta\phi$ is finite over $U$, where $\delta\co\mathcal{O}_{X}\to \omega_{X/S}$ is the relative superconformal differential in \cref{delta diagram}.
	
	Further, there exist global functions $f',\phi'\in\Gamma(Y,\mathcal{O}_{X\setminus P})$ as above, with the relative zero divisor of $\phi'$ over $U$ disjoint from the relative zero divisor of $\phi$ over $U$. 
	
\end{lem}

	\begin{proof}	
	Assume $n>6g$.

	\Cref{claim A} implies the existence of $(f|\phi)$ corresponding to a superconformal Noether normalization. Since $\phi$ was constructed on a compact SRS, then $(f|\phi)$ define a finite surjective morphism with finite relative zero divisor. 
	
	\Cref{claim B} and \cref{claim C} then show the existence of a $(f'|\phi')$ with disjoint divisor. 
	In particular, \cref{claim C} shows that the space of $\phi'\in \Gamma(\pi^{-1}(U),\mathcal{O}_X(nP))$ which have relative zero divisors which intersect the relative zero divisor of $\delta \phi$ (denoted by $R$ previously) is of dimension $n-2g-1$. Therefore, there exists an odd $\phi'\in \Gamma(\pi^{-1}(U),\mathcal{O}_X(nP))$ with relative zero divisor disjoint from $R$. 
\end{proof}

\begin{proof}[Proof of \cref{superconformalNoether}]

	We have reduced the proof to showing that \cref{Noether lemma lemma} implies the result.

	Define the map $p$ to be the morphism over $U$ corresponding to the map of functions $\CC[x|\xi]\to \Gamma(Y,\mathcal{O}_{X\setminus P})$ given by $x\mapsto f$ and $\xi \mapsto \phi$.

	Shrinking $U$ if needed, consider local relative superconformal coordinates $(z|\zeta)$ on $Y$. In this chart we have that
	\begin{align*}
		\phi=\gamma+c\zeta
	\end{align*} 
	for some even|odd functions respectively $c|\gamma$, functions of $z$ and on $U$. Since \cref{delta condition} gives that $D_\zeta f=\phi \,D_\zeta\phi$ (see \cref{delta superconformal}), then 
	\begin{align}\label{f phi vector field}
		\pp{\phi}+\phi\pp{f}\cong\left(\frac{1}{c-\gamma'\zeta}\right)\left(\pp{\zeta}+\zeta\pp{z}\right).
	\end{align}
We see that the vector field on $Y$ in \cref{f phi vector field} has poles exactly where $D_\zeta\phi=c-\gamma'\zeta=0$, which is the local coordinate expression for $\delta\phi$ by \cref{delta in local coor}. 
	
	Further, the Berezinian of the Jacobian matrix is $c-\gamma'\zeta$.
	And since $f=p^*(x)$ and $\phi=p^*(\xi)$, the ramification of $p$ is given by the zeros of $c-\gamma'\zeta$.
	
	This shows that $p^*\left(\pp{\xi}+\xi\pp{x}\right)$, which is given in local relative coordinates in \cref{f phi vector field}, is a rational vector field which generates $\mathcal{D}_{X/S}\big|_{Y}(R)$, where $R$ is the relative divisor over $U$ given by $c-\gamma'\zeta=0$. Thus, if the relative zero divisor of $\delta\phi$ is disjoint from the relative zero divisor of $\delta\phi'$ over $U$, then the ramification of $p$ is also disjoint from the ramification of $p'$.

	\end{proof}

\section{Proof of \cref{CompactDiscreteDef} (a condition to be a discrete submodule)}\label{alt discrete def proof}

\Cref{CompactDiscreteDef} provides an alternative characterization of a discrete subspace.
We first provide a few lemmas.

\begin{lem}\label{containing compact}
	Consider a discrete $\mathcal{O}_U$-submodule $L$, such that $L_U\oplus \hat{K}_U\to(\hat{H}_j)_U$ is an isomorphism.
		For every compact $K'\supseteq K$, then $L_U\cap \hat{K}'_U$ is locally free of finite type and $(\hat{H}_j)_U/(L_U+\hat{K}'_U)=0$.
\end{lem}
\begin{proof}
	We have the exact sequence
	\begin{align*}
		0 \to L_U\cap \hat{K}_U\to L_U\cap \hat{K}'_U\to \hat{K}'_U/\hat{K}_U\to (\hat{H}_j)_U/(L_U+ \hat{K}_U)\to (\hat{H}_j)_U/(L_U+\hat{K}'_U)\to0
	\end{align*}
	which simplifies to $L_U\cap \hat{K}'_U\cong \hat{K}'_U/\hat{K}_U$. 
	 Further, we see that $\hat{K}'_U/\hat{K}_U\cong \hat{(K'/K)}_U$ which is locally free of rank $\dim(K'/K)<\infty$.
\end{proof}

\begin{lem}\label{quotient zero locally free kernel}
	Consider a discrete $\mathcal{O}_S$-submodule $L$.
		If $(\hat{H}_j)_S/(L_S+\hat{K}_S)=0$, then $L_S\cap \hat{K}_S$ is locally free of finite type.
\end{lem}
\begin{proof}
	By \cref{containing compact}, we see that locally we can find $K'\supseteq K$, so that $L_U\cap \hat{K}'_U$ is free of finite type and $(\hat{H}_j)_U/(L_U+\hat{K}'_U)=0$.
	Then we have the exact sequence
	\begin{align*}
		0 \to L_U\cap \hat{K}_U\to L_U\cap \hat{K}'_U\to \hat{K}'_U/\hat{K}_U\to 0.
	\end{align*}
	Since $L_U\cap \hat{K}_U$ is the kernel of a surjective map between free modules of finite type, then the result follows.
\end{proof}

\begin{lem}\label{quotient finite pres}
	Let $L$ be a super $\mathcal{O}_S$-submodule which for every $s\in S$ there exists a neighborhood $U$ of $s$ and a compact $K$ such that $L_U\cap \hat{K}_U$ is free of finite type and $(\hat{H}_j)_U/(L_U+\hat{K}_U)=0$.
	
	Then for every compact $K'$, $(\hat{H}_j)_S/(L_S+\hat{K}'_S)$ is locally of finite presentation.
\end{lem}
\begin{proof}
	Choose compact $K''\supseteq (K'\cup K)$. Then by \cref{containing compact}, $L_U\cap \hat{K}''_U$ is free of finite type and $(\hat{H}_j)_U/(L_U+\hat{K}''_U)=0$ for some $U\ni s$. Then we have the exact sequence
	\begin{align*}
		L_U\cap \hat{K}''_U\to (K''/K')_U\to (\hat{H}_j)_U/(L_U+\hat{K}'_U)\to 0
	\end{align*}
	from which the result follows.
\end{proof}

\begin{proof}[Proof of \cref{CompactDiscreteDef}]
	The definition of a discrete module clearly satisfies the condition of \cref{CompactDiscreteDef} automatically.
	
	For the other direction, assume $L$ is a super $\mathcal{O}_S$-submodule which for every $s\in S$ there exists a neighborhood $U$ of $s$ and a compact $K$ such that $L_U\cap \hat{K}_U$ is free of finite type and $(\hat{H}_j)_U/(L_U+\hat{K}_U)=0$.
	
	Then there exists $K'\subseteq K$ such that $L_U\cap \hat{K}'_U=0$. Then since $\hat{H}_s/(L_s+\hat{K}'_s)$ is a finite dimensional super vector space, we can find $K''\supseteq K'$ such that the natural map $(K''/K')_s\to (\hat{H}_j)_s/(L_s+\hat{K}'_s)$ is an isomorphism.

	Then since $(\hat{H}_j)_S/(L+\hat{K}_S)$ is locally of finite presentation by \cref{quotient finite pres} and $(\hat{H}_j)_s/(L_s+\hat{K}''_s)=0$, we may apply the super Nakayama lemma \cite[Lemma 4.7.1 ii]{Varadarajan--2004} to conclude that there exists a neighborhood $U'\ni s$ such that $(\hat{H}_j)_{U'}/(L_{U'}+\hat{K}_{U'})=0$. Similarly, since $L\cap\hat{K}''_S$ is locally free of finite type by \cref{quotient zero locally free kernel} and $L_s\cap \hat{K}''_s=0$, there exists $U''\ni s$ such that $L_{U''}\cap \hat{K''}_{U''}=0$.	
	
	Therefore $U'\cap U''\ni s$ and $K''$ exhibit $L$ as a discrete module.
\end{proof}

\section*{Acknowledgments}
This paper began as a collaboration with D. Diroff and I would like to thank him for his  contributions to the first results of this project. The advice and encouragement from A. Voronov has been invaluable. I thank the anonymous referee for suggestions on minor improvements throughout the paper, and especially for suggesting major revision of \cref{GrDef} and the proof of \cref{commutant}.

Part of this work was completed while at the Max Planck Institute for Mathematics in Bonn, and the author acknowledges their funding.

\printbibliography

\clearpage\end{CJK}

\end{document}